\documentclass[11pt, 3p,authoryear,a4paper,table]{elsarticle}

\usepackage{titlesec}

\titlespacing*{\section}{0pt}{1.0ex}{0.6ex}
\titlespacing*{\subsection}{0pt}{0.8ex}{0.4ex}

\usepackage{graphicx}
\usepackage{tikz}
\usepackage{todonotes}

\graphicspath{{figures/}}

\usepackage{slashbox}
\usepackage{setspace}
\usepackage{dsfont}
\usepackage{amsmath}
\usepackage{amssymb}
\usepackage{amsfonts}
\usepackage{amsbsy}
\usepackage{subcaption}
\usepackage{amsthm}
\usepackage{bbm}
\usepackage{empheq}
\usepackage{array}
\usepackage{booktabs} 
\usepackage{verbatim}
\usepackage[english]{babel}
\usepackage{natbib}
\usepackage{ulem}
\usepackage{url}
\usepackage{multirow}
\usepackage{float}
\usepackage{numprint}
\usepackage{enumitem}
\usepackage[title]{appendix}
\usepackage[font=normalsize, labelfont=bf, justification=centering]{caption}
\usepackage{tabularx}
\usepackage{ltablex}
\usepackage{makecell}
\usepackage[hidelinks]{hyperref}
\usepackage{xcolor}

\numberwithin{table}{section}
\numberwithin{figure}{section}

\newcolumntype{R}{>{\raggedleft\arraybackslash}X} 
\newcolumntype{P}[1]{>{\centering\arraybackslash}p{#1}} 

\renewcommand\appendix{\par
\setcounter{section}{0}
\setcounter{subsection}{0}
\setcounter{table}{0}
\setcounter{figure}{0}
\gdef\thetable{\Alph{table}}
\gdef\thefigure{\Alph{figure}}
\gdef\thesection{\Alph{section}}
\setcounter{section}{0}}

\newtheorem{definition}{Definition}[section]
\newtheorem{theorem}{Theorem}[section]
\newtheorem{lemma}{Lemma}[section]

\newtheorem{proposition}{Proposition}[section]
\newtheorem{remark}{Remark}[section]
\numberwithin{equation}{section}

\newcounter{arclist}

\newcounter{arcenum}

\begin{document}


\begin{frontmatter}

\title{Balancing Shareholder Value and Financial Stability \\
       under a Reduced-Form Liquidation Model}

\cortext[cor]{Corresponding author. }
\address[UMelb]{Centre for Actuarial Studies, Department of Economics, University of Melbourne VIC 3010, Australia}
\address[UNSW]{School of Risk and Actuarial Studies, UNSW  Business School, UNSW Sydney NSW 2052, Australia}

\author[UMelb]{Benjamin Avanzi
\corref{cor}}
\ead{b.avanzi@unimelb.edu.au}

\author[UNSW]{Bernard Wong}
\ead{bernard.wong@unsw.edu.au}

\author[UNSW]{Jinxia Zhu 
}
\ead{jinxia.zhu@unsw.edu.au}

\begin{abstract}


Modern resolution and prudential regimes increasingly wind up a distressed firm not at a single hard threshold but through a graduated, state-dependent process. We study how the design of such a regime shapes the trade-off between shareholder value and financial stability for a firm whose surplus follows a general diffusion. Forced liquidation is modelled in reduced form, arriving at a surplus-dependent hazard rate that rises as the firm's position deteriorates. The framework has three regions: an unregulated region where dividends may be paid, a regulated region where solvency requirements prohibit distributions, and a distress region in which the firm faces the liquidation hazard. To quantify shareholder value we solve the resulting singular stochastic control problem: which is to maximise the expected present value of distributions until liquidation. We establish a verification theorem, prove that a barrier strategy is optimal, and obtain tractable expressions for the value function and the expected survival time, so that alternative designs can be compared at low cost. We show that a distress region placed solely below or solely above the classical ruin threshold does not consistently improve both shareholder value and firm survival, whereas combining the two yields a Pareto improvement. Regulatory design is decisive.

\end{abstract}

\end{frontmatter}

\section{Introduction}\label{introduction}

\subsection{Background and motivation}\label{intro:motivation}

A central tension in corporate governance and financial regulation concerns the trade-off between shareholder value and financial stability. Shareholders prefer payout policies that extract cash from the firm; regulators and creditors favour conservative policies that maintain adequate reserves against adverse shocks. This tension is most acute for firms approaching financial distress, where the institutional design of bankruptcy procedures and the constraints imposed on firms with weakened balance sheets can dramatically alter both the expected lifetime of the firm and the present value of  distributions to shareholders.

In practice, liquidation is rarely an instantaneous event triggered at a single threshold. Most jurisdictions provide mechanisms that allow financially distressed firms to continue operating while the situation evolves. In the United States, Chapter~11 of the Bankruptcy Code permits distressed firms to reorganise under court supervision rather than face immediate liquidation under Chapter~7~\citep{BroadieChernovSundaresan2007, Bris2006}. In Australia, the so-called ``safe harbour'' provisions of section~588GA of the \textit{Corporations Act 2001} shield directors from personal liability for insolvent trading while they pursue a restructuring plan likely to lead to a better outcome than immediate administration. Similar features arise in prudential regulation: the Solvency~II framework requires insurers to maintain capital above a minimum solvency capital requirement (SCR), and restricts dividend payments when reserves fall below this threshold~\citep{EiopaSolvency2}, while Basel~III imposes a comparable capital conservation buffer on banks with explicit restrictions on distributions inside the buffer~\citep{BaselIII2011}. Firms routinely hold capital well above regulatory minima---not only for
precautionary reasons but specifically to avoid the dividend restrictions, ratings downgrades, and reputational costs associated with breaching these thresholds.

The standard mathematical framework for the study of firm survival has been a \textit{first-passage} formulation: ruin occurs the instant the firm's surplus falls below a fixed lower barrier~\citep[see, e.g.,]{AsAl10}. While analytically tractable, this formulation collapses the institutional landscape described above into a single deterministic threshold. It cannot distinguish a firm operating comfortably above a regulatory floor from one on the brink of failure in a distress zone; it cannot capture the fact that forced liquidation is typically a stochastic event triggered by a confluence of creditor actions, audit findings, regulatory intervention, and exhaustion of legal protections, rather than a single mechanical passage.

In this paper we develop a continuous-time framework that works towards addressing these limitations. The firm's surplus $R_t^D$ evolves according to a one-dimensional diffusion with zone-dependent drift and volatility,
\begin{equation}\label{eq:intro-SDE}
dR_t^D \;=\; \mu(R_{t-}^D)\,dt \;+\; \sigma(R_{t-}^D)\,dW_t \;-\; dD_t,
\end{equation}
where $\{W_t\}$ is a standard Brownian motion and $\{D_t\}$ is the cumulative amount of dividend distributions up to time~$t$. Three critical thresholds $a_0 < a_d \leq a_s$ delineate four distinct
regimes, illustrated in Figure~\ref{F_zones} and described in the paragraphs that follow. Forced liquidation is modelled as a point event whose arrival rate depends on the current surplus level: while $R_t^D$ lies in the distress zone $(a_0, a_d)$, liquidation arrives with \textit{liquidation intensity} $\omega(R_t^D)$, meaning that the conditional probability of liquidation in the next instant
$(t, t+dt)$ given $R_t^D = x$ is $\omega(x)\,dt$. At the lower threshold $a_0$ the firm's remaining cash or regulatory viability is exhausted and liquidation is immediate ($\omega(a_0)\equiv\infty$); above $a_d$ the firm is solvent in the accounting sense and no forced liquidation arises from the intensity mechanism ($\omega(\cdot)\equiv 0$). Note that the condition $\omega(a_0)=\infty$ encodes the absorbing boundary at $a_0$: the firm is liquidated immediately upon reaching $a_0$, so $V(a_0)=0$. In the numerical examples this boundary is imposed directly (an absorbing barrier with $V(a_0)=0$) and $\omega$ is taken finite and non-increasing on the open distress interval $(a_0,a_d)$; the value of $\omega$ at $a_0$ itself is never used.

\definecolor{zoneDistress}{HTML}{F39899}
\definecolor{zoneRegulated}{HTML}{A5C5DF}
\definecolor{zoneUnreg}{HTML}{AFDBAE}
\definecolor{zoneLiquidation}{HTML}{B22222}

\begin{figure}[htb]
\begin{center}
\scalebox{0.8}{
\begin{tikzpicture}
    \draw[->] (1,0) -- (1,6.5) node[left] {Surplus level $x$};

    \draw[dotted] (1,4.5) -- (15,4.5);
    \node[left] at (1,4.5) {$a_s$ (solvency level)};
    \draw[dotted] (1,3) -- (15,3);
    \node[left] at (1,3) {$a_d$ (distress level)};
    \draw[dotted] (1,1.5) -- (15,1.5);
    \node[left] at (1,1.5) {$a_0$ (liquidation level)};

    \fill[zoneUnreg] (1,4.5) rectangle (4,6);
    \node at (2.5,5.25) {unregulated zone};

    \fill[zoneRegulated] (1,3) rectangle (4,4.5);
    \node at (2.5,4) {solvent, but};
    \node at (2.5,3.5) {regulated zone};

    \fill[zoneDistress] (1,1.5) rectangle (4,3);
    \node at (2.5,2.25) {distress zone};

    \fill[zoneLiquidation] (1,0) rectangle (4,1.5);
    \node at (2.5,.95) [text=white] {immediate};
    \node at (2.5,.55) [text=white] {liquidation};

    \node at (2.5,7.2)  {\textbf{Zone}};
    \node at (5.5,7.2)  {\textbf{Liquidation intensity}};
    \node at (9.25,7.2) {\textbf{Surplus dynamics}};
    \node at (13.25,7.2){\textbf{Dividend admissibility}};

    \node at (5.5,5.25) {$\omega(\cdot) = 0$};
    \node at (5.5,3.75) {$\omega(\cdot) = 0$};
    \node at (5.5,2.25) {$0 < \omega(\cdot) < \infty$};
    \node at (5.5,1.5)  {$\omega(a_0) \equiv \infty$};

    \node at (9.25,5.25) {$\mu(\cdot)=\mu_u(\cdot),\ \sigma(\cdot)=\sigma_u(\cdot)$};
    \node at (9.25,4)    {$\mu(\cdot)=\mu_s(\cdot),\ \sigma(\cdot)=\sigma_s(\cdot)$};
    \node at (9.25,3.5)  {(potentially different)};
    \node at (9.25,2.5)  {$\mu(\cdot)=\mu_d(\cdot),\ \sigma(\cdot)=\sigma_d(\cdot)$};
    \node at (9.25,2)    {(potentially different again)};

    \node at (13.25,5.25) {yes};
    \node at (13.25,3.75) [draw, blue] {no};
    \node at (13.25,2.25) [draw, blue] {no};
\end{tikzpicture}
}
\end{center}
\caption{The four regimes of the surplus process, with their associated liquidation intensity, surplus dynamics, and dividend admissibility. The drift and volatility functions $(\mu_d,\sigma_d)$, $(\mu_s,\sigma_s)$ and $(\mu_u,\sigma_u)$ need not coincide, but are assumed to satisfy standard regularity conditions within each zone.}
\label{F_zones}
\end{figure}
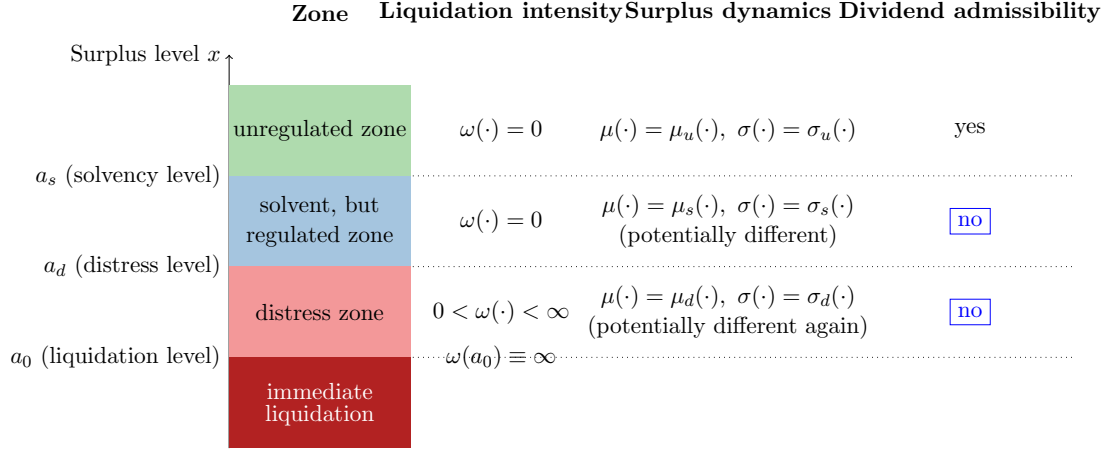

The four regimes in Figure~\ref{F_zones} each correspond to a distinct institutional situation.

\noindent \textbf{Unregulated zone $[a_s, \infty)$}
The firm operates freely: it may distribute dividends to shareholders and makes cash-management decisions to maximise shareholder value. Regulatory oversight is routine and non-intrusive. In insurance, this corresponds to holding capital comfortably above the Solvency~II SCR; in banking, to operating with Common Equity Tier~1 capital above regulatory minima and above the conservation buffer.

\noindent \textbf{Regulated zone $[a_d, a_s)$}
The firm remains solvent and continues to operate, but its surplus has fallen below the level at which shareholder distributions are permitted. Regulatory constraints prohibit dividend payments until capital is restored above $a_s$. This zone captures binding solvency capital requirements: the Solvency~II SCR
in insurance~\citep{EiopaSolvency2}, the Basel~III capital conservation buffer in banking~\citep{BaselIII2011}, or dividend restrictions in debt covenants.

\noindent \textbf{Distress zone $(a_0, a_d)$.}
The firm is technically insolvent in the accounting sense. Its surplus has fallen below the level at which directors must form a view on the firm's viability, but continues to operate while the situation evolves. At any moment, forced liquidation may be triggered by one of several mechanisms: a creditor petition, a director's realisation of insolvency, an adverse audit finding, regulatory action, or the inability to meet an immediate cash obligation. In Australia, for example, section~588G of the \textit{Corporations Act 2001} makes directors personally liable for debts incurred while trading
insolvently, unless the safe harbour provisions of section~588GA apply; these protect directors developing a restructuring plan subject to specified conditions~\citep{Tiba2019}. The liquidation intensity $\omega(x)$ aggregates the rate at which any of these mechanisms triggers forced liquidation, conditional on the current surplus level $x$. Empirically, this rate is higher when surplus is deeper in distress, which we encode by assuming $\omega$ is non-increasing on $(a_0, a_d)$. This distress zone is also the natural home for institutional mechanisms that allow firms to continue operating through temporary financial difficulty. In the United States, a firm in this zone may file for Chapter~11
reorganisation, continuing to operate under court supervision while restructuring its obligations~\citep{BroadieChernovSundaresan2007, Bris2006}. Our framework is agnostic to which mechanism applies: $\omega(x)$ captures the net effect of all paths from technical insolvency to formal winding-up.

\noindent \textbf{Liquidation level $a_0$.}
The firm has exhausted its capacity to continue operating. This corresponds to running out of cash to meet unavoidable obligations, breaching a hard regulatory threshold (for example, negative net assets
triggering mandatory winding-up), or a definitive court ruling ending any safe harbour. Liquidation is immediate and irreversible.

The specification~\eqref{eq:intro-SDE} together with the liquidation intensity $\omega$ places our framework within the tradition of \textit{reduced-form} models of default that is standard in credit risk
modelling~\citep{JarrowTurnbull1995, DuffieSingleton1999, Lando1998}. Reduced-form models treat default as an exogenous Cox-process event whose intensity is a specified function of state variables---in contrast to structural models, in which default is determined endogenously by firm-value crossings of a barrier~\citep{Merton1974, BlackCox1976}. Whereas structural models require detailed assumptions about the firm's asset and liability dynamics, reduced-form models remain agnostic about the microfoundations of default and focus on the statistical regularity that forced liquidation becomes more likely as financial health deteriorates. Our intensity function $\omega$ plays precisely this role. We emphasise that the specification is \textit{descriptive}, not normative: we do not claim that firms \textit{should} be randomly liquidated in the distress zone, but rather model the fact that they \textit{are}, at rates that increase with the severity of distress. Correspondingly, and despite the local memorylessness of the Cox-process increment, the time to forced liquidation is generally \textit{not} exponentially distributed: it is the first jump of a doubly stochastic Poisson process whose intensity inherits the full state-dependent structure of the controlled surplus process $R_t$.

The mathematical structure we adopt generalises the ruin component of the so-called ``Gamma-Omega model'' introduced in the actuarial literature by \citet{AlbrecherGerberShiu2011}. That literature has established a number of useful technical results, but has typically restricted attention to stationary surplus processes (drifted
Brownian motion, compound Poisson, spectrally negative L\'evy) and simplified intensity specifications (constant or piecewise constant $\omega$). Our framework allows for general diffusion dynamics, zone-dependent drift and volatility, and a general intensity function, and so substantially extends the existing work while placing it within a much broader modelling tradition. It also considers optimal distributions without constraint, that is, without assuming that dividends can only be paid at exponentially distributed intervals of time (the ``gamma'' component of the gamma-omega model).

We use our framework to address two interrelated questions.
\begin{itemize}
\item From the perspective of a regulator or social planner: \textit{when a firm becomes financially distressed, does it improve financial stability to (i) resolve it gradually rather than wind it up the moment it breaches insolvency (that is, by allowing it to keep operating below that point while it restructures, as, for instance, under US Chapter~11 reorganisation or the safe-harbour provisions of the Australian \textit{Corporations Act 2001}), (ii) by restricting shareholder distributions while its reserves are low but still adequate (as under the Solvency~II SCR or the Basel~III capital conservation buffer), or (iii) by combining the two? And if so, under what configurations?}
\item From the perspective of shareholders: \textit{how does the presence and calibration of this zone affect the value available to shareholders under the firm's chosen cash-management strategy?}
\end{itemize}

To quantify the shareholder side of this trade-off, we adopt the Gordon model \citep{Gor62}, which corresponds to the standard valuation approach introduced by \citet{deF57}: the shareholder value of a firm at current surplus $x$ is defined as the maximal expected discounted value of dividend distributions achievable from $x$ under any admissible cash-management strategy. We emphasise that this optimal-dividend formulation is used here as a \textit{descriptive valuation device}, not as a normative prescription for payout policy: it characterises the value that, in principle, shareholders could extract from the firm, which is what we compare across regulatory designs. We pair this shareholder-value criterion with the expected survival time
under the optimal strategy as our measure of financial stability, and compare the results against those of the traditional first-passage benchmark.

\subsection{Related literature}\label{intro:lit}

Our work connects to three strands of the literature. We first place our contributions within the literature on financial distress, bankruptcy procedures, and regulation (first strand). This is achieved by contributions to the dividend optimisation literature (second strand), which we borrow to determine shareholder value, as well as the reduced-form liquidation model literature (third strand), whose approach is used to model solvency.  

 \textbf{Financial distress, bankruptcy procedures, and regulation.} 
The corporate finance literature has extensively studied the interaction between capital structure, financial distress, and bankruptcy procedures. Important contributions on capital structure and bankruptcy design include \citet{Leland1994}, \citet{LelandToft1996}, and \citet{GoldsteinJuLeland2001}. \citet{BroadieChernovSundaresan2007} analyse optimal capital structure in the presence of both Chapter~7 and Chapter~11; \citet{Bris2006} provide empirical
evidence on the costs and incidence of the two procedures; \citet{LiLiuTangZhu2020} study liquidation risk under contemporary regulatory frameworks. Our framework complements this literature: rather than modelling the capital structure decision, we take the firm's surplus dynamics as given and ask how the \textit{design} of the distress zone and associated regulatory constraints affects the trade-off between shareholder value (next strand) and financial stability (third strand).

 \textbf{Dividend optimisation to determine shareholder value.} The expected present value of distributions under optimal cash management has been a standard measure of shareholder value in corporate finance \citep[e.g.,][]{Gor62} and insurance mathematics since~\citet{deF57}; further foundational work includes~\citet{Mor66} and \citet{Ger72}, and comprehensive surveys are given by~\citet{Ava09,AvTuWo16c} and \citet{AlTh09}. Under the traditional first-passage definition of ruin, barrier strategies have been shown to be optimal in many settings, including general diffusions \citep{Paulsen2007, ZhuChen2013IME}, models with fixed transaction costs, and models with solvency constraints~\citep{Pau03, ZhuChen2015Econ,AvChHeWo23}. The three-zone structure adopted in the present paper generalises that of \citet{AvWo12}, which introduced a solvency-constrained band in a mean-reverting dividend model without optimality results. Relative to that work, we allow arbitrary zone boundaries, general diffusion dynamics, and a substantially broader class of liquidation mechanisms. Our framework also differs in a subtle behavioural respect: in~\citet{AvWo12} dividends are permitted in the middle band provided the surplus entered it from above, whereas in our specification the middle band prohibits distributions uniformly, reflecting the typical design of solvency capital requirements.

\textbf{Reduced-form liquidation models and the so-called `gamma-omega' literature.}
Reduced-form models of bankruptcy (in which default or liquidation is a Cox-process event with state-dependent intensity) originate in credit risk \citep{JarrowTurnbull1995, DuffieSingleton1999, Lando1998}, in contrast to structural models in the tradition of~\citet{Merton1974} and \citet{BlackCox1976}. Within actuarial science, the analogous framework was introduced by~\citet{AlbrecherGerberShiu2011} under the name ``Gamma-Omega Model'', and subsequent work has studied ruin probabilities and related quantities for specific stationary processes: compound Poisson \citep{AlbrecherLautscham2013, GerberShiuYang2012}, drifted Brownian motion \citep{AlbrecherGerberShiu2011}, and spectrally negative L\'evy processes \citep{LoeffenRenaudZhou2014, AlbrecherIvanovs2014, Renaud2014, WangWangZhang2016, Glau2016}. Most of these papers impose simplified functional forms for the intensity (typically constant or piecewise constant). Within this framework, dividend questions have been considered under pre-specified (non-optimal) barrier strategies by~\citet{LiuLiu2014} and \citet{CzarnaKaszubowskiPalmowski2020}. Full characterisation of the optimal strategy has been addressed only for drifted Brownian motion~\citep{AlbrecherGerberShiu2011, JinYin2014} and, very recently, for spectrally negative L\'evy processes with piecewise-constant intensity by~\citet{MetaRenaud2024}. Our paper substantially extends this literature by solving the
shareholder-value problem under general diffusion dynamics with zone-dependent drift and volatility and a general intensity function, without relying on stationarity or specific functional forms. It thereby also closes the methodological gap between the reduced-form credit-risk tradition and its actuarial counterpart.

\subsection{Statement of contributions}\label{intro:contributions}

This paper makes the following main contributions.

\textbf{Regulatory design insight.}
We study how the legal and regulatory treatment of a financially distressed firm shapes the trade-off between shareholder value and financial stability. Real regimes rarely wind a firm up the instant it becomes insolvent: distress instead sets off a graduated, partly discretionary process (e.g., court-supervised reorganisation under US Chapter~11, the safe-harbour protections of the Australian \textit{Corporations Act 2001}, supervisory forbearance) in which forced liquidation arrives at a rate that rises as distress deepens. We represent this process by a state-dependent liquidation intensity and, within it, evaluate two interventions a regime may apply: allowing a distressed firm to keep operating below the point of insolvency rather than liquidating it immediately (a recovery buffer below the classical ruin threshold), and restricting shareholder distributions while the firm's reserves are low but still adequate (as under the Solvency~II SCR or the Basel~III capital conservation buffer; a constraint above that threshold). The design proves consequential: on its own, neither intervention reliably improves both objectives: the recovery buffer raises shareholder value but can shorten the life of a well-capitalised firm, while the distribution restriction lengthens survival only by depressing value. Yet, \textit{combining} them can improve both at once, a potential Pareto improvement over immediate liquidation at insolvency (the traditional benchmark). We identify when the combination is most effective and draw out the implications for the design of distress-resolution and solvency regulation.

\textbf{General stochastic control framework.}
To quantify shareholder value in our framework, we formulate and solve the associated singular stochastic control problem for a general diffusion with zone-dependent drift and volatility and a general liquidation intensity function. We establish a verification theorem, prove that a barrier dividend strategy is optimal, and characterise the optimal barrier. This extends substantially beyond existing work, which has been confined to stationary processes (drifted Brownian motion, compound Poisson, or
spectrally negative L\'evy) and simplified intensity specifications.

\textbf{Computational tractability.}
We provide a constructive, implementable procedure for computing the optimal barrier and the associated shareholder-value function via a small system of second-order linear ODEs. We additionally derive closed-form expressions for the expected survival time under the optimal strategy, enabling direct quantitative comparison of alternative regulatory designs.

\textbf{A new qualitative phenomenon.}
We show that the shareholder-value function is concave in the unregulated region but may exhibit convexity---or a mix of convexity and concavity---in the distress and regulated zones. This non-concavity is a new phenomenon in the dividend optimisation literature and arises from the interaction between the liquidation intensity and the solvency constraint, reflecting the fact that the firm cannot exercise optimal control when surplus is below $a_s$.

\medskip
Our framework is also sufficiently flexible to accommodate extensions to related settings, including models with reorganisation features such as random grace periods, and alternative notions of default such as Parisian ruin with stochastic delay.

\subsection{Paper structure}\label{intro:structure}

The remainder of the paper is organised as follows. Section~\ref{formulation} reformulates the problem mathematically: the surplus dynamics with region-dependent coefficients, the liquidation intensity mechanism, the class
of admissible dividend strategies, the shareholder-value function, and the expected survival time.
Section~\ref{results} solves the associated singular stochastic control problem in the general three-zone framework: we establish the Hamilton--Jacobi--Bellman equations, prove a verification theorem, characterise the optimal barrier strategy, and provide a constructive computational procedure. Section~\ref{sec:expectation} deals with the expected survival time under the optimal strategy identified in Section~\ref{results}. Section~\ref{results-reduced} considers the reduced two-region case in which the solvency threshold coincides with the distress threshold ($a_s = a_d$), a special case of practical interest. Section~\ref{sec:numerical} presents the numerical illustrations: a general three-zone example, a demonstration that shareholder value can fail to be concave below the solvency threshold, and---at the heart of the section---a comparison of the traditional regime of immediate liquidation at insolvency against three distress-resolution designs (a recovery buffer below the insolvency point, a distribution restriction above it, and the two combined), quantifying
their effect on shareholder value and firm survival. A Monte-Carlo study then examines the full distribution of the time to liquidation and of the dividends paid, under common random numbers across designs. 
Section~\ref{conclusion} concludes.
Proofs are collected in the Appendix.

\section{Mathematical formulation}\label{formulation}

We work on a complete filtered probability space $(\Omega,\mathcal{F},\{\mathcal{F}_t\}_{t\ge 0},\mathbb{P})$ satisfying the usual conditions and carrying a standard Brownian motion $\{W_t\}_{t\ge 0}$. The firm's surplus is controlled through a dividend strategy $D=\{D_t\}_{t\ge 0}$, where $D_t$ is the cumulative amount distributed to shareholders up to time $t$. Under $D$, the surplus $R^D=\{R_t^D\}_{t\ge 0}$ follows the region-dependent diffusion
\begin{equation}\label{eq:SDE}
dR_t^D=\mu(R_{t-}^D)\,dt+\sigma(R_{t-}^D)\,dW_t-dD_t,
\qquad
(\mu,\sigma)=
\begin{cases}
(\mu_d,\sigma_d), & a_0<R_{t-}^D<a_d,\\[2pt]
(\mu_s,\sigma_s), & a_d\le R_{t-}^D<a_s,\\[2pt]
(\mu_u,\sigma_u), & R_{t-}^D\ge a_s,
\end{cases}
\end{equation}
where the fixed thresholds $a_0<a_d\le a_s$ are the \textit{liquidation},
\textit{distress}, and \textit{solvency} levels, respectively. Dividends may be paid only in the unregulated region $\{R^D\ge a_s\}$, so that the term $-dD_t$ in \eqref{eq:SDE} is active only there; refer to Definition~\ref{admi} below. The drift functions $\mu_i$ are Lipschitz continuous and the volatility functions $\sigma_i$ are strictly positive and Lipschitz continuous on their respective regions, $i\in\{d,s,u\}$; they may differ across regions, allowing operating conditions, financing costs, and regulatory effects to change with the level of surplus; refer to Figure~\ref{F_zones}.

Forced liquidation is modelled through a state-dependent intensity, as explained in the previous section. While the surplus lies in the distress region $(a_0,a_d)$, liquidation arrives as the first event of a Cox process with intensity $\omega(R_t^D)$. Should the process reach the liquidation level $a_0$, liquidation occurs immediately. We take $\omega$ to be continuous, non-negative and non-increasing on $(a_0,a_d)$, and extend it by
$\omega(x)=0$ for $x\ge a_d,$
so that no intensity-driven liquidation occurs once the firm is solvent. The modelling rationale for this reduced-form mechanism, and its interpretation in terms of bankruptcy and prudential-regulation regimes, are given in Section~\ref{intro:motivation}. The regulated region imposes a single constraint on the control: no dividend may be paid while $R^D<a_s$.

\begin{definition}\label{admi}
A dividend strategy $D=\{D_t\}_{t\ge 0}$ is \textit{admissible} if
(i) $D$ is non-decreasing and $\{\mathcal{F}_t\}_{t\ge 0}$-predictable; (ii) $D$ is right-continuous with left limits (c\`adl\`ag);
(iii) $D_{0-}=0$;
(iv) $0\le D_t-D_{t-}\le R_t^D-a_s$ on $\{R_{t-}^D>a_s\}$; and
(v) $D_t-D_{t-}=0$ and $dD_t=0$ on $\{R_{t-}^D<a_s\}$.
The collection of all admissible strategies is denoted by $\Pi$.
\end{definition}

For any $D\in\Pi$ the controlled surplus $R^D$ is c\`adl\`ag and $\{\mathcal{F}_t\}_{t\ge 0}$-adapted. A lump-sum payment at time $t$ (a jump of $D$) is taken to occur at $t-$, so that $R^D$ remains well defined.

Let $T^D$ denote the liquidation time under $D$. The shareholder-value functional is the expected discounted dividend stream 
\begin{gather}\label{perf}
\mathcal{P}_D(x)=\mathbb{E}_x\bigg[\int_0^{T^D}e^{-\delta t}\,dD_t\bigg],
\end{gather}
where $\delta>0$ is the shareholders' discount rate and $\mathbb{E}_x[\cdot]=\mathbb{E}[\,\cdot\mid R_{0-}=x\,]$. Shareholder value is the supremum of \eqref{perf} over admissible strategies,
\begin{gather}\label{eq:V-def}
V(x)=\sup_{D\in\Pi}\mathcal{P}_D(x),
\end{gather}
and a strategy $D^\ast\in\Pi$ is \textit{optimal} if $V(x)=\mathcal{P}_{D^\ast}(x)$. Financial stability is measured by the expected survival time under the optimal strategy,
\begin{align}\label{suvival-new}
m(x)=\mathbb{E}_x\big[T^{D^\ast}\big].
\end{align}

Throughout, we impose the transversality condition
\begin{align}\label{assumption}
\sup_{x\in[a_s,\infty)}\mu_u^{\prime}(x)<\delta,
\end{align}
which guarantees finiteness of $V$ by ruling out unbounded discounted value generated through indefinite accumulation in the unregulated region. Here $\mu_u^{\prime}$ is the rate at which the drift changes with surplus---not the firm's expected return; under the linear specification $\mu(x)=\mu_0+rx$ used in our examples it equals the constant $r$, so \eqref{assumption} reduces to $r<\delta$. Economically, this is the direct analogue of the Gordon growth condition and of the transversality condition in infinite-horizon optimisation: if the marginal growth rate exceeded $\delta$ at some level, the firm could accumulate surplus indefinitely and generate unbounded discounted value, so no finite optimal dividend strategy would exist. In any internally consistent valuation framework, $\delta$ comprises an expected-growth component and a risk premium, and the assumption is equivalent to requiring a strictly positive risk premium. This is the same justification advanced by \citet{Paulsen2007} in a closely related problem: ``In a financial context this is a natural condition, since values are typically calculated under an equivalent martingale measure.'' We retain this assumption throughout.

For any $D\in\Pi$, let
\begin{align}\label{eq:tau-a0}
\tau_{a_0}^D=\inf\{t\ge 0:R_t^D=a_0\}
\end{align}
be the first hitting time of the liquidation level. Since liquidation is immediate at $a_0$ we have $T^D\le\tau_{a_0}^D$, and hence $\mathcal{P}_D(x)\equiv 0$ and $V(x)\equiv 0$ for $x\le a_0$. The intensity-based
liquidation mechanism admits an equivalent representation that absorbs the random liquidation time into a state-dependent discount factor; it underlies the analysis of Section~\ref{results}.

\begin{proposition}\label{alternative-value}
For any admissible strategy $D\in\Pi$ and $x>a_0$,
\begin{align}
\mathcal{P}_D(x)
&=\mathbb{E}_x\Big[\int_0^{\tau_{a_0}^D}
   e^{-\big(\delta t+\int_0^t\omega(R_u^D)\,du\big)}\,dD_t\Big],
\quad V(x)
=\sup_{D\in\Pi}\mathbb{E}_x\Big[\int_0^{\tau_{a_0}^D}
   e^{-\big(\delta t+\int_0^t\omega(R_u^D)\,du\big)}\,dD_t\Big].
   \label{alt-valueF}
\end{align}
\end{proposition}

\begin{remark}
The classical first-passage model is recovered in the limit $a_d\to a_0$, in which liquidation occurs immediately on reaching $a_0$. The reduced two-region case $a_s=a_d$, in which the solvency and distress thresholds coincide, is treated in Section~\ref{results-reduced}; the formulation also contains Parisian
ruin with stochastic delay as a special case \citep[see, e.g.,][]{DassiosWuJul2008}.
\end{remark}

\section{Optimal shareholder value}\label{results}

We now characterize the optimal dividend policy and the corresponding shareholder value. The equivalent representation \eqref{alt-valueF} reformulates the problem as a singular stochastic control problem with state-dependent discounting, making it amenable to dynamic programming techniques. We therefore begin by deriving the associated dynamic programming equation and the corresponding Hamilton--Jacobi--Bellman (HJB) system. Applying the dynamic programming principle, we obtain that, for any stopping time $\tau$, 
 \begin{align}
 V(x)
   =&\sup_{D\in\Pi}\mathbb{E}_x\left[
\int_0^{\tau_{{a_0}}^D\wedge\tau}
e^{-\left(\delta s+\int_0^s\omega(R_u^D)du\right)}dD_s+e^{-\left(\delta (\tau_{{a_0}}^D\wedge\tau)
+\int_0^{\tau_{{a_0}}^D\wedge\tau}\omega(R_u^D)du\right)}V(R_{\tau_{{a_0}}^D\wedge\tau}^D)\right],\ \ x> {a_0}.
 \end{align}
 Formally, the dynamic programming principle leads to the following Hamilton-Jacobi-Bellman (HJB) system for the value function (if it is  smooth enough):
 \begin{align}
      &\frac{1}{2}\sigma_d^2(x)V^{\prime\prime}(x)
     +\mu_d(x) V^{\prime}(x)-(\delta+\omega(x) ) V(x)=0,\ \ {a_0}<x<{a_d},\\
     &\frac{1}{2}\sigma_s^2(x)V^{\prime\prime}(x)
     +\mu_s(x) V^{\prime}(x)-\delta V(x)=0,\ \ {a_d}\le x<{a_s},\\
     &\max\left\{\frac{1}{2}\sigma^2(x)V^{\prime\prime}(x)
     +\mu(x) V^{\prime}(x)-\delta V(x), 1-V^\prime(x)\right\}= 0,\ \ x\ge  {a_s}.
 \end{align}
The next result provides the verification argument linking classical solutions of the HJB system to the value function.

 \begin{lemma}\label{lemma6.1}
For any function $f$ and any strategy $D$,
\begin{align}
    \mathbb{E}_{x}\Big[e^{-\delta T^D}{f}(R_{T^D}^D); T^D<\tau_{{a_0}}^D\Big]=
{}&\mathbb{E}_x\bigg[\int_0^{\tau_{{a_0}}^D}
    e^{-\delta s}{f}(R_s^D)\omega(R_s^D)
    e^{-\int_0^s\omega(R_t^D)dt}ds\bigg].
\end{align}
\end{lemma}

The following verification theorem establishes sufficient conditions for a candidate solution to coincide with the value function.
\begin{theorem}\label{verification-ext}
If ${f}$ is a differentiable and piecewise twice continuously differentiable solution to
\begin{align}
& \frac{1}{2}\sigma_d^2(x) {f}^{\prime\prime}(x)+\mu_d(x){f}^{\prime}(x)-(\delta+\omega(x)){f}(x)=0, \ {a_0}<x<{a_d},\label{eq1-ext}\\
&\frac{1}{2}\sigma_s^2(x) {f}^{\prime\prime}(x)+\mu_s(x){f}^{\prime}(x)-\delta {f}(x)=0, \ {a_d}<x<{a_s},\label{eq2-ext}\\
   & \max{\bigg\{ \frac{1}{2}\sigma^2(x) {f}^{\prime\prime}(x)+\mu(x){f}^{\prime}(x)-\delta {f}(x),1-{f}^{\prime}(x)\bigg\}} = 0,\ x> {a_s},\label{eq3-ext}\\
    & f(x)= 0,\ x\le {a_0}, \label{eq4-ext}
\end{align}
then ${f}(x)\ge V(x)$ for $x\ge 0$.
\end{theorem}

Barrier strategies 
are known to be optimal in many singular control problems. Motivated by this, 
we will investigate this class of strategies first to see whether it  contains the optimal solution.

\begin{definition} (Barrier Strategy)
For any  $b\ge {a_s}$, let $D^b$ denote the barrier strategy that immediately distributes any surplus above $b$  and reflects the controlled surplus process at the level $b$ thereafter.
\end{definition}
We now characterize the value generated by a barrier strategy. Define \begin{align*}
    V_b(x):=\mathcal{P}_x(D^b).
\end{align*}

The construction of the barrier value function relies on the following fundamental solutions.

\begin{definition}\label{fundamentalsol-ext}

Let $g_1$ denote the unique solution to
\begin{equation}
\frac12\sigma_d^2(x)f''(x)+\mu_d(x)f'(x)-(\delta+\omega(x))f(x)=0,
\quad x>a_0,
\label{1111-1-ext}
\end{equation}
subject to $g_1(a_0)=0$ and $g_1'(a_0)=1$.
Let $g_2$ and $g_3$ denote the unique solutions to
\begin{equation}
\frac12\sigma_s^2(x)f''(x)+\mu_s(x)f'(x)-\delta f(x)=0,
\quad x>a_d,
\label{1111-3-ext}
\end{equation}
subject respectively to
$g_2(a_d)=1$, $g_2'(a_d)=0$, and
$g_3(a_d)=0$, $g_3'(a_d)=1$.
Similarly, let $g_4$ and $g_5$ denote the unique solutions to
\begin{equation}
\frac12\sigma_u^2(x)f''(x)+\mu_u(x)f'(x)-\delta f(x)=0,
\quad x>a_s,
\label{1111-6-ext}
\end{equation}
subject respectively to
$g_4(a_s)=1$, $g_4'(a_s)=0$, and
$g_5(a_s)=0$, $g_5'(a_s)=1$.

\end{definition}


The existence and uniqueness of those solutions follow from standard ODE theory (see, for instance,  Theorem
5.4.2 of \cite{Krylov1996}).

Next, we study two boundary value problems and derive their solution, which will be shown to coincide with the performance function of a barrier strategy in certain situations later. 

\begin{lemma}\label{fb-sol-ext}
For any given $b>{a_s}$, the following boundary value problem,
\begin{align}
    &\frac{1}{2}\sigma_d^2(x) {f}^{\prime\prime}(x)+\mu_d(x){f}^{\prime}(x)-(\delta+\omega(x)) {f}(x)=0,\quad x\in(
{a_0},{a_d}), \label{diff1-ext}\\
&\frac{1}{2}\sigma_s^2(x) {f}^{\prime\prime}(x)+\mu_s(x){f}^{\prime}(x)-\delta {f}(x)=0,\quad x\in(
{a_d},{a_s}), \label{diff2-ext}\\
 &\frac{1}{2}\sigma_u^2(x) {f}^{\prime\prime}(x)+\mu_u(x){f}^{\prime}(x)-\delta {f}(x)=0,\quad x\in({a_s},b), \label{diff3-ext}\\
&f({a_0})=0, \quad f^{\prime}(b)=1, \label{diff4-ext}
\end{align}
admits a unique  continuously differentiable solution $f_b$ on $({a_0},b]$. Moreover, 
\begin{align}
    f_b(x)=\begin{cases}
        c_1(b)  g_1(x), & x\in[{a_0},{a_d}],\\
        c_2(b) g_2(x)+c_3(b) g_3(x), & x\in({a_d},{a_s}],\\
        c_4(b) g_4(x)+c_5(b) g_5(x), & x\in({a_s},b],
         \end{cases}\label{fb-ext}
\end{align}
where
\begin{align}
       c_1(b)&=\frac{1}{(g_1(a_d)g_2(a_s)+g_1^{\prime}(a_d)g_3(a_s))g_4^{\prime}(b)
    +(g_1(a_d)g_2^{\prime}(a_s)+g_1^{\prime}(a_d)g_3^{\prime}(a_s))g_5^{\prime}(b)},
\label{c1b-ext}\\    
    c_2(b)&=g_1({a_d})c_1(b),
    \qquad
    c_3(b)=g_1^{\prime}({a_d})c_1(b),\label{c3b-ext}\\
    c_4(b)&=\left(g_2(a_s)g_1({a_d})+g_3(a_s)g_1^{\prime}({a_d})\right)c_1(b),
    \qquad c_5(b)=\left(g_2^{\prime}(a_s)g_1({a_d})+g_3^{\prime}(a_s)g_1^{\prime}({a_d})\right)c_1(b).\label{c5b-ext}
    \end{align}
   Furthermore, $f_b$ is differentiable  on $({a_0},\infty)$ and twice continuously differentiable on $({a_0},{a_d})$, $({a_d},{a_s})$ and $({a_s}, b)$.
\end{lemma}

\begin{lemma}\label{fxbar-sol-ext}
The following initial value differential equation,
\begin{align}
    &\frac{1}{2}\sigma_d^2(x) {f}^{\prime\prime}(x)+\mu_d(x){f}^{\prime}(x)-(\delta+\omega(x)) {f}(x)=0,\ \ x\in(
{a_0},{a_d}), \label{diff2-0-ext}\\
&\frac{1}{2}\sigma_s^2(x) {f}^{\prime\prime}(x)+\mu_s(x){f}^{\prime}(x)-\delta {f}(x)=0,\quad x\in(
{a_d},{a_s}), \label{diff2-0-ext-1}\\
&f({a_0})=0,\quad f^{\prime}(a_s)=1, \label{diff3-0-ext}
\end{align}
has a unique  continuously differentiable solution on $(0,a_s]$, $f_{{a_s}}$,  which is given by
\begin{align}
    f_{a_s}(x)=\begin{cases}
        c_6({a_s})  g_1(x), & x\in[{a_0},{a_d}],\\
        c_7({a_s}) g_2(x)+c_8({a_s}) g_3(x), & x\in({a_d},{a_s}],       \end{cases}\label{fb-ext-03}
\end{align}
where 
\begin{align}
    c_6({a_s})&=\frac{1}{(g_1(a_d)g_2^\prime(a_s)+g_1^{\prime}(a_d)g_3^\prime(a_s))},
   \quad  c_7({a_s})=\frac{g_1(a_d)}{(g_1(a_d)g_2^\prime(a_s)+g_1^{\prime}(a_d)g_3^\prime(a_s))},\label{c7-ext}\\
     c_8({a_s})&=\frac{g_1^\prime(a_d)}{(g_1(a_d)g_2^\prime(a_s)+g_1^{\prime}(a_d)g_3^\prime(a_s))}.\label{c8-ext}
    \end{align}
Furthermore, $f_{{a_s}}$ is  twice continuously differentiable on $({a_0},{a_d})\cup (a_d, a_s)$.
\end{lemma}

Next, we will show that the solution to the above boundary value problem is identical to the performance function of the barrier strategy $D^b$. Due to the solvency constraint,  we only need to  consider $D^b$ with $b\ge a_s$. So in the following, we will investigate properties for the performance function associated with $D^b$ for $b\ge {a_s}$.

\begin{theorem}\label{Vb-ext}
    For $b>{a_s}$, the value generated by the barrier strategy $D^b$ is given by 
  \begin{align}
    V_b(x):=\mathcal{P}_x(D^b)=\begin{cases}
    0 & x<a_0,\\
        f_b(x)& {a_0} \le x\le b,\\
        f_b(b)+x-b & x> b,
    \end{cases}\label{vb-x+-ext}
\end{align}
where   $f_b$ is defined  as in  Lemma \ref{fb-sol-ext} for $b>a_s$ and as in Lemma \ref{fxbar-sol-ext} for $b={a_s}$.
Moreover, $V_b$ is continuously differentiable on $[{a_0},\infty)$, and twice continuously differentiable on $({a_0},{a_d})$, $({a_d},{a_s})$,  $({a_s},b)$ and $ (b, \infty)$. Furthermore, $V_b^{\prime\prime}(b-)$ is continuous with respect to $b$ for $b>{a_s}$.

\end{theorem}

Since dividends are admissible only above $a_s$, optimisation is restricted to barriers satisfying $b\ge a_s$. 
The optimal barrier is characterized through the smooth-fit condition at the reflection boundary. This motivates the following definition:
\begin{align}\label{b*-1-ext}
b^{\ast}:&=\inf\{b\ge {a_s}:\ c_4(b)g_4(b)+c_5(b)g_5(b)\ge \frac{\mu(b)}{\delta}\}\nonumber\\
&=\inf\left\{b\ge {a_s}:\ \frac{(g_1(a_d)g_2(a_s)+g_1^{\prime}(a_d)g_3(a_s))g_4(b)
    +(g_1(a_d)g_2^{\prime}(a_s)+g_1^{\prime}(a_d)g_3^{\prime}(a_s))g_5(b)}{(g_1(a_d)g_2(a_s)+g_1^{\prime}(a_d)g_3(a_s))g_4^{\prime}(b)
    +(g_1(a_d)g_2^{\prime}(a_s)+g_1^{\prime}(a_d)g_3^{\prime}(a_s))g_5^{\prime}(b)}\ge \frac{\mu(b)}{\delta}\right\}.
\end{align}
 As usual, we define $\inf \emptyset =+\infty$, where $\emptyset$ represents an empty set. It follows that $a_s\le b^{\ast}\le +\infty$. Moreover, by \eqref{fb(b)-ext} and \eqref{Vbbdd-ext} we know that $c_4(b)g_4(b)+c_5(b)g_5(b)\ge \frac{\mu(b)}{\delta}$ is equivalent to $V_b^{\prime\prime}(b-)\ge 0$, and hence,
\begin{align}
b^{\ast}=\inf\{b>a_s: V_b^{\prime\prime}(b-)\ge 0\}.\label{b*-ext}
\end{align}
\begin{remark} \label{11123-1-ext}
Since $V_b''(b-)$ is continuous in $b$ for $b>a_s$, it follows that if $b^\ast>a_s$, then $
V_{b^\ast}''(b^\ast-)=0$,  and if $b^\ast=a_s$, $V_{b^\ast}^{\prime\prime}(b^\ast-)\ge 0$.
\end{remark}

The next result guarantees existence of a finite optimal barrier.

\begin{lemma}\label{finiteness}
It holds that $a_s\le b^{\ast}<+\infty$.
\end{lemma}

We will show later that the strategy $D^{b^{\ast}}$ leads to the optimal performance. To this end, we will first study properties of the performance function associated with the strategy $D^{b^{\ast}}$. We distinguish $b^\ast>a_s$ and $b^\ast=a_s$.
The following structural property plays a central role in the verification argument and establishes concavity of the value function in the controlled region.
\begin{lemma} \label{V-concavity-ext}
If $b^{*}>a_s$, then 
$V_{b^{*}}^{\prime\prime}(x)\le 0$ for $ a_s\le x<b^{*}$,$V_{b^{*}}^{\prime}(x)\ge 1$ for $a_s\le x<b^{\ast}$, and $V_{b^{*}}^{\prime}(x)=1$ for $x\ge b^{\ast}$. 
 \end{lemma}

We can now characterize the optimal dividend policy.

\begin{theorem}\label{maintheorem-ext}
The value function $V(x)=V_{b^{\ast}}(x)$ for $x \ge  {a_0}$ and the barrier strategy $D^{b^{\ast}}$ is optimal.
\end{theorem}
We have now characterized the optimal solution to the optimisation problem. 
The optimal strategy is the barrier strategy $D^{b^\ast}$ with barrier level $b^\ast$, under which dividends are paid whenever the surplus reaches or exceeds $b^\ast$. Specifically, any excess above $b^\ast$ is paid out immediately so that the controlled surplus process is reflected at the level $b^\ast$. The corresponding computational procedure is summarized below.

\noindent
\textbf{Step 1.}: Analytically or numerically determine the functions  $g_1(\cdot), \ldots 
g_5(\cdot)$ by solving the systems \eqref{1111-1-ext}
  and  \eqref{1111-3-ext}, 
respectively. Equivalently, $g_1(\cdot)$--$g_5(\cdot)$ can be obtained from the following initial value problems:
\begin{align}
\frac{1}{2}\sigma_d^2(x) g_1^{\prime\prime}(x)
+\mu_d(x) g_1^{\prime}(x)
-(\delta+\omega(x)) g_1(x)&=0,\ \quad x>a_0,\qquad
g_1(a_0)=0,\ \ g_1^{\prime}(a_0)=1.
\label{1111-1-00-ext}
\\
\frac{1}{2}\sigma_s^2(x) g_2^{\prime\prime}(x)
+\mu_s(x) g_2^{\prime}(x)
-\delta g_2(x)&=0,\ \quad x>a_d,\qquad
g_2(a_d)=1,\ g_2^{\prime}(a_d)=0,
\label{1111-3-00-ext}
\\
\frac{1}{2}\sigma_s^2(x) g_3^{\prime\prime}(x)
+\mu_s(x) g_3^{\prime}(x)
-\delta g_3(x)&=0,\ \quad x>a_d,\qquad
g_3(a_d)=0,\ g_3^{\prime}(a_d)=1,
\label{1111-5-00-ext}
\\
\frac{1}{2}\sigma^2(x) g_4^{\prime\prime}(x)
+\mu(x) g_4^{\prime}(x)
-\delta g_4(x)&=0,\ \quad x>a_s,\qquad
g_4(a_s)=1,\ g_4^{\prime}(a_s)=0,
\label{1111-7-00-ext}
\\
\frac{1}{2}\sigma^2(x) g_5^{\prime\prime}(x)
+\mu(x) g_5^{\prime}(x)
-\delta g_5(x)&=0,\ \quad x>a_s,\qquad
g_5(a_s)=0,\ g_5^{\prime}(a_s)=1.
\label{1111-9-001-ext}
\end{align}

\noindent
\textbf{Step 2}:    Compute the derivatives $g_1^\prime(\cdot)$ and $g_2^\prime(\cdot)$.

\noindent
\textbf{Step 3}: Determine the functions $C_1(\cdot)$, $C_2(\cdot)$ and $C_3(\cdot)$  using \eqref{c1b-ext}--\eqref{c5b-ext}:
   
\begin{align}
    C_1(b)&=\frac{1}{(g_1(a_d)g_2(a_s)+g_1^{\prime}(a_d)g_3(a_s))g_4^{\prime}(b)
    +(g_1(a_d)g_2^{\prime}(a_s)+g_1^{\prime}(a_d)g_3^{\prime}(a_s))g_5^{\prime}(b)},\label{c1b-00-ext}\\    
    C_2(b)&=g_1({a_d})C_1(b),
   \quad  C_3(b)=g_1^{\prime}({a_d})C_1(b),\label{c3b-00-ext}\\    
    C_4(b)&=g_2(a_s)C_2(b)+g_3(a_s)C_3(b)=\left(g_2(a_s)g_1({a_d})+g_3(a_s)g_1^{\prime}({a_d})\right)C_1(b),\label{c4b-00-ext}\\
    C_5(b)&=g_2^{\prime}(a_s)C_2(b)+g_3^{\prime}(a_s)C_3(b)=\left(g_2^{\prime}(a_s)g_1({a_d})+g_3^{\prime}(a_s)g_1^{\prime}({a_d})\right)C_1(b).\label{c5b-00-ext}
    \end{align}

\noindent
\textbf{Step 4}: The optimal strategy is the barrier strategy with a barrier $b^\ast$, where the barrier $b^{\ast}=\inf\{b\ge {a_s}: V_b^{\prime\prime}(b-)\ge 0\}$. Equivalently, $b^\ast$ is  the smallest $b\ge {a_s}$ such that $$\frac{(g_1(a_d)g_2(a_s)+g_1^{\prime}(a_d)g_3(a_s))g_4(b)
    +(g_1(a_d)g_2^{\prime}(a_s)+g_1^{\prime}(a_d)g_3^{\prime}(a_s))g_5(b)}{(g_1(a_d)g_2(a_s)+g_1^{\prime}(a_d)g_3(a_s))g_4^{\prime}(b)
    +(g_1(a_d)g_2^{\prime}(a_s)+g_1^{\prime}(a_d)g_3^{\prime}(a_s))g_5^{\prime}(b)}-\frac{\mu(b)}{\delta}\ge 0.$$

This can be implemented as follows:
 
(i)  Evaluate whether $\lim_{b\downarrow a_s}V_b^{\prime\prime}(b-)=\frac{(g_1(a_d)g_2(a_s)+g_1^{\prime}(a_d)g_3(a_s))g_4(a_s)
    +(g_1(a_d)g_2^{\prime}(a_s)+g_1^{\prime}(a_d)g_3^{\prime}(a_s))g_5(a_s)}{(g_1(a_d)g_2(a_s)+g_1^{\prime}(a_d)g_3(a_s))g_4^{\prime}(a_s)
    +(g_1(a_d)g_2^{\prime}(a_s)+g_1^{\prime}(a_d)g_3^{\prime}(a_s))g_5^{\prime}(a_s)}-\frac{\mu(a_s)}{\delta}$  is greater than or equal to $0$.

(ii) If the condition in (i)  holds,  then set  $b^{\ast}=a_s$.  Otherwise, determine $b^{\ast}$ as the smallest root on $(a_s,+\infty)$ of
    $\frac{(g_1(a_d)g_2(a_s)+g_1^{\prime}(a_d)g_3(a_s))g_4(b)
    +(g_1(a_d)g_2^{\prime}(a_s)+g_1^{\prime}(a_d)g_3^{\prime}(a_s))g_5(b)}{(g_1(a_d)g_2(a_s)+g_1^{\prime}(a_d)g_3(a_s))g_4^{\prime}(b)
    +(g_1(a_d)g_2^{\prime}(a_s)+g_1^{\prime}(a_d)g_3^{\prime}(a_s))g_5^{\prime}(b)}-\frac{\mu(b)}{\delta}$.\\
    
\noindent
\textbf{Step 5}: (i)  If $b^{\ast}>a_s$, determine the value function $V(\cdot)=V_{b^{\ast}}(\cdot)=f_{b^{\ast}}(\cdot)$ for $a_s<x<{b^{\ast}}$ by using \eqref{fb-ext}:
\begin{align}
   V(x)= V_{b^{\ast}}(x)=\begin{cases}
        c_1({b^{\ast}})  g_1(x), & x\in[{a_0},{a_d}],\\
        c_2({b^{\ast}}) g_2(x)+c_3({b^{\ast}}) g_3(x), & x\in ({a_d},{a_s}],\\
        c_4(b^{\ast}) g_4(x)+c_5(b^{\ast}) g_5(x), & 
        x\in({a_s},    {b^{\ast}}].       \end{cases}\label{fb-00-ext}
\end{align}
For $x>b^\ast$, the value function is given by
$
    V_{b^{\ast}}(x)=x-{b^{\ast}}+V_{b^{\ast}}({b^{\ast}}), \quad x>b^\ast.
$

\noindent (ii) 
If $b^{\ast}=a_s$, then 
\begin{align}
   V(x)=f_{a_s}(x)=\begin{cases}
        c_6({a_s})  g_1(x), & x\in[{a_0},{a_d}],\\
        c_7({a_s}) g_2(x)+c_8({a_s}) g_3(x), & x\in({a_d},{a_s}],   \\
        V(a_s)+x-a_s&x\in(a_s,\infty)\end{cases}
\end{align}
where 
\begin{align}
    C_6({a_s})&=\frac{1}{(g_1(a_d)g_2^\prime(a_s)+g_1^{\prime}(a_d)g_3^\prime(a_s))},
    \quad  
    C_7({a_s})=\frac{g_1(a_d)}{(g_1(a_d)g_2^\prime(a_s)+g_1^{\prime}(a_d)g_3^\prime(a_s))},\label{c7-ext-new} \\
     C_8({a_s})&=\frac{g_1^\prime(a_d)}{(g_1(a_d)g_2^\prime(a_s)+g_1^{\prime}(a_d)g_3^\prime(a_s))}.\label{c8-ext-new}
    \end{align}

In subsequent sections, 
we  apply the above framework to several examples  and investigate the ``value" of the additional feature in the Omega model. Before we proceed with this, we now first comment on an interesting impact that the liquidation-intensity feature has on the shape and concavity of the value function.

\begin{remark}\label{nonconcavity-ext}
By Theorem \ref{maintheorem-ext}, the value function satisfies $V(x)=V_{b^\ast}(x)$. Hence, by Lemma \ref{V-concavity-ext} the value function is concave on  $[a_s,\infty)$. However, the concavity of the value function  $[0,a_s)$ is unclear.
Our subsequent examples demonstrate that, on $[0,a_s)$, the value function may be concave, convex, or exhibit both concave and convex regions.  
This behavior contrasts sharply with the classical diffusion models without the liquidation-intensity feature or solvency constraints, where the value function is typically concave on the entire state space.
The threshold $a_s$ is determined by the distress region and any additional solvency constraints. When the surplus level is below $a_s$, dividend payments are optimally postponed until the surplus returns to the admissible region above $a_s$. 
This structural feature fundamentally alters the shape of the value function and gives rise to the potential loss of global concavity.
\end{remark}

\section{Expected survival 
time}
\label{sec:expectation}

Recall that $m(x)$ represents the expected survival time under the optimal strategy that achieves the maximal shareholder value in \eqref{suvival-new}. 
 Since Theorem~\ref{maintheorem-ext} established that the barrier strategy $D^{b^\ast}$ is optimal,
\begin{align*}
m(x)=\mathbb E_x\big[T^{D^{b^\ast}}\big], \qquad x\ge a_0.
\end{align*}

Under the Omega ruin mechanism, bankruptcy occurs either upon reaching the ultimate bankruptcy level $a_0$ or through the state-dependent hazard rate $\omega(\cdot)$ in the distress region $(a_0,a_d)$. Since under the barrier strategy the controlled surplus is confined to $[a_0,b^\ast]$ with $b^\ast<\infty$, the survival time is finite under the standard assumptions.
The following proposition provides an equivalent representation of the expected survival time.

\begin{proposition}\label{survival-representation-prop}
For $x\ge a_0$,
$
m(x)
=
\mathbb{E}_x\bigg[
\int_0^{\tau_{a_0}^{D^{b^\ast}}}
e^{-\int_0^s \omega(R_u^{D^{b^\ast}})\,du}\,ds
\bigg].
$
\end{proposition}

Proposition~\ref{survival-representation-prop} provides a convenient representation of the expected survival time in terms of the optimally controlled surplus process. As a consequence, the function $m$ satisfies a linear boundary-value problem. In the general case, the continuation region has three pieces, and the corresponding differential equations are inhomogeneous with forcing term equal to $1$.

\begin{theorem}\label{verification-ext-new}
If $b^{\ast}>a_s$, then 
  $m(x)$ is   continuously differentiable on $[a_0,b^\ast]$ and piecewise twice continuously differentiable  on $(a_0,a_d)$, $(a_d,a_s)$, and $(a_s,b^\ast)$, and it solves
\begin{align}
& \frac{1}{2}\sigma_d^2(x) {f}^{\prime\prime}(x)+\mu_d(x){f}^{\prime}(x)-\omega(x){f}(x)+1=0, \ {a_0}<x<{a_d},\label{eq1-00-ext}\\
&\frac{1}{2}\sigma_s^2(x) {f}^{\prime\prime}(x)+\mu_s(x){f}^{\prime}(x)+1=0, \ {a_d}<x<{a_s},\label{eq2-00-ext}\\
&\frac{1}{2}\sigma_u^2(x) {f}^{\prime\prime}(x)+\mu_u(x){f}^{\prime}(x)+1=0, \ {a_s}<x<b^\ast,\label{eq4-00-ext}\\
    & f(a_0)=0,\quad f^\prime(b^\ast)=0. \label{eq3-00-ext}
\end{align} 
Moreover, 
$m(x)=m(b^\ast)$ for $ x>b^\ast.$

If $b^{\ast}=a_s$, then 
$m(x)$    is continuously differentiable on $[a_0,\infty)$ and piecewise twice continuously differentiable on $(a_0,a_d)$ and $(a_d,a_s)$, and it solves
\begin{align}
& \frac{1}{2}\sigma_d^2(x) {f}^{\prime\prime}(x)+\mu_d(x){f}^{\prime}(x)-\omega(x){f}(x)+1=0, \ {a_0}<x<{a_d},\label{eq1-0000-ext}\\
& \frac{1}{2}\sigma_s^2(x) {f}^{\prime\prime}(x)+\mu_s(x){f}^{\prime}(x)+1=0, \ {a_d}<x<{a_s},\label{eq2-0000-ext}\\
    & f(a_0)=0,\quad f^\prime(a_s)=0. \label{eq3-0000-ext}
\end{align} 
Furthermore, 
$m(x)=m(a_s)$ for $ x>a_s.$
\end{theorem}

To solve the boundary-value problem, we introduce the corresponding homogeneous fundamental solutions and particular solutions.

\begin{definition}
\label{survival-fundamental-ext}
Let $\tilde g_4 , \ldots,\tilde g_9$ be the unique solutions to the following homogeneous initial value problems:
\begin{align}
(\tilde g_4):\quad
& \frac{1}{2}\sigma_d^2(x)\tilde g_4''(x)
  + \mu_d(x)\tilde g_4'(x)
  - \omega(x)\tilde g_4(x)=0, \quad x>a_0, \label{1111-1-000-ext}
  \qquad \tilde g_4(a_0)=0, \quad \tilde g_4'(a_0)=1, 
  \\
(\tilde g_5):\quad
& \frac{1}{2}\sigma_d^2(x)\tilde g_5''(x)
  + \mu_d(x)\tilde g_5'(x)
  - \omega(x)\tilde g_5(x)=0, \quad x>a_0, \label{1111-3-000-ext}
  \qquad \tilde g_5(a_0)=1, \quad \tilde g_5'(a_0)=0, 
  \\
(\tilde g_6):\quad
& \frac{1}{2}\sigma_s^2(x)\tilde g_6''(x)
  + \mu_s(x)\tilde g_6'(x)=0, \quad x>a_d, \label{1111-5-000-ext}
  \qquad \tilde g_6(a_d)=0, \quad \tilde g_6'(a_d)=1, 
  \\
(\tilde g_7):\quad
& \frac{1}{2}\sigma_s^2(x)\tilde g_7''(x)
  + \mu_s(x)\tilde g_7'(x)=0, \quad x>a_d, \label{1111-7-000-ext}
  \qquad \tilde g_7(a_d)=1, \quad \tilde g_7'(a_d)=0, 
  \\
(\tilde g_8):\quad
& \frac{1}{2}\sigma^2_u(x)\tilde g_8''(x)
  + \mu_u(x)\tilde g_8'(x)=0, \quad x>a_s, \label{1111-9-000-ext}
  \qquad \tilde g_8(a_s)=0, \quad \tilde g_8'(a_s)=1, 
  \\
(\tilde g_9):\quad
& \frac{1}{2}\sigma^2_u(x)\tilde g_9''(x)
  + \mu_u(x)\tilde g_9'(x)=0, \quad x>a_s, \label{1111-11-000-ext}
  \qquad \tilde g_9(a_s)=1, \quad \tilde g_9'(a_s)=0. 
\end{align}

\end{definition}
The existence and uniqueness of $\tilde{g}_4-\tilde{g}_9$ follows by Theorem 5.4.2. in \cite{Krylov1996}. 

\begin{definition}
Define the Wronskians
\[
W_{\tilde g_4,\tilde g_5}(x)=\tilde g_4(x)\tilde g_5'(x)-\tilde g_4'(x)\tilde g_5(x),\quad
W_{\tilde g_6,\tilde g_7}(x)=\tilde g_6(x)\tilde g_7'(x)-\tilde g_6'(x)\tilde g_7(x),
\]
\[
W_{\tilde g_8,\tilde g_9}(x)=\tilde g_8(x)\tilde g_9'(x)-\tilde g_8'(x)\tilde g_9(x),
\]
and
\[
B_j(x)=\tilde g_{2j+2}(x)\int_{a_j}^x
\frac{\tilde g_{2j+3}(y)}{W_{\tilde g_{2j+2},\tilde g_{2j+3}}(y)}
\frac{2}{\sigma_j^2(y)}\,dy
-\tilde g_{2j+3}(x)\int_{a_j}^x
\frac{\tilde g_{2j+2}(y)}{W_{\tilde g_{2j+2},\tilde g_{2j+3}}(y)}
\frac{2}{\sigma_j^2(y)}\,dy,
\]
for \(j=1,2,3\), with \((a_1,a_2,a_3)=(a_0,a_d,a_s)\).
%

\end{definition}

The functions $B_1,B_2,$ and $B_3$ are obtained by the method of variation of parameters.
Since  $W_{\tilde g_4,\tilde g_5}(a_0)=-1\neq 0$, $W_{\tilde g_6,\tilde g_7}(a_d)=-1\neq 0$ and $W_{\tilde g_8,\tilde g_9}(a_s)=-1\neq 0$, the pair  $(\tilde g_4, \tilde g_5)$ forms a linearly  independent pair of solutions to $\frac{1}{2}\sigma_d^2(x) {f}^{\prime\prime}(x)+\mu_d(x){f}^{\prime}(x)-\omega(x) {f}(x)=0$. Similarly,   $(\tilde g_6,\tilde g_7)$  forms a linearly independent pair of solutions to $\frac{1}{2}\sigma_s^2(x) {f}^{\prime\prime}(x)+\mu_s(x){f}^{\prime}(x)=0$, and $(\tilde g_8,\tilde g_9)$ forms a linearly independent pair of solutions  to $\frac{1}{2}\sigma^2_u(x) {f}^{\prime\prime}(x)+\mu_u(x){f}^{\prime}(x)=0$. 
Moreover, $B_1$, $B_2$ and $B_3$ are particular solutions to \eqref{eq1-00-ext}, \eqref{eq1-00-ext}, and \eqref{eq3-00-ext}, respectively.

The expected survival time can now be written in closed form.
\begin{theorem}\label{survival-representation}
If $b^\ast>a_s$, then the function admits the representation

\begin{align}
m(x)=\begin{cases}
\tilde C_4 \tilde g_4(x)+B_1(x)&a_0\le x\le a_d,\\
\tilde C_6 \tilde g_6(x)+ \tilde C_7 \tilde g_7(x)+B_2(x)& a_d< x\le a_s,\\
\tilde C_8 \tilde g_8(x)+ \tilde C_9 \tilde g_9(x)+B_3(x)& a_s< x\le b^\ast,\\
m(b^\ast)&x>b^\ast,
\end{cases}\label{171123-10-ext}
\end{align}

where the constants 
\begin{align}
&\tilde C_4= -\frac{B_1^\prime(a_d)\tilde g_6^\prime(a_s)\tilde g_8^\prime(b^\ast)+B_1^\prime(a_d)\tilde g_6(a_s)\tilde g_9^\prime(b^\ast)+B_1(a_d) \tilde g_7^\prime(a_s)\tilde g_8^\prime(b^\ast)+B_1(a_d) \tilde g_7(a_s)\tilde g_{{9}}^\prime(b^\ast)}
{\tilde g_4^\prime(a_d) \tilde g_6^\prime(a_s)\tilde g_8^\prime(b^\ast)+\tilde g_4^\prime(a_d) \tilde g_6(a_s)\tilde g_9^\prime(b^\ast)+\tilde g_4(a_d) \tilde g_7^\prime(a_s)\tilde g_8^\prime(b^\ast)+\tilde g_4(a_d) \tilde g_7(a_s)\tilde g_9^\prime(b^\ast)}\nonumber\\
&\quad\quad\quad-\frac{B_2^\prime(a_s)\tilde g_8^\prime(b^\ast)+B_2(a_s)\tilde g_9^\prime(b^\ast)+B_3^\prime(b^\ast)}{\tilde g_4^\prime(a_d) \tilde g_6^\prime(a_s)\tilde g_8^\prime(b^\ast)+\tilde g_4^\prime(a_d) \tilde g_6(a_s)\tilde g_9^\prime(b^\ast)+\tilde g_4(a_d) \tilde g_7^\prime(a_s)\tilde g_8^\prime(b^\ast)+\tilde g_4(a_d) \tilde g_7(a_s)\tilde g_9^\prime(b^\ast)},\label{171123-6-ext}\\
&\tilde C_6=  \tilde C_4 \tilde g_4^\prime(a_d)+B_1^\prime(a_d),
\quad \tilde C_7=\tilde C_4 \tilde g_4(a_d)+B_1(a_d),
\quad \tilde C_8=\tilde C_6 \tilde g_6^\prime(a_s)+ \tilde C_7 \tilde g_7^\prime(a_s)+B_2^\prime(a_s),\label{171123-9-ext}\\
&\tilde C_9=\tilde C_6 \tilde g_6(a_s)+ \tilde C_7 \tilde g_7(a_s)+B_2(a_s).\label{171123-10-ext-00}
\end{align}
If $b^{\ast}=a_s$, then 
\begin{align}
m(x)=\begin{cases}
\tilde C_{10} \tilde g_4(x)+B_1(x) & \quad a_0\le x\le a_d,\\
\tilde C_{12}\tilde g_6(x)+\tilde C_{13}\tilde g_7(x)+B_2(x) & \quad a_d< x\le a_s,\\
m(a_s) & \quad  x\ge a_s,
\end{cases}
\end{align}
where
\begin{align}
\tilde{ C}_{10}= -\frac{B_1^\prime(a_d)\tilde g_6^\prime(a_s)+B_1(a_d) \tilde g_7^\prime(a_s)+B_2^\prime(a_s)}{\tilde g_4^\prime(a_d) \tilde g_6^\prime(a_s)+\tilde g_4(a_d) \tilde g_7^\prime(a_s)},
\ 
\tilde{C}_{12}= \tilde C_{10} \tilde g_4^\prime(a_d)+B_1^\prime(a_d),
\ \tilde{C}_{13}= \tilde C_{10} \tilde g_4(a_d)+B_1(a_d).\label{171123-600-ext3}
\end{align}
\end{theorem}

\section{Reduced case: coinciding distress and solvency thresholds}\label{results-reduced}
We now consider the reduced case in which the distress threshold and the solvency threshold coincide,  $a_s=a_d$. 
In this setting,  the intermediate continuation region $(a_d,a_s)$ disappears. 
Consequently, in contrast to the general framework developed  in Section \ref{results},  the value function is constructed by solving the HJB system only on the distress region $(a_0,a_d)$ and the admissible dividend region $(a_d,\infty)$, rather than piecing together solutions across the distress region and the solvency region.
The reduced model is of independent interest because it corresponds to the practically important case in which the regulatory solvency threshold coincides with the upper boundary of the distress zone. 

\subsection{Shareholder value and optimal strategy}
The arguments below follow the same dynamic programming, verification, and smooth-fit ideas developed in Section \ref{results}. However, because the underlying ODE system and matching structure simplify considerably, we present the resulting formulas separately and omit proofs that are direct analogues of those already established in the general case.

Following standard  stochastic control arguments, the value function satisfies the dynamic programming principle: 
 \begin{align}
 V(x)
   =&\sup_{D\in\Pi}\mathbb{E}_x\Big[
\int_0^{\tau_{{a_0}}^D\wedge\tau}
e^{-(\delta s+\int_0^s\omega(R_u^D)du)}dD_s+e^{-\delta(\tau_{{a_0}}^D\wedge\tau)}V(R_{\tau_{{a_0}}^D\wedge\tau}^D)\Big],\ \ x> {a_0}.
 \end{align}
Formally, the associated HJB system becomes
 \begin{align}
      &\frac{1}{2}\sigma_d^2(x)V^{\prime\prime}(x)
     +\mu_d(x) V^{\prime}(x)-(\delta+\omega(x) ) V(x)=0,\ \ {a_0}<x<{a_d},\label{reduced-hjb-1}\\
        &\max\left\{\frac{1}{2}\sigma^2(x)V^{\prime\prime}(x)
     +\mu(x) V^{\prime}(x)-\delta V(x), 1-V^\prime(x)\right\}= 0,\ \ x\ge  {a_d}.\label{reduced-hjb-2}
 \end{align}

 The corresponding verification theorem is the direct analogue of Theorem \ref{verification-ext} and is therefore omitted. 

 As in the general case, we focus on barrier strategies. The construction of the barrier value function relies on the following fundamental solutions.



\begin{definition}\label{fundamentalsol}
(i) Let $g_1$ denote the unique solution to 
\begin{align}
&\frac{1}{2}\sigma_d^2(x) {f}^{\prime\prime}(x)+\mu_d(x){f}^{\prime}(x)-(\delta+\omega(x)) {f}(x)=0,\ \ x>{a_0},\label{1111-1}
\end{align}
subject to $f({a_0})=0$ and $ f^\prime({a_0})=1.$ 
(ii) Let $g_2(x)$ and $g_3(x)$ denote the unique   solutions to the 
\begin{align}
&\frac{1}{2}\sigma^2(x) {f}^{\prime\prime}(x)+\mu(x){f}^{\prime}(x)-\delta {f}(x)=0,\ \ x>{a_d}, \label{1111-3}
\end{align}
subject  respectively to 
$ g_2({a_d})=1,\ g_2^{\prime}({a_d})=0,$ 
and
$
 g_3({a_d})=0,\ g_3^{\prime}({a_d})=1.$ 

\end{definition}
The existence and uniqueness of these solutions follow from standard ODE theory.

The next lemma characterizes the solution  to the associated boundary-value problem.

\begin{lemma}\label{fb-sol-reduced}
For any $b>{a_d}$, the  boundary value problem,
\begin{align}
    &\frac{1}{2}\sigma_d^2(x) {f}^{\prime\prime}(x)+\mu_d(x){f}^{\prime}(x)-(\delta+\omega(x)) {f}(x)=0,\quad x\in(
{a_0},{a_d}), \label{diff1}\\
 &\frac{1}{2}\sigma^2(x) {f}^{\prime\prime}(x)+\mu(x){f}^{\prime}(x)-\delta {f}(x)=0,\quad x\in({a_d},b), \label{diff2}\\
&f({a_0})=0, \quad f^{\prime}(b)=1, \label{diff3}
\end{align}
admits a unique continuously differentiable solution  $f_b$ on   $({a_0},b]$. Moreover, 
\begin{align}
    f_b(x)=\begin{cases}
        c_1(b)  g_1(x), & x\in[{a_0},{a_d}),\\
         c_2(b) g_2(x)+c_3(b) g_3(x), & x\in[{a_d},b),
         \end{cases}\label{fb}
\end{align}
where
\begin{align}
    c_1(b)&=\left(g_1({a_d})g_2^{\prime}(b)
    +g_1^{\prime}({a_d})g_3^{\prime}(b)\right)^{-1},\label{c1b}
    \quad c_2(b)=g_1({a_d})c_1(b),
    \quad c_3(b)=g_1^{\prime}({a_d})c_1(b).
    \end{align}
   Furthermore, $f_b$ is differentiable  on $({a_0},\infty)$ and twice continuously differentiable on $({a_0},{a_d})$ and $({a_d}, b)$.
\end{lemma}

\begin{lemma}\label{fxbar-sol}
The following initial value differential equation,
\begin{align}
    &\frac{1}{2}\sigma_d^2(x) {f}^{\prime\prime}(x)+\mu_d(x){f}^{\prime}(x)-(\delta+\omega(x)) {f}(x)=0,\ \ x\in(
{a_0},{a_d}), \label{diff2-0}\\
&f({a_0})=0, \quad  f^{\prime}({a_d})=1, \label{diff3-0}
\end{align}
has a unique solution $f_{{a_d}}$, which is given by
$f_{{a_d}}(x)=\frac{g_1(x)}{g_1^\prime({a_d})},\quad x\in[{a_0},{a_d}),$
where  $g_1$ is the unique solution to \eqref{diff2-0} that satisfies the initial conditions $g_1({a_0})=0$ and $g_1^\prime({a_0})=1$. Furthermore, $f_{{a_d}}$ is  twice continuously differentiable on $({a_0},{a_d})$.
\end{lemma}
The proof follows the same matching argument as in Lemma \ref{fb-sol-ext} and is therefore omitted.


The next result identifies the value generated by a barrier strategy. 
\begin{theorem}\label{Vb}
For $b>{a_d}$, the value generated by the barrier strategy $D^b$ is
\begin{align}
V_b(x)=
\begin{cases}
    0& x<a_0,\\
    f_b(x),& {a_0} \le x\le b,\\
    f_b(b)+x-b,& x> b.\label{vb-x+}
\end{cases}
\end{align}
Moreover, $V_b$ is continuously differentiable on $[{a_0},\infty)$, and twice continuously differentiable on $({a_0},{a_d})$, $({a_d},b)$ and $ (b, \infty)$. Furthermore, $V_b^{\prime\prime}(b-)$ is continuous with respect to $b$ for $b>{a_d}$.

\end{theorem}
The proof follows the same verification and martingale arguments as in Theorem \ref{Vb-ext} and is omitted.



As in the general case, the optimal barrier is characterized through the smooth-fit condition at the reflection boundary. Define 
\begin{align}
b^{\ast}=\inf\{b>a: V_b^{\prime\prime}(b-)\ge 0\}.\label{b*}
\end{align}
Using the representation \eqref{fb}, this can equivalently be written as
\begin{align}\label{b*-1}
b^{\ast}&=\inf\{b\ge a_d:\ c_2(b)g_2(b)+c_3(b)g_3(b)>\frac{\mu(b)}{\delta}\}.\nonumber\\
=&\inf\left\{b\ge a_d:\
\frac{g_1({a_d})g_2(b)+g_1^\prime({a_d})g_3(b)}{g_1({a_d})g_2^{\prime}(b)
+g_1^{\prime}({a_d})g_3^{\prime}(b)}\ge \frac{\mu(b)}{\delta} \right\}.
\end{align}


As in the general framework, the optimal barrier is finite. The following result is the reduced-case analogue of Lemma \ref{finiteness}.

\begin{lemma}
It holds: $a_d\le b^{\ast}<+\infty$.
\end{lemma}
The proof follows directly from Lemma \ref{finiteness} after identifying $a_s=a_d$.


The next result is the analogue of Lemma \ref{V-concavity-ext}.

\begin{lemma} \label{V-concavity}
If $b^{*}>a_d$, then 
\begin{align}
  &V_{b^{*}}^{\prime\prime}(x)\le 0, \quad a_d\le x<b^{*},\quad V_{b^{*}}^{\prime}(x)>1, \quad a_d\le x<b^{\ast}, \quad V_{b^{*}}^{\prime}(x)=1, \quad  x\ge b^{\ast}\label{con2}.
   \end{align}
\end{lemma}
The proof follows the same contradiction argument as in Lemma \ref{V-concavity-ext}, after replacing the three-region decomposition by the present two-region structure.


We can now characterize the optimal strategy. 
\begin{theorem}\label{maintheorem}
The value function satisfies  $V(x)=V_{b^{\ast}}(x)$ for $x \ge  {a_0}$ and the barrier strategy $D^{b^{\ast}}$ is optimal.
\end{theorem}
The proof follows the same verification argument as in Theorem \ref{maintheorem-ext}, using the reduced HJB system \eqref{reduced-hjb-1}--\eqref{reduced-hjb-2} together with the structural properties established in Lemma \ref{V-concavity}.


Compared with the general framework, the reduced model admits a substantially simpler representation since only two continuation regions remain. In particular, the value function is constructed by matching solutions across a single threshold $a_d$,  rather than across the two interfaces $a_d$ and $a_s$ appearing in the general model. Accordingly, the computational procedure simplifies as 

\medskip

 \noindent Step 1: Determine the fundamental solutions   $g_1(\cdot)$, $g_2(\cdot)$ and $g_3(\cdot)$ from \eqref{1111-1} and \eqref{1111-3}, equivalently from the following initial value problems:
\begin{align}
&\frac{1}{2}\sigma_d^2(x)  g_1^{\prime\prime}(x)+\mu_d(x) g_1^{\prime}(x)-(\delta+\omega(x)) {g_1}(x)=0,\ \ x>{a_0},\label{1111-1-00} \quad g_1({a_0})=0,\ \ g_1^\prime({a_0})=1.
\\
&\frac{1}{2}\sigma^2(x)  g_2^{\prime\prime}(x)+\mu(x) g_2^{\prime}(x)-\delta {g_2}(x)=0,\ \ x>{a_0}, \label{1111-3-00}
\quad g_2({a_d})=1,\ g_2^{\prime}({a_d})=0,
\\
&\frac{1}{2}\sigma^2(x)  g_3^{\prime\prime}(x)+\mu(x) g_3^{\prime}(x)-\delta {g_3}(x)=0,\ \ x>{a_0}, \label{1111-3-001}
\quad g_3({a_d})=0,\ g_3^{\prime}({a_d})=1.
\end{align}

\medskip
 
\noindent  Step 2:  Compute
$c_1(b)=\left(g_1({a_d})g_2^{\prime}(b)
    +g_1^{\prime}({a_d})g_3^{\prime}(b)\right)^{-1}$,
    $ c_2(b)=g_1({a_d})c_1(b)$,
     $ c_3(b)=g_1^{\prime}({a_d})c_1(b).
 $

\medskip

\noindent Step 3: Determine the optimal barrier 
 $b^\ast=\inf\{b>a_d: V_b^{\prime\prime}(b-)\ge 0\}$. Equivalently, first evaluate  $\lim_{b\downarrow a_d}V_b^{\prime\prime}(b-)=\frac{g_1({a_d})g_2(a_d)+g_1^\prime({a_d})g_3(a_d)}{g_1({a_d})g_2^{\prime}(a_d)
+g_1^{\prime}({a_d})g_3^{\prime}(a_d)}-\frac{\mu(a_d)}{\delta}$. If this quantity is non-negative, then    $b^{\ast}=a_d$. 
Otherwise, determine $b^\ast$ as the smallest  root on $(a_d,\infty)$ of $\frac{g_1({a_d})g_2(b)+g_1^\prime({a_d})g_3(b)}{g_1({a_d})g_2^{\prime}(b)
+g_1^{\prime}({a_d})g_3^{\prime}(b)}-\frac{\mu(b)}{\delta}$.

\noindent Step 4: The value function is given by 
\begin{align}
    V_{b^{\ast}}(x)=\begin{cases}
        c_1({b^{\ast}})  g_1(x), & x\in[{a_0},{a_d}),\\
        c_2({b^{\ast}}) g_2(x)+c_3({b^{\ast}}) g_3(x), & x\in[{a_d},{b^{\ast}}),  \\
        =x-{b^{\ast}}+V_{b^{\ast}}({b^{\ast}}),& x>b^\ast.\end{cases}\label{fb-00}
\end{align}

The above procedure provides a complete characterization of the optimal strategy and the value function in the reduced model. We conclude this section by highlighting an important structural implication of the liquidation-intensity feature concerning the shape of the value function.

\begin{remark}\label{nonconcavity}
By Theorem \ref{maintheorem},  the value function satisfies $V(x)=V_{b^\ast}(x)$. Hence, by Lemma \ref{V-concavity}, the value function is concave on  $[a_s,\infty)$. However,global concavity of the value function is not guaranteed. In particular, on the distress region $(a_0,a_d)$, the value function may exhibit convexity or mixed convex-concave behavior, depending on the model parameters and the bankruptcy intensity function $\omega(\cdot)$. 
This phenomenon contrasts with many classical diffusion dividend models without a liquidation-intensity feature or solvency constraints, where the value function is typically globally concave. The loss of global concavity arises from the interaction between the distress  region and the restriction that dividend payments are suspended below the solvency threshold. 
\end{remark}

\subsection{Expected survival time}

Recall that under the Omega ruin mechanism, bankruptcy can occur either upon reaching the ultimate bankruptcy level $a_0$ or through the state-dependent hazard rate $\omega(\cdot)$ in the distress region $(a_0,a_d)$. In the reduced case $a_s=a_d$, the continuation region consists of only two pieces: $(a_0,a_d)$ and $(a_d,b^\ast)$.
Since the barrier strategy $D^{b^\ast}$  remains  optimal,  the expected survival time is
$$
m(x)=\mathbb{E}_x\!\left[T^{D^{b^\ast}}\right], \qquad x\ge a_0.
$$
The same argument as in  Proposition \ref{survival-representation-prop} 
yields  
\begin{align}
m(x)
=\mathbb{E}_x\bigg[\int_0^{\tau_{a_0}^{D^{b^{\ast}}}}
    e^{-\int_0^s\omega(R_u^{D^{b^{\ast}}})du}ds\bigg], \quad x\ge a_0.\label{expectedsurvival-reduced}
\end{align}
The derivation follows the same general ideas as in the general case. However, since the reduced model involves matching only two regions instead three, the algebraic details are different. We therefore state only the final boundary-value problem and the explicit representation and omit the proof. 

\begin{theorem}\label{thm:survival-reduced}
Assume that $b^\ast<\infty$.

If $b^\ast>a_d$, then $m(\cdot)$ is continuously differentiable on $[a_0,b^\ast]$ and piecewise twice continuously differentiable on $(a_0,a_d)$ and $(a_d,b^\ast)$, and it solves
\begin{align}
&\frac12 \sigma_d^2(x)f''(x)+\mu_d(x)f'(x)-\omega(x)f(x)+1=0,
\qquad a_0<x<a_d, \label{eq:red1}\\
&\frac12 \sigma^2(x)f''(x)+\mu(x)f'(x)+1=0,
\qquad a_d<x<b^\ast, \label{eq:red2}\\
&f(a_0)=0,\qquad f'(b^\ast)=0. \label{eq:red3}
\end{align}
Moreover,
$
m(x)=m(b^\ast)$ for $x>b^\ast.
$

If $b^\ast=a_d$, then $m(x)$ is continuously differentiable on $[a_0,\infty)$ and piecewise twice continuously differentiable on $(a_0,a_d)$, and it solves
\begin{align}
&\frac12 \sigma_d^2(x)f''(x)+\mu_d(x)f'(x)-\omega(x)f(x)+1=0,
\qquad a_0<x<a_d, \label{eq:red4}\\
&f(a_0)=0,\qquad f'(a_d)=0. \label{eq:red5}
\end{align}
Furthermore,
$
m(x)=m(a_d)$ for $ x>a_d.
$
\end{theorem}

Theorem \ref{thm:survival-reduced} characterizes the expected survival time through a boundary-value problem. We now derive an explicit representation of the solution in terms of fundamental solutions of the associated homogeneous equations and particular solutions obtained via variation of parameters.

Let $g_4,g_5$ solve
\begin{align}
&\frac12 \sigma_d^2(x)g''(x)+\mu_d(x)g'(x)-\omega(x)g(x)=0,
\qquad x>a_0, \\
&g_4(a_0)=0,\quad g_4'(a_0)=1, \qquad
g_5(a_0)=1,\quad g_5'(a_0)=0,
\end{align}
and let $g_6,g_7$ solve
\begin{align}
\frac12 &\sigma^2(x)g''(x)+\mu(x)g'(x)=0,
\qquad x>a_d, \\
&g_6(a_d)=0,\quad g_6'(a_d)=1, \qquad
g_7(a_d)=1,\quad g_7'(a_d)=0.
\end{align}
Define
\begin{align*}
B_1(x)
&=g_4(x)\int_{a_0}^x \frac{g_5(y)}{W(g_4,g_5)(y)}\frac{2}{\sigma_d^2(y)}\,dy
-g_5(x)\int_{a_0}^x \frac{g_4(y)}{W(g_4,g_5)(y)}\frac{2}{\sigma_d^2(y)}\,dy,\\
B_2(x)
&=g_6(x)\int_{a_d}^x \frac{g_7(y)}{W(g_6,g_7)(y)}\frac{2}{\sigma^2(y)}\,dy
-g_7(x)\int_{a_d}^x \frac{g_6(y)}{W(g_6,g_7)(y)}\frac{2}{\sigma^2(y)}\,dy.
\end{align*}

Using the above fundamental solutions and particular solutions, the expected survival time admits the following explicit representation.
\begin{theorem}
If $b^\ast>a_d$,
\begin{align}
m(x)=
\begin{cases}
C_4\,g_4(x)+B_1(x), & a_0\le x<a_d,\\
C_6\,g_6(x)+C_7\,g_7(x)+B_2(x), & a_d\le x\le b^\ast,\\
m(b^\ast;\theta,b^\ast), & x>b^\ast,
\end{cases}
\label{expsurvial-expression1}
\end{align}
 where the constants are given by
\begin{align*}
C_4
&=
-\frac{B_1'(a_d)\,g_6'(b^\ast)+B_1(a_d)\,g_7'(b^\ast)+B_2'(b^\ast)}
{g_4'(a_d)\,g_6'(b^\ast)+g_4(a_d)\,g_7'(b^\ast)}, 
\quad 
C_6=C_4 g_4'(a_d)+B_1'(a_d), 
\quad C_7=C_4 g_4(a_d)+B_1(a_d). 
\end{align*}


If $b^\ast=a_d$, then
\begin{align}
m(x)=
\begin{cases}
C_8\,g_4(x)+B_1(x), & a_0\le x\le a_d,\\
m(a_d;\theta,b^\ast), & x>a_d,
\end{cases}
\label{expsurvial-expression2}
\end{align}
where 
$
C_8=-\frac{B_1'(a_d)}{g_4'(a_d)}.$ 

\end{theorem}

%
%
%

\section{Numerical illustrations}\label{sec:numerical}

This section illustrates the framework and its main messages with five examples. Example~\ref{num:general} exercises the full generality of Section~\ref{results}: a model in which the drift and volatility genuinely
differ across the three zones. Example~\ref{num:nonconcavity} exhibits the non-concavity phenomenon. Example~\ref{num:design} is the central one: it compares the traditional model with three reduced-form redesigns and documents when the liquidation-intensity feature improves shareholder value, firm survival, or both. Example~\ref{num:sensitivity} maps the sensitivity to the design levers and contrasts the non-increasing intensity used throughout with the constant intensity of the earlier literature. Example~\ref{num:simulation}
goes beyond the closed-form expected survival time of Section~\ref{sec:expectation} and studies, by simulation, the full joint distribution of the time to liquidation and the realised shareholder value.

Throughout we take the discount rate \(\delta = 0.10\). Unless stated otherwise the liquidation intensity is the non-increasing function 
\begin{equation}\label{eq:omega-decr}
  \omega(x) \;=\; d\left(\frac{a_d - x}{a_d - a_0}\right)^{\gamma} \text{ for }a_0 < x < a_d, \qquad  \text{ and } \qquad \omega(x) = 0 \text{ for } x \ge a_d,
\end{equation}
with peak intensity \(d>0\) attained at the lower threshold \(a_0\) and decaying to zero at \(a_d\) (we use the linear choice \(\gamma = 1\) unless noted otherwise). The form \eqref{eq:omega-decr} is natural: the liquidation pressure is strongest when the firm is deepest in distress and relaxes as the surplus recovers, and it joins the regulated zone continuously since \(\delta + \omega(a_d-) = \delta\). The constant intensity \(\omega(\cdot) \equiv d\) of \cite{AlbrecherGerberShiu2011} and subsequent work is the special case obtained by replacing \eqref{eq:omega-decr} with a constant; we return to this contrast in Example~\ref{num:sensitivity}. All optimal barriers, value functions and survival times below are computed with the constructive procedure of Section~\ref{results}; the accompanying code reproduces every figure and table.

\subsection{Example 1: a general three-zone model}\label{num:general}

We first consider an instance with genuinely zone-dependent dynamics, to emphasise that the framework is not confined to a stationary process. Let \(a_0 = 0\), \(a_d = 4\) and \(a_s = 7\), and take
\begin{equation*}
  \mathrm{d}R_t^D =
  \begin{cases}
    \mu_d\,\mathrm{d}t + \sigma_d\,\mathrm{d}W_t, & a_0 < R_{t-}^D < a_d,\\[2pt]
    \mu_s(R_{t-}^D)\,\mathrm{d}t + \sigma_s(R_{t-}^D)\,\mathrm{d}W_t, & a_d \le R_{t-}^D < a_s,\\[2pt]
    \mu_u(R_{t-}^D)\,\mathrm{d}t + \sigma_u(R_{t-}^D)\,\mathrm{d}W_t, & R_{t-}^D \ge a_s,
  \end{cases}
\end{equation*}
with a low, flat drift in distress (\(\mu_d = 0.6\), \(\sigma_d = 1.8\)), and piecewise-linear coefficients that steepen as the firm stabilises and then operates freely: \(\mu_s(x) = \mu_d + 0.05\,(x-a_d)\), \(\sigma_s(x) = \sigma_d + 0.18\,(x-a_d)\) on the regulated zone $[a_d,a_s)$, and \(\mu_u(x) = \mu_s(a_s) + 0.09\,(x-a_s)\), \(\sigma_u(x) = \sigma_s(a_s) + 0.40\,(x-a_s)\) on the unregulated zone $[a_s,\infty)$. The intensity is \eqref{eq:omega-decr} with \(d = 0.8\). The growth condition \(\sup \mu_u'(x) = 0.09 < \delta = 0.10\) for $x\ge a_s$ holds.

Solving the model gives an optimal barrier \(b^{\ast} = 7.844\), which lies above \(a_s\) so that all three zones and an interior dividend barrier are present. Table~\ref{tab:ex1} reports the value function and expected survival time at selected surplus levels, and Figure~\ref{fig:ex1} displays \(V\) over the three (shaded) zones together with the zone-dependent drift and volatility profiles. The value function is smooth across the thresholds and linear with unit slope beyond \(b^{\ast}\), as it must be; the expected survival time saturates at
\(b^{\ast}\) (any surplus above the barrier is immediately distributed).

\begin{table}[H]
\centering
\begin{tabular}{r r r r r r r}
\hline
\(x\)   & 0.5 & 2 & 4 & 6 & 8 & 10 \\
\hline
\(V(x)\) & 0.474 & 1.881 & 4.245 & 6.408 & 8.416 & 10.416 \\
\(m(x)\) & 2.63  & 9.28  & 17.37 & 20.99 & \cellcolor{gray!25}{21.66} & \cellcolor{gray!25}{21.66} \\
\hline
\end{tabular}
\caption{Example 1: shareholder value \(V(x)\) and expected survival time
\(m(x)\) for the general three-zone model (\(b^{\ast}=7.844\)).}
\label{tab:ex1}
\end{table}

\begin{figure}[H]
\centering
\includegraphics[width=0.49\linewidth]{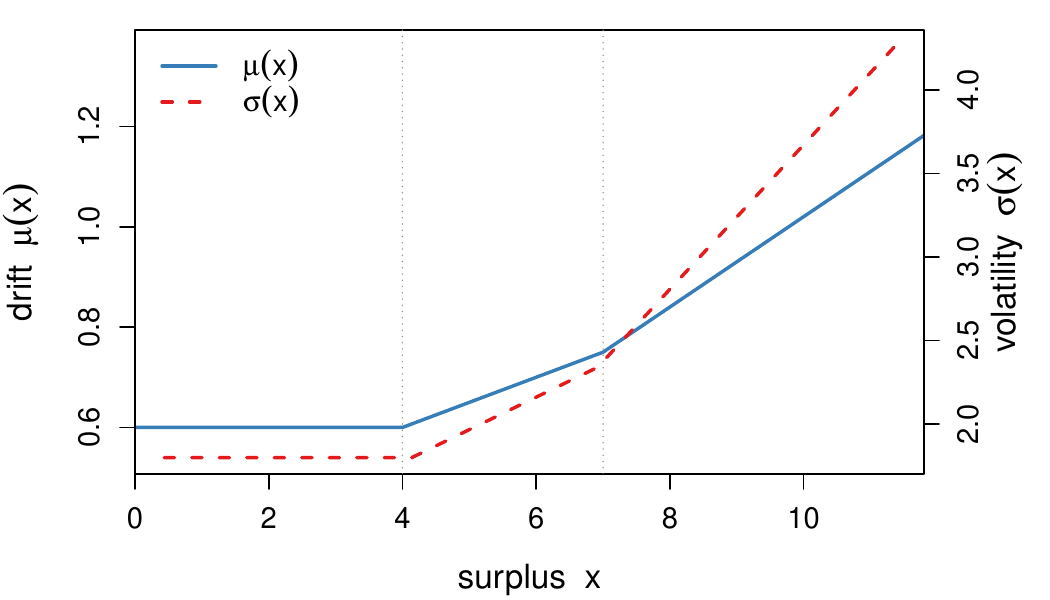}
\includegraphics[width=0.49\linewidth]{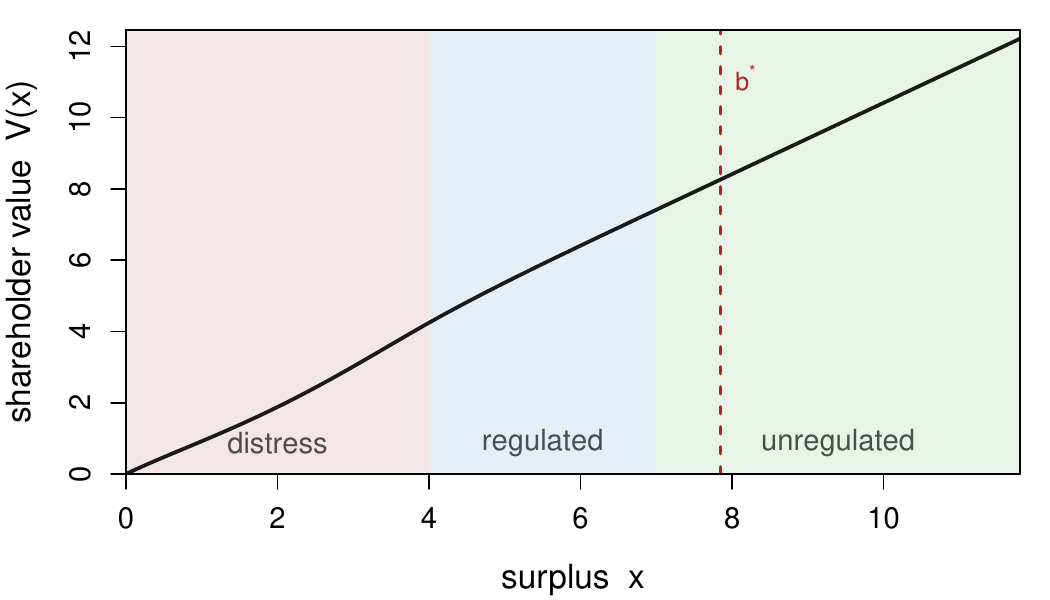}
\caption{Example 1. \textit{Left}: the zone-dependent drift \(\mu(x)\) and volatility \(\sigma(x)\). \textit{Right}: the value function \(V(x)\) across the distress, regulated and unregulated zones (shaded), with the optimal barrier \(b^{\ast}\) marked.}
\label{fig:ex1}
\end{figure}

\subsection{Example 2: non-concavity of the value function}\label{num:nonconcavity}

In the classical diffusion dividend problem the value function is concave everywhere. This is not observed here: the interaction of the liquidation intensity with the solvency constraint can make \(V\) convex over part of the low-surplus region. We take \(a_0 = 0\), \(a_d = 5\), \(a_s = 11\), a moderate set of dynamics (\(\mu_d = 0.8\), \(\sigma_d = 2.0\), with slopes \(0.04\) and \(0.10\)/\(0.15\) on the regulated/unregulated zones), and intensity
\eqref{eq:omega-decr} with \(d = 2\).

\begin{figure}[htb]
\centering
\begin{subfigure}{0.49\linewidth}
  \centering
  \includegraphics[width=\linewidth]{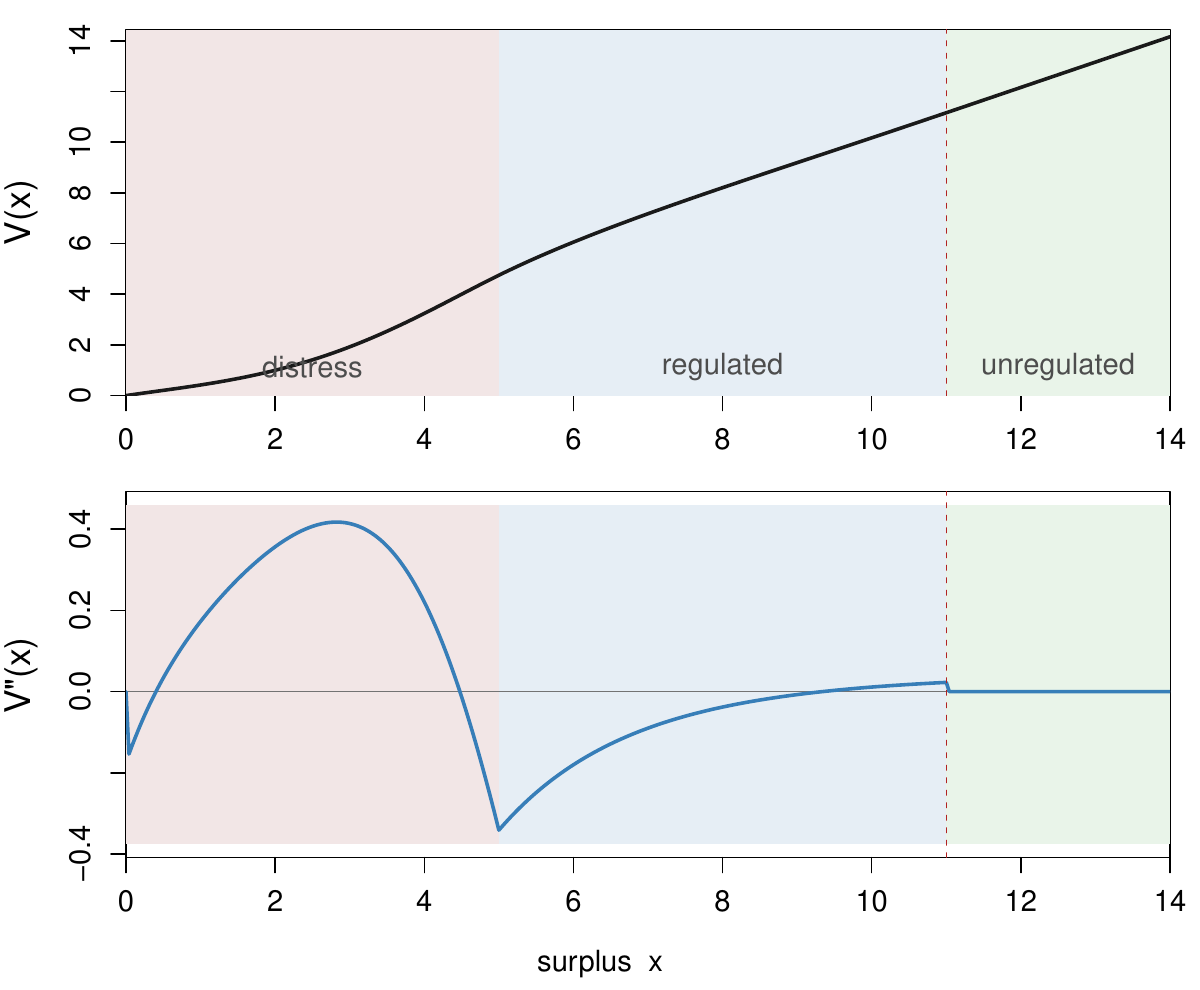}
  \caption{Binding constraint (\(a_s=11\)), so \(b^{\ast}=a_s\).}
  \label{fig:ex2nc}
\end{subfigure}\hfill
\begin{subfigure}{0.49\linewidth}
  \centering
  \includegraphics[width=\linewidth]{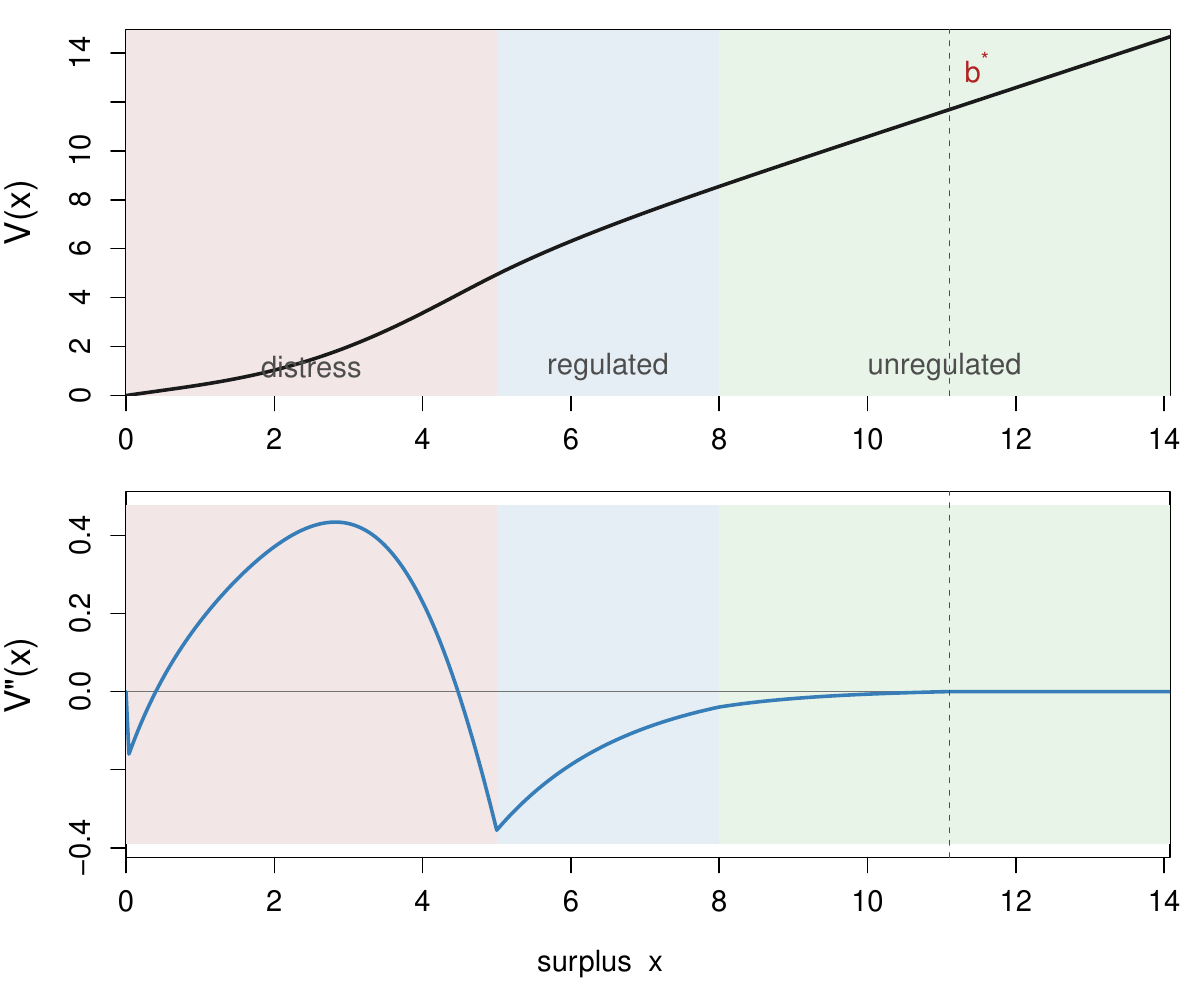}
  \caption{Slack constraint (\(a_s=8\)), so \(b^{\ast}=11.1>a_s\).}
  \label{fig:ex2ncslack}
\end{subfigure}
\caption{Example 2: the value function \(V\) (top of each panel) and its second
derivative \(V''\) (bottom); zones shaded as in Figure~\ref{fig:ex1}. In both
cases \(V''>0\) over part of the distress zone (convex) and \(V''<0\) in the
regulated zone (concave). \textit{(a)} When the solvency constraint binds,
\(b^{\ast}=a_s=11\) is fixed by the constraint rather than by smooth fit, so
\(V''\) is positive just below \(b^{\ast}\) before vanishing beyond it.
\textit{(b)} When the constraint is slack, \(b^{\ast}=11.1>a_s\) is set by smooth
fit, so \(V''\to0\) smoothly at \(b^{\ast}\) and the unregulated zone is concave.}
\label{fig:ex2nc-combined}
\end{figure}

The optimal barrier is \(b^{\ast} = 11.000\). Figure~\ref{fig:ex2nc} plots \(V\) and its second derivative. The curvature changes sign within \([0,a_s)\): \(V''\) is \textit{positive} (the value function is convex) over a sub-interval of the distress zone and \textit{negative} (concave) in the regulated zone, before becoming zero in the unregulated zone where the firm controls the surplus optimally. Numerically, \(V''(2) = +0.356\) while \(V''(5) = -0.341\) and \(V''(9) = -0.007\); the second derivative changes sign near \(x \approx 0.40\), \(4.47\) and \(9.34\). The value function is therefore a genuine mix of convexity and concavity on \([0,a_s)\) while remaining concave on \([a_s,\infty)\), consistent with the statement in Section~\ref{intro:contributions}: below \(a_s\) the firm cannot exercise optimal dividend control, and the prospect of intensity-driven liquidation makes an extra unit of surplus increasingly valuable over a range, breaking concavity.

The positive values of \(V''\) just below \(b^{\ast}\) in Figure~\ref{fig:ex2nc} are not a numerical artefact. There the solvency constraint binds (\(b^{\ast}=a_s=11\), whereas the unconstrained barrier for these dynamics would be \(13.6\)), so the barrier is fixed by the constraint rather than by the smooth-fit condition \(V''(b^{\ast})=0\); the firm is forced to retain capital it would rather distribute, so \(V'<1\) rises to \(1\) at \(a_s\) and \(V''>0\) just below it. Figure~\ref{fig:ex2ncslack} shows the contrasting unconstrained case, with \(V''\) returning smoothly to \(0\) at \(b^{\ast}\).

\begin{figure}[htb]
\centering
\includegraphics[width=0.6\linewidth]{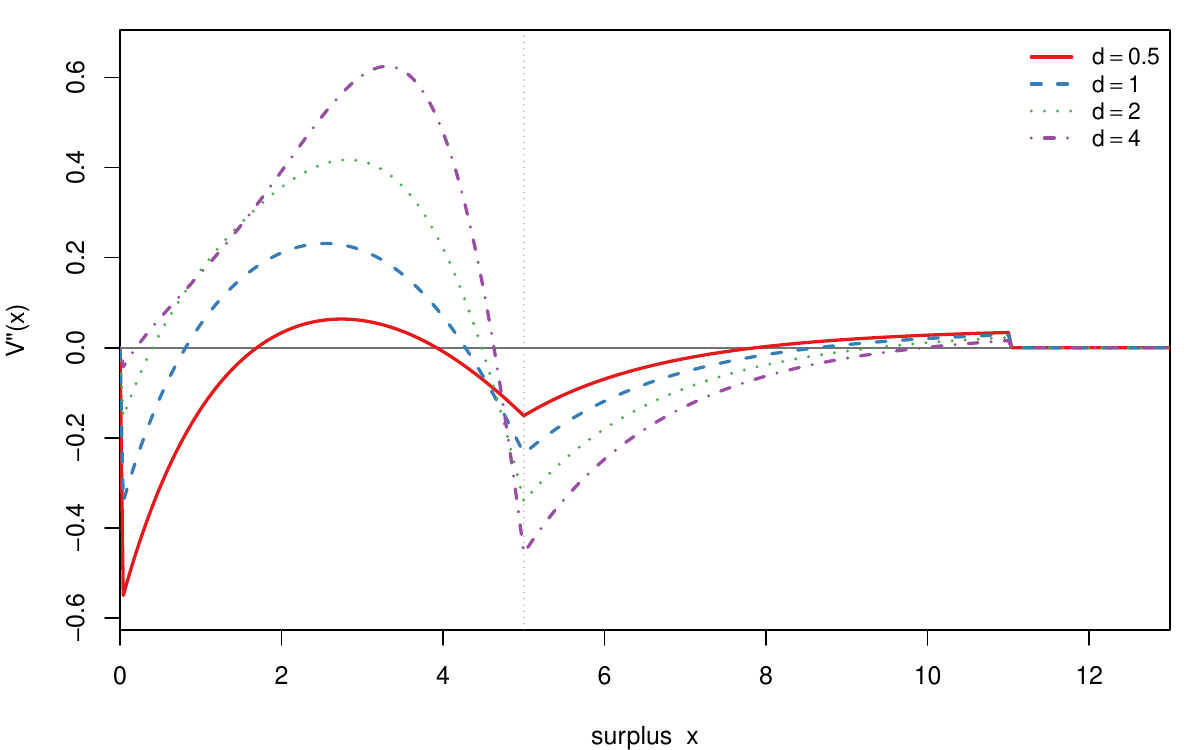}
\caption{Example 2. The second derivative \(V''(x)\) for peak intensities
\(d \in \{0.5,1,2,4\}\). Higher liquidation intensity enlarges and deepens the
convex region in the distress zone.}
\label{fig:ex2sweep}
\end{figure}

Figure~\ref{fig:ex2sweep} shows how the effect strengthens with the intensity. As the peak \(d\) rises from \(0.5\) to \(4\), the second derivative in the distress zone shifts upward---\(V''(0.5)\) moves from \(-0.308\) through \(-0.109\), \(+0.034\) to \(+0.066\)---so the convex region emerges and deepens as liquidation pressure grows.

\subsection{Example 3: regulatory design and the value--stability trade-off}\label{num:design}

This is the central example. We ask whether, and how, the liquidation-intensity feature can improve on the traditional model. To isolate the effect of the \textit{design} from any effect of the dynamics, all models in this example share the \textit{same} diffusion coefficients in every zone, \(\mu(x) = 2 + 0.04\,x\) and \(\sigma(x) = 4\), and differ only in how the zones are placed relative to the classical ruin threshold \(a_r = 0\):
\begin{itemize}
\item \textbf{Traditional}: liquidation at first passage to \(a_0 = 0\), with no intensity zone: $a_0=a_d=a_s$. This is the benchmark.
\item \textbf{Buffer-only}: an intensity zone placed entirely \textit{below} the classical threshold, \(a_0 = -2\), \(a_d = a_s = 0\), with intensity \eqref{eq:omega-decr} (\(d=1\)). The firm may continue operating in distress on \((-2,0)\) instead of being liquidated at \(0\); no dividend restriction is added above \(0\).
\item \textbf{Regulation-only}: a no-dividend solvency band placed entirely \textit{above} the classical threshold, \(a_0 = 0\), \(a_s = 16\), with \(\omega \equiv 0\) (no liquidation intensity). Dividends are prohibited on \((0,16)\), an otherwise solvent low-reserves zone.
\item \textbf{Combined}: both features together, \(a_0 = -2\), \(a_d = 0\), \(a_s = 13\), with intensity \eqref{eq:omega-decr} (\(d=1\)) on the buffer \((-2,0)\) and a no-dividend band on \((0,13)\).
\end{itemize}

The corresponding optimal barriers are \(b^{\ast} = 13.384\) (traditional), \(11.438\) (buffer-only), \(16.000\) (regulation-only, where the binding solvency band fixes the barrier at \(a_s\)) and \(13.000\) (combined).

Two observations explain the results that follow. A no-dividend band only bites if it sits at or above the unconstrained barrier: with the natural barrier near \(13.4\), a band reaching only to a low \(a_s\) is slack and regulation-only then collapses onto the traditional model. We therefore set \(a_s = 16\) for regulation-only and \(a_s = 13\) for combined, so that the restriction is binding in both.

\begin{table}[htb]
\centering
\begin{tabular}{l c c c c c}
\hline
Design & \(b^{\ast}\) & \(V(5)\) & \(m(5)\) &
  \(\min_x \Delta V\) & \(\min_x \Delta m\) \\
\hline
Traditional     & 13.384 & 15.341 & 49.64  & ---    & ---    \\
Buffer-only     & 11.438 & 17.550 & 51.59  & \(+1.17\) & \(-4.78\) \\
Regulation-only & 16.000 & 15.043 & 113.71 & \(-0.52\) & \(+10.71\) \\
Combined        & 13.000 & 17.414 & 83.62  & \(+0.97\) & \(+32.28\) \\
\hline
\end{tabular}
\caption{Example 3: optimal barrier, value and expected survival at \(x_0=5\), and the minimum gains over the traditional model across \(x\in[0.5,15]\) (\(\Delta V = V_{\text{design}}-V_{\text{trad}}\), likewise \(\Delta m\)). A design improves on the traditional model in the Pareto sense iff both minima are non-negative: here, only the combined design does.}
\label{tab:ex3}
\end{table}

Table~\ref{tab:ex3} summarises the outcome at a representative starting surplus \(x_0 = 5\), together with the worst-case gains over the traditional model across \(x \in [0.5, 15]\) (from just above the lowest threshold to a well-capitalised firm). The single-lever designs each improve exactly one objective at the expense of the other. The buffer-only design raises shareholder value everywhere (\(\min_x \Delta V = +1.17\)): the option to keep operating below \(0\) is worth having, and it lowers the optimal barrier to \(11.44\). But survival is \textit{not} uniformly improved (\(\min_x \Delta m = -4.78\)): because the firm now operates leaner and distributes sooner, a well-capitalised firm is liquidated \textit{earlier} than under the traditional rule. Conversely the regulation-only design lengthens survival everywhere (\(\min_x \Delta m = +10.71\), and far more at higher surplus) by forcing the firm to retain capital, but it \textit{reduces} value everywhere (\(\min_x \Delta V = -0.52\)) because dividends are delayed and discounted more heavily.

Remarkably, the combined design improves \textit{both}: \(\min_x \Delta V = +0.97 > 0\) and \(\min_x \Delta m = +32.28 > 0\) over the whole range, a Pareto improvement over the traditional model. The mechanism is complementary. The buffer thus substitutes for part of the regulation: the combined design attains regulation-only's stability gains at a strictly lower solvency threshold (\(a_s=13\) vs.\ \(16\)), which is what preserves its value  advantage. In other words, the buffer protects the downside and lifts value; the binding solvency band rebuilds the survival that the buffer alone would erode at higher surplus. Strikingly, the combined design delivers essentially the \textit{same} shareholder value as buffer-only (\(V(5) = 17.41\) vs.\ \(17.55\)) while \textit{also} delivering far longer
expected survival (\(m(5) = 83.6\) vs.\ \(51.6\)). 

\begin{figure}[htb]
\centering
\includegraphics[width=0.49\linewidth]{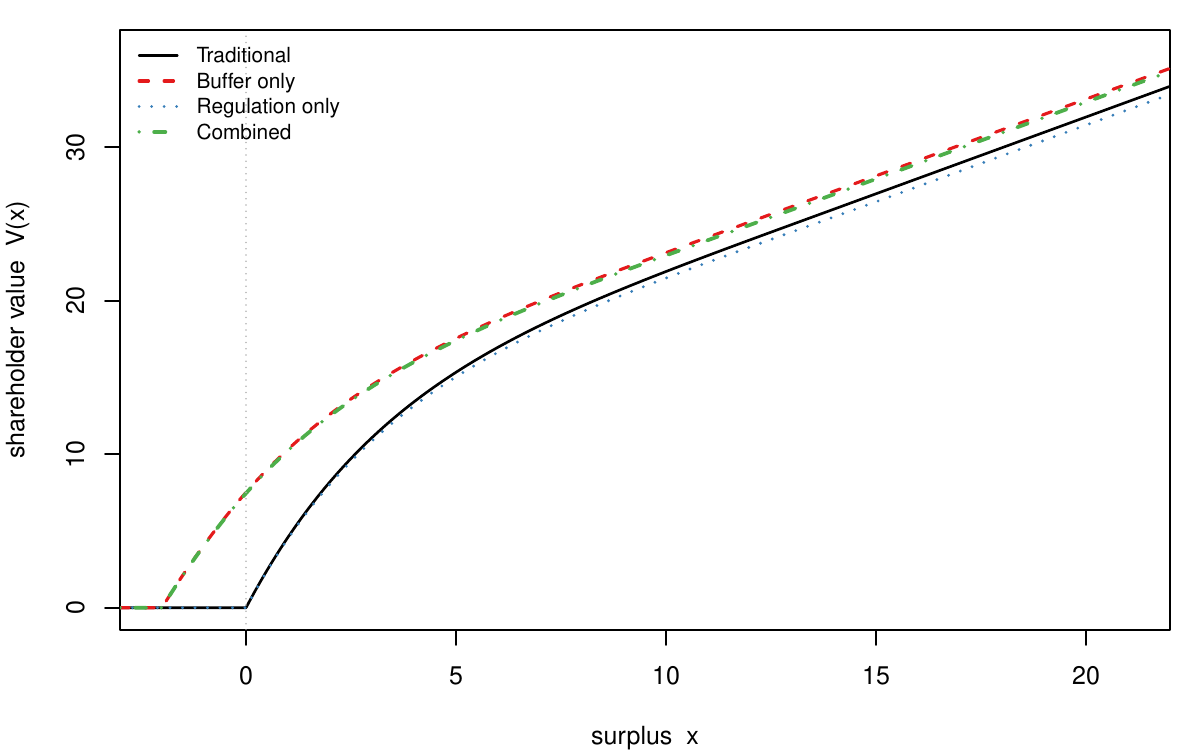}\hfill
\includegraphics[width=0.49\linewidth]{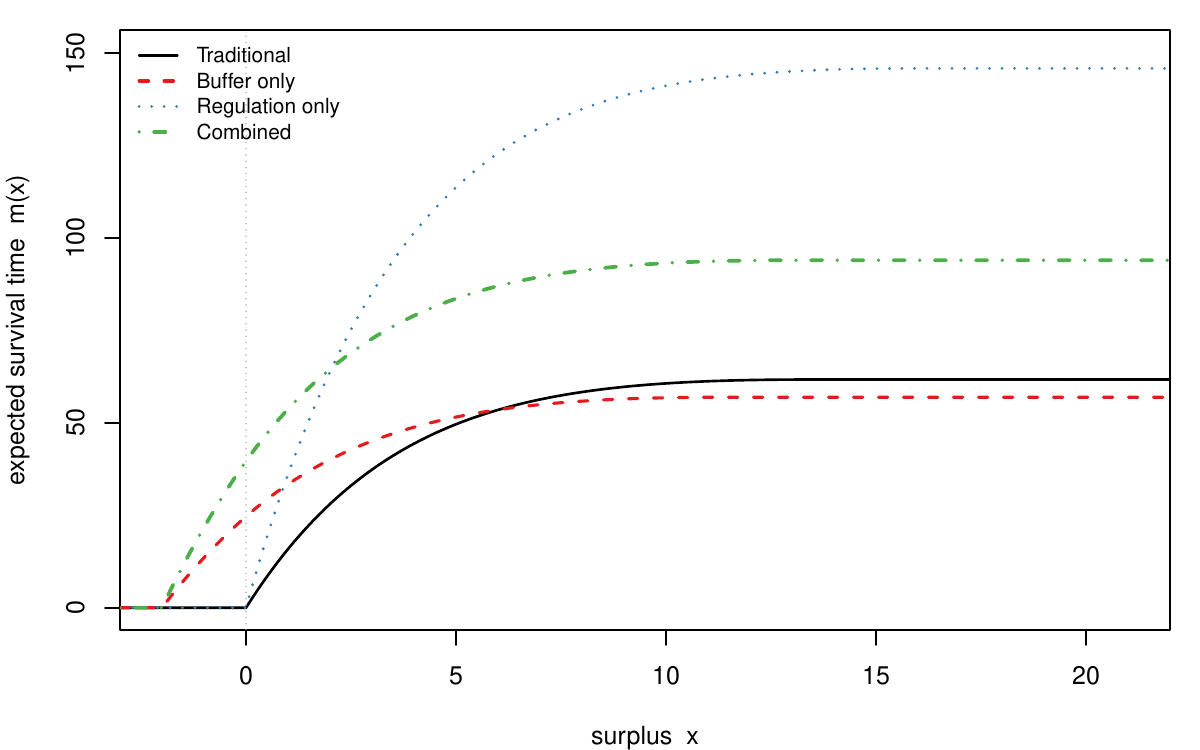}
\caption{Example 3. Shareholder value \(V(x)\) (left) and expected survival
\(m(x)\) (right) for the four designs.}
\label{fig:ex3curves}
\end{figure}

\begin{figure}[htb]
\centering
\includegraphics[width=0.6\linewidth]{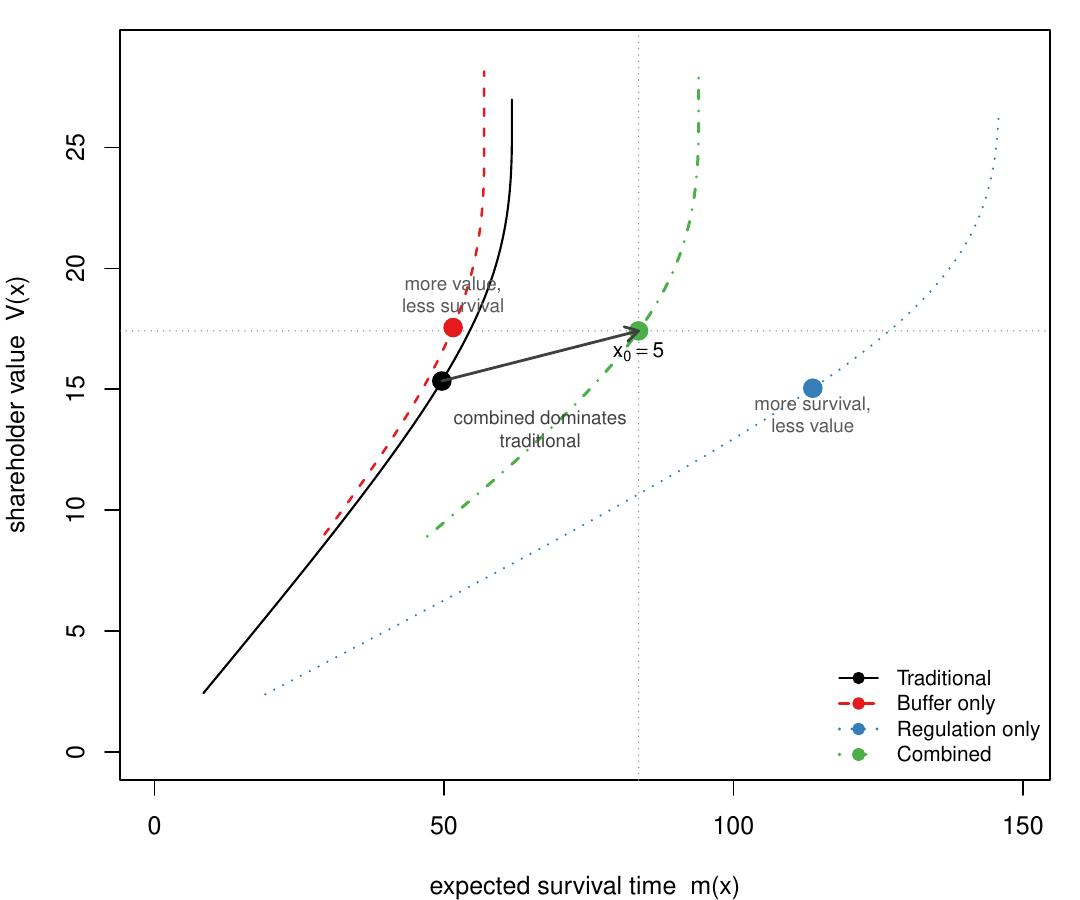}
\caption{Example 3. The value--stability trade-off, traced parametrically in the starting surplus \(x\) (curves), with the four designs marked at the common \(x_0=5\) (dots). Dotted crosshairs pass through the combined design's point: its upper-right quadrant is the set of outcomes that would Pareto-dominate it, and no rival design's \(x_0=5\) marker lies there. The combined design dominates the traditional model (arrow); buffer-only gives more value but less survival (upper-left), regulation-only more survival but less value (lower-right).}
\label{fig:ex3tradeoff}
\end{figure}

Figures~\ref{fig:ex3curves} and~\ref{fig:ex3tradeoff} display this: in the value--survival plane (Figure \ref{fig:ex3tradeoff}) the combined curve lies up and to the right of the traditional curve, dominating it, whereas buffer-only sits above but bulges to the left at high survival (more value, less survival) and regulation-only lies to the right but below (more survival, less value). Figure~\ref{fig:ex3gains} shows the gains \(\Delta V\) and \(\Delta m\) as functions of the surplus.

\begin{figure}[H]
\centering
\includegraphics[width=0.92\linewidth]{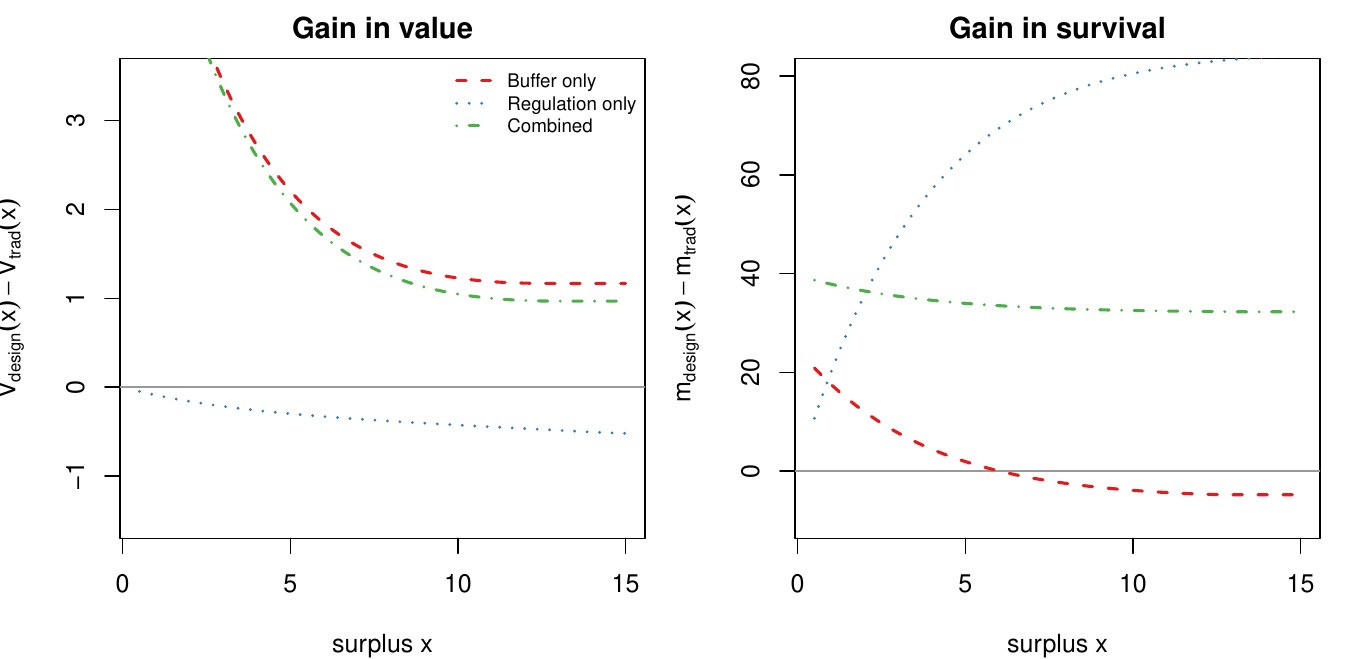}
\caption{Example 3. Gains over the traditional model: \(V_{\text{design}}(x) -
V_{\text{trad}}(x)\) (left) and \(m_{\text{design}}(x) - m_{\text{trad}}(x)\)
(right). The combined design is the only one that stays non-negative on both
panels.}
\label{fig:ex3gains}
\end{figure}

\subsection{Example 4: sensitivity and the role of the intensity shape}\label{num:sensitivity}

We now vary the design levers around the combined configuration of Example~\ref{num:design} (buffer \(a_0 = -2\), band up to \(a_s\), intensity \eqref{eq:omega-decr}), evaluating at \(x_0 = 5\), where the traditional model has \(V = 15.34\) and \(m = 49.64\).

\subsubsection{Strength of the solvency band.} 
Table~\ref{tab:ex4as} varies the solvency threshold \(a_s\). Tightening regulation trades value for stability monotonically: as \(a_s\) rises from \(11\) to \(18\), \(V(5)\) falls from \(17.55\) to \(16.09\) while \(m(5)\) climbs from \(51.6\) to \(408.2\). Over this whole range value stays above the traditional level, so the combined design remains a value improvement; the band simply buys progressively more survival at a declining value cost. Figure~\ref{fig:ex4as} plots both curves.

\begin{table}[H]
\centering
\begin{tabular}{r r r r r r r r r}
\hline
\(a_s\)     & 11 & 12 & 13 & 14 & 15 & 16 & 17 & 18 \\
\hline
\(b^{\ast}\) & 11.44 & 12.00 & 13.00 & 14.00 & 15.00 & 16.00 & 17.00 & 18.00 \\
\(V(5)\)     & 17.55 & 17.53 & 17.41 & 17.22 & 16.97 & 16.69 & 16.39 & 16.09 \\
\(m(5)\)     & 51.6  & 61.4  & 83.6  & 114.1 & 156.1 & 214.3 & 295.1 & 408.2 \\
\hline
\end{tabular}
\caption{Example 4: Effect of the solvency threshold \(a_s\) on the combined design (buffer \(a_0=-2\), peak intensity \(d=1\)).}
\label{tab:ex4as}
\end{table}

\subsubsection{Strength of the liquidation intensity.} 
Holding \(a_s = 13\) fixed and varying the peak intensity \(d\) (Table~\ref{tab:ex4d}) has a milder, also monotone, effect: a stronger hazard in the buffer slightly reduces both value and survival (as \(d\) goes from \(0.25\) to \(4\), \(V(5)\) falls from \(17.46\) to \(17.26\) and \(m(5)\) from \(85.1\) to \(78.8\)), while the binding band keeps \(b^{\ast}\) at \(13\) throughout. For every \(d\) the unconstrained barrier (\(\approx 11.4\)–\(11.6\), itself increasing in \(d\)) lies below \(a_s=13\), so the solvency constraint binds and the solver returns \(b^{\ast}=a_s=13\) throughout; \(d\) therefore moves \(V\) and \(m\) but not the barrier.
Figure~\ref{fig:ex4d} plots these.

\begin{table}[H]
\centering
\begin{tabular}{r r r r r r}
\hline
\(d\)        & 0.25 & 0.5 & 1 & 2 & 4 \\
\hline
\(V(5)\)     & 17.459 & 17.444 & 17.414 & 17.358 & 17.260 \\
\(m(5)\)     & 85.14  & 84.62  & 83.62  & 81.82  & 78.81  \\
\hline
\end{tabular}
\caption{Example 4: Effect of the peak liquidation intensity \(d\) on the combined design (buffer \(a_0=-2\), \(a_s=13\), so \(b^{\ast}=13\) throughout).}
\label{tab:ex4d}
\end{table}

\begin{figure}[htb]
\centering
\includegraphics[width=0.62\linewidth]{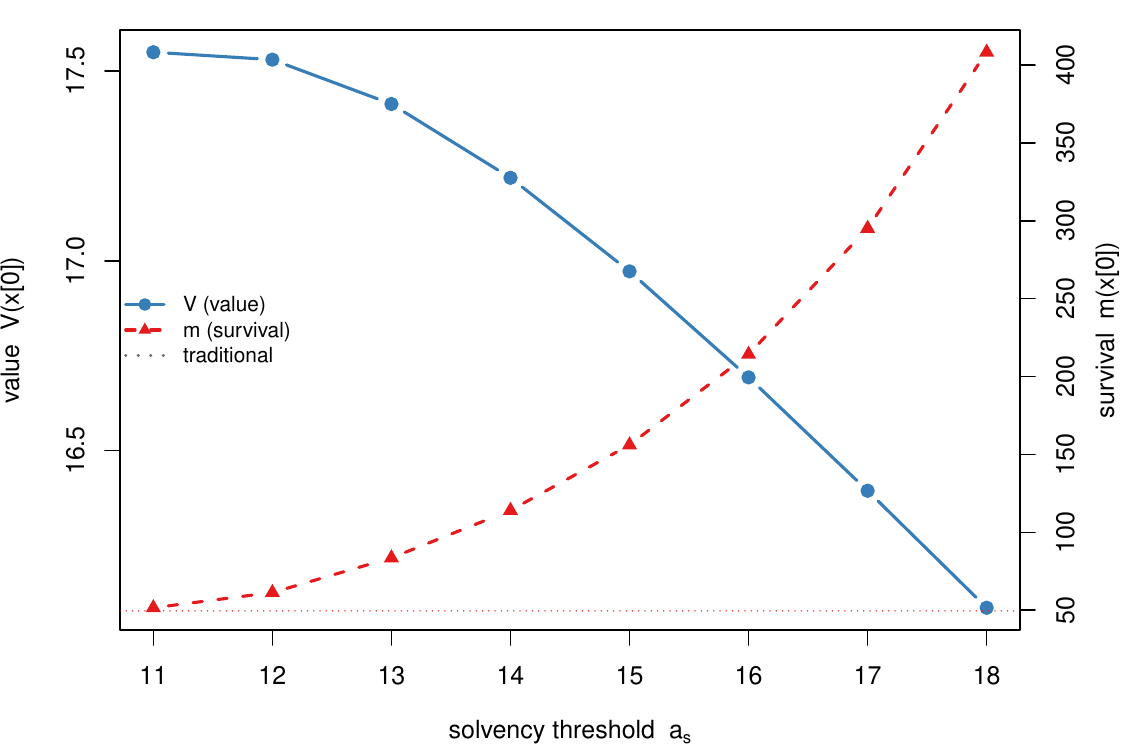}
\caption{Example 4. Value \(V(x_0)\) (left axis) and expected survival \(m(x_0)\) (right axis) at \(x_0=5\) as the solvency threshold \(a_s\) varies; dotted lines mark the traditional benchmark. Tighter regulation trades value for stability.}
\label{fig:ex4as}
\end{figure}

\begin{figure}[htb]
\centering
\includegraphics[width=0.62\linewidth]{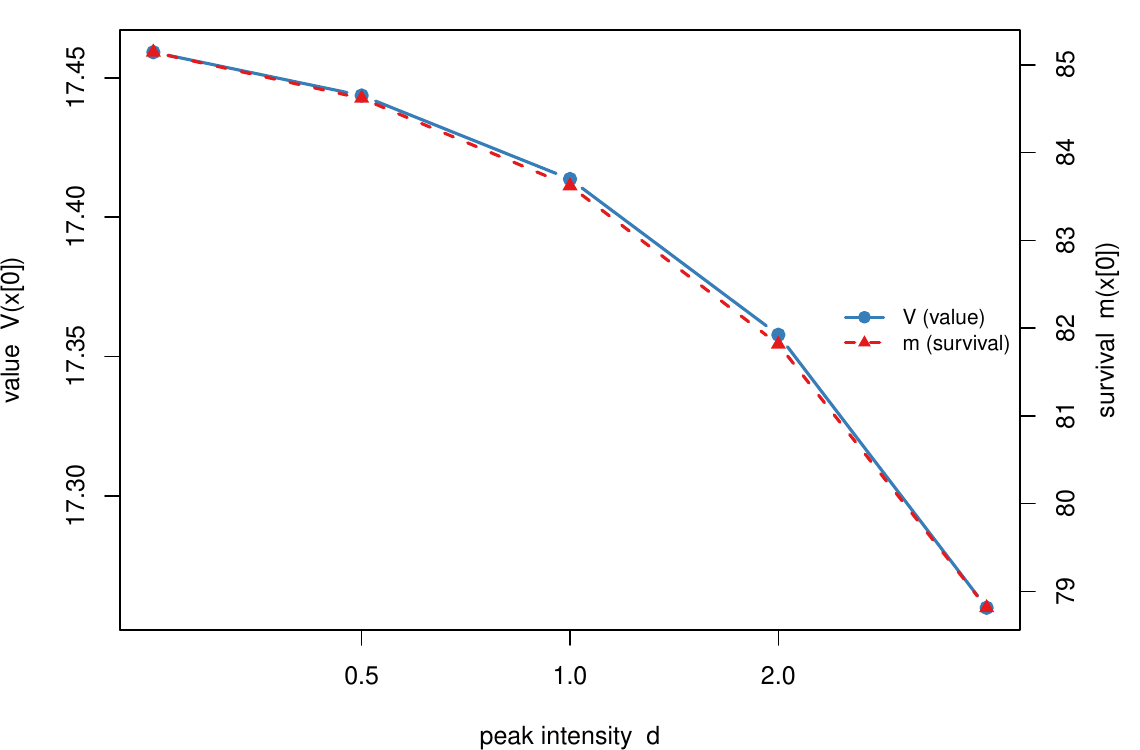}
\caption{Example 4. Value and expected survival at \(x_0=5\) as the peak intensity \(d\) varies (log scale).}
\label{fig:ex4d}
\end{figure}

\subsubsection{Non-increasing versus constant intensity.} 
Finally we compare the non-increasing intensity \eqref{eq:omega-decr} with the constant intensity \(\omega \equiv d\) of the earlier literature, at the same peak \(d = 1\) and the same thresholds. Because the constant form applies the full intensity across the
entire buffer---whereas \eqref{eq:omega-decr} relaxes it as the firm recovers---it imposes a heavier average liquidation pressure and yields both a lower value and shorter survival: \(V(5) = 17.25\) and \(m(5) = 78.54\) under the constant intensity, against \(V(5) = 17.41\) and \(m(5) = 83.62\) under the non-increasing one (both with \(b^{\ast}=13\)). The non-increasing specification is thus not merely a generalisation but a materially more
favourable one, and the constant-intensity results of prior work are recovered as a special case. Figure~\ref{fig:ex4omega} shows the two intensity profiles and the resulting value functions.

\begin{figure}[H]
\centering
\includegraphics[width=0.92\linewidth]{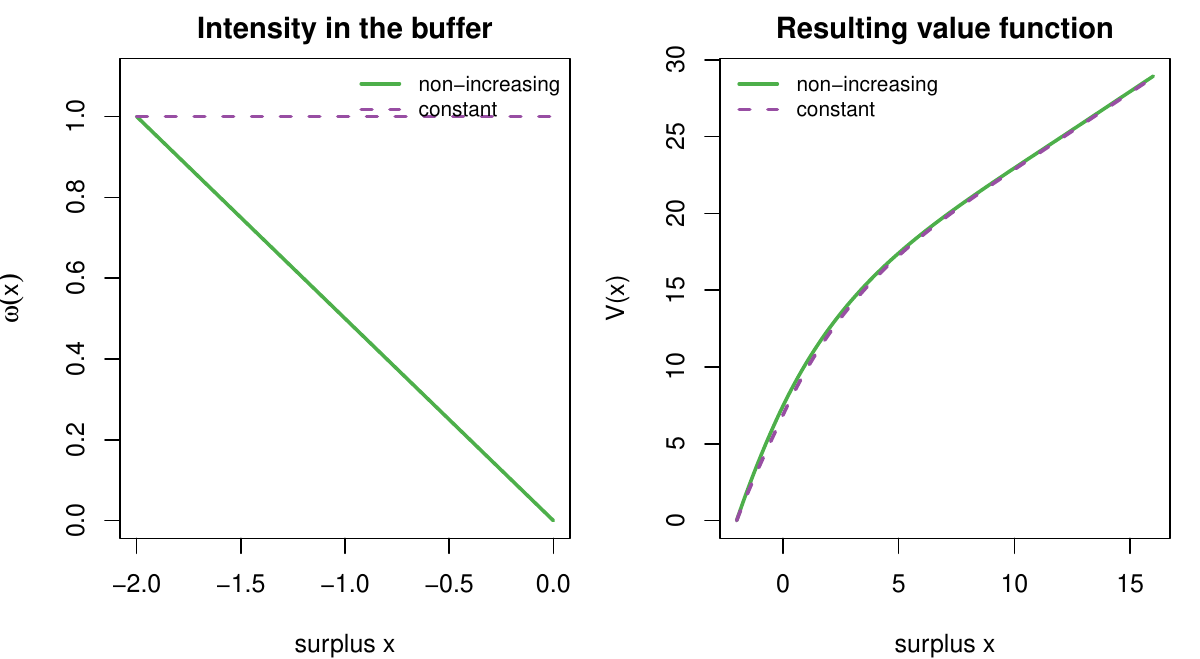}
\caption{Example 4. \textit{Left}: the non-increasing intensity \eqref{eq:omega-decr} versus a constant intensity of the same peak \(d=1\) on the buffer. \textit{Right}: the resulting value functions.}
\label{fig:ex4omega}
\end{figure}

\subsection{Example 5: a simulation study of the liquidation-time and dividend distributions} \label{num:simulation}

The expected survival time \(m(x)\) is only the first moment of the time to liquidation under the optimal strategy. To examine the full distribution---and the \textit{joint} distribution with the dividends actually paid---we simulate the controlled
surplus under the optimal barrier strategy of each of the four designs of Example~\ref{num:design}. For every sample path we record both the liquidation time \(T^{\ast}\) (the first time the controlled surplus reaches \(a_0\) under the optimal barrier strategy \(D^{b^{\ast}}\)) and the realised discounted dividends $\mathrm{PV} = \int_0^{T^{\ast}} e^{-\delta }\,\mathrm{d}D_s^{b^{\ast}}.$

The surplus is discretised by an Euler--Maruyama scheme (step \(\Delta t = 0.025\), horizon \(T_{\max}=650\), \(N = 10{,}000\) paths), reflecting at the barrier \(b^{\ast}\) and absorbing at \(a_0\). Liquidation in the distress zone is generated by the usual inverse-hazard construction: each path draws an independent \(\mathrm{Exp}(1)\) threshold and is liquidated when its accumulated hazard \(\int_0^t \omega(R_s)\,\mathrm{d}s\) (with \(\omega\) the
intensity \eqref{eq:omega-decr}) first reaches that threshold. A Brownian-bridge correction accounts for crossings of the absorbing boundary \(a_0\) between grid points; see Appendix \ref{app:bridge}. Crucially, the four designs are driven by \\textit{common random numbers}: each design sees the same Brownian increments and the same hazard budgets, so sample path \(i\) yields a \textit{matched} pair of outcomes under every design. This makes the design-versus-design scatterplots below genuine paired comparisons rather than comparisons of unrelated marginals. 

Table~\ref{tab:ex5} summarises the simulated distributions. As a validation, the mean discounted dividends reproduce the analytic value function almost exactly (e.g.\ \(17.41\) versus \(V(5)=17.41\) for the combined design), and the mean liquidation time tracks the closed-form \(m(5)\) up to a small, uniform discretisation bias (about \(12\%\) across all designs at this step size, shrinking roughly linearly as \(\Delta t\) is reduced; being common to all designs we believe it does not affect the comparisons below in a material way).

\begin{table}[H]
\centering
\begin{tabular}{l c c c c c c c}
\hline
\\[-1.5ex]
Design & \(\overline{T^*}\) & \(m(5)\) & median \(T^*\) &
  \(q_{0.05}\) & \(q_{0.95}\) & \(\mathbb{P}(T^*<10)\) & \(\overline{\mathrm{PV}}\) \\
\hline
Traditional     & 55.7  & 49.6  & 34.0 & 0.6 & 186.4 & 0.291 & 15.34 \\
Buffer-only     & 58.1  & 51.6  & 39.1 & 1.4 & 180.4 & 0.205 & 17.59 \\
Regulation-only & 126.7 & 113.7 & 75.0 & 0.6 & 450.0 & 0.257 & 15.00 \\
Combined        & 93.8  & 83.6  & 62.9 & 1.4 & 292.2 & 0.174 & 17.41 \\
\hline
\end{tabular}
\caption{Example 5. Simulated liquidation-time distribution and mean discounted dividends for an initial surplus \(x_0=5\) (\(N=10{,}000\) paths, common random numbers). Here \(\overline{T^*}\) denotes the sample mean liquidation time and \(\overline{\mathrm{PV}}\) the sample mean discounted dividends, while \(m(5)\) is the analytic mean survival time; the median, the quantiles \(q_{0.05},q_{0.95}\) and \(\mathbb{P}(T^*<10)\) are likewise sample estimates from the \(N\) paths. Censoring at \(T_{\max}=650\) was below \(2\%\) for every design.}
\label{tab:ex5}
\end{table}

Three features stand out. First, the distributions are strongly right-skewed: the median liquidation time is well below the mean for every design (e.g.\ \(62.9\) versus \(93.8\) for the combined design), so the long expected survival times are driven by a minority of very long-lived paths. Second, the combined design has the \textit{lowest} probability of early failure: the simulated frequency \(\mathbb{P}(T^*<10) = 0.174\), against \(0.291\) for the traditional model. Thus, the combined design reduces the chance of an early liquidation by roughly two fifths. Third, the ranking of the survival curves (Figure~\ref{fig:ex5surv}) matches the analytic results, with regulation-only and combined designs stochastically dominating the traditional and buffer-only designs.

The paired scatterplots make the design comparison concrete at the level of individual scenarios. Figure~\ref{fig:ex5pair} compares the combined design with the traditional model under identical trajectories. On survival {time} (left), almost every point lies above the \(45^{\circ}\) line: scenario by scenario, the combined design survives at least as long as the traditional model and very often much longer, with a dense cluster on the diagonal at small \(T^{{*}}\) where a common early adverse trajectory liquidates both. On discounted dividends (right), the two are tightly aligned along the diagonal---the combined design's stability gains come at \\textit{no significant} systematic dividend cost relative to the traditional model. Figure~\ref{fig:ex5pairbc} repeats the paired comparison for the two designs that carry almost identical shareholder value but very different stability: buffer-only and combined. On dividends (right) the points sit just under the diagonal, confirming that the two pay almost the same amount scenario by scenario (but with a cost to additional safety); on survival (left) the combined points lie overwhelmingly above it. The combined design thus buys a large amount of additional survival for a negligible give-up in dividends---exactly the message of Example~\ref{num:design}, now at the level of individual scenarios.

\begin{figure}[h]
\centering
\includegraphics[width=0.72\linewidth]{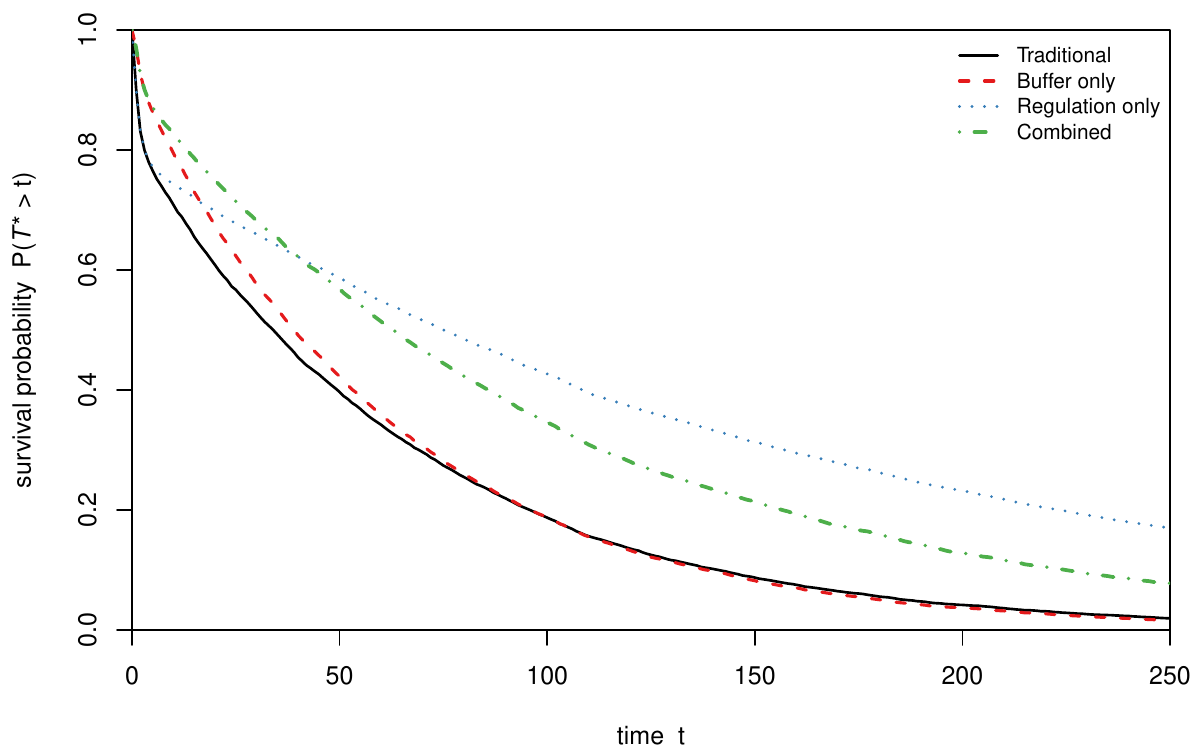}
\caption{Example 5. Empirical survival functions \(\mathbb{P}(T^*>t)\) of the time to liquidation under each design, based on \(N=10{,}000\) simulated paths of the surplus process with an initial surplus of  \(x_0=5\).}
\label{fig:ex5surv}
\end{figure}

\begin{figure}[H]
\centering
\includegraphics[width=0.49\linewidth]{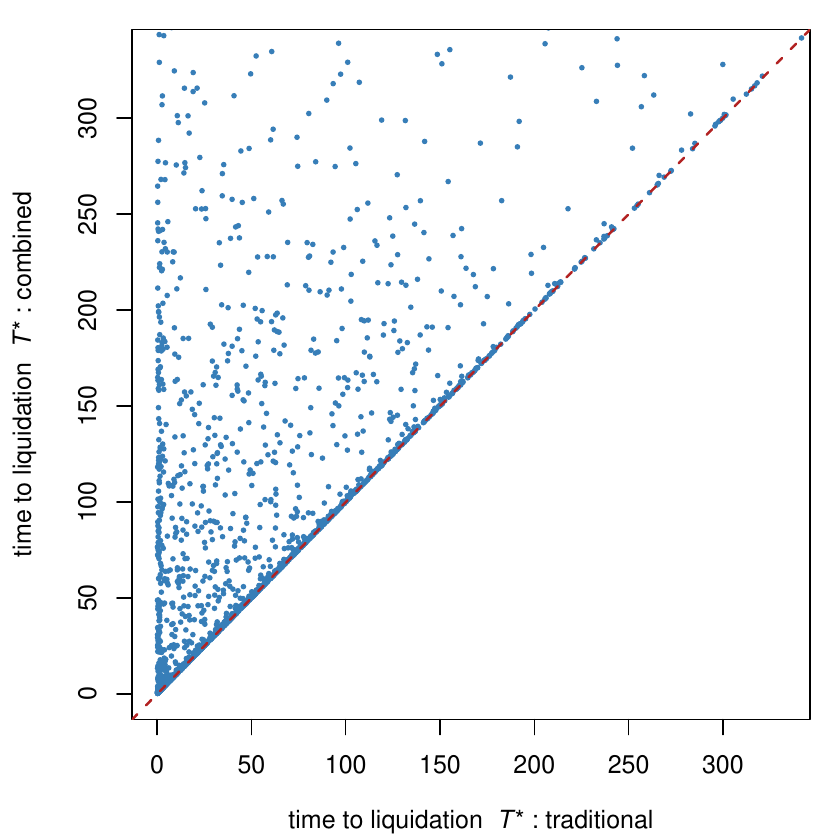}\hfill
\includegraphics[width=0.49\linewidth]{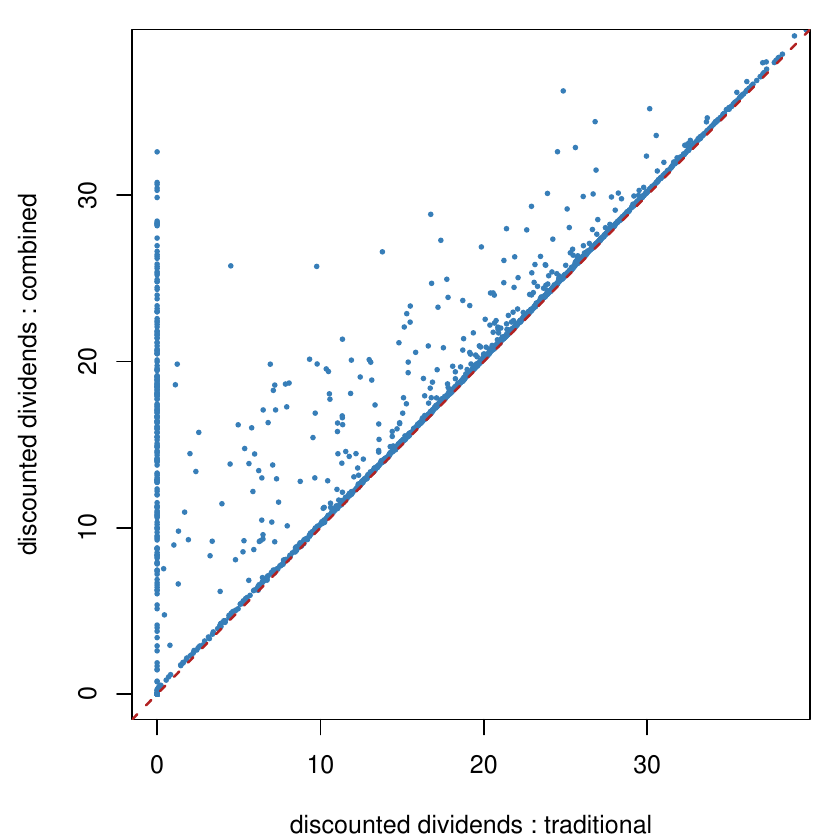}
\caption{Example 5. Paired outcomes for the combined design versus the traditional model under common random numbers, with one point representing  per scenario (subsampled for legibility). The dashed line denotes the \(45^{\circ}\) diagonal line.
\textit{Left}: liquidation time  \(T^*\). \textit{Right}: discounted dividends \(\mathrm{PV}\).}
\label{fig:ex5pair}
\end{figure}

\begin{figure}[H]
\centering
\includegraphics[width=0.49\linewidth]{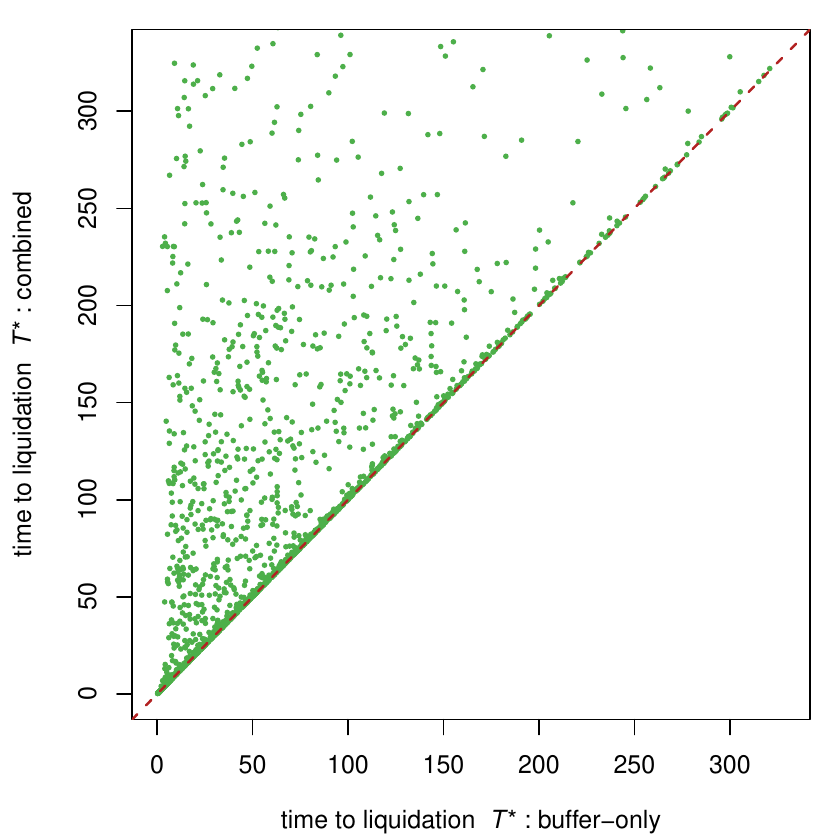}\hfill
\includegraphics[width=0.49\linewidth]{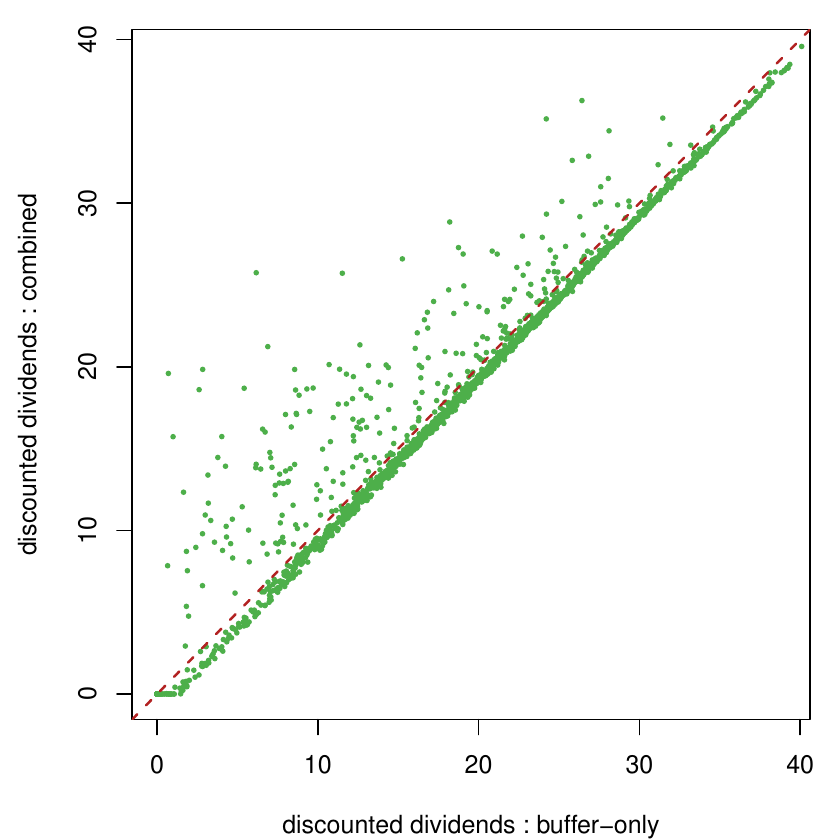}
\caption{Example 5. Paired outcomes for the combined design versus the buffer-only
design under common random numbers, with one point representing each scenario (subsampled for
legibility).  The dashed line represents the \(45^{\circ}\) diagonal line. \textit{Left}:
liquidation time \(T^*\). \textit{Right}: discounted dividends \(\mathrm{PV}\). The two
designs pay almost identically (right), yet the combined design survives
substantially longer (left).}
\label{fig:ex5pairbc}
\end{figure}

\section{Conclusion}\label{conclusion}

This paper studies the trade-off between shareholder value and financial stability when a financially distressed firm is wound up not at a single hard threshold but through a graduated, state-dependent process. Motivated by the resolution and prudential regimes that let distressed firms keep operating before liquidation (e.g., court-supervised reorganisation, supervisory forbearance, and distribution restrictions on weakly capitalised firms) we model forced liquidation through a state-dependent intensity and embed it in a general diffusion with three regions (distress, regulated, and unregulated) and region-dependent dynamics. Shareholder value is measured by the maximal expected discounted dividends and financial stability by the expected survival time under the corresponding optimal policy, which together define a singular stochastic control problem.

Our central finding is that the design of such a regime is decisive. Allowing a distressed firm to keep operating below the point of insolvency (a recovery buffer below the classical threshold) raises shareholder value but need not improve survival, and can in fact shorten the life of a well-capitalised firm; restricting distributions while reserves are low but still adequate (a regulatory constraint above the threshold) lengthens survival only by depressing value. Neither intervention, on its own, improves both objectives. Combining the two, however, can raise shareholder value and expected survival simultaneously: a Pareto improvement over immediate liquidation at insolvency. What matters is therefore not whether such a feature exists, but how its components are placed relative to the insolvency point and combined.

Along the way we uncover a qualitative feature absent from the classical problem. The shareholder-value function, although concave in the unregulated region, may be convex (or a mix of convex and concave) below the solvency threshold. This non-concavity arises from the interaction between liquidation risk and the suspension of dividends in distress, and has no counterpart in the standard diffusion dividend problem. A Monte-Carlo study, comparing designs under common random numbers, confirms the analytic comparisons scenario by scenario, and shows that the combined design markedly lowers the probability of early failure while paying dividends comparable to a pure recovery buffer.

Methodologically, the paper extends the optimal-dividend literature to a class of singular stochastic control problems with state-dependent liquidation intensity, region-dependent surplus dynamics, and solvency constraints. We establish a verification theorem, prove optimality of a barrier strategy, characterise the
optimal barrier, and give a constructive procedure (a small system of second-order linear ODEs) for computing both the shareholder-value function and a closed-form expected survival time, so that alternative regulatory designs can be compared directly and at low cost.

Several extensions remain open. The liquidation intensity could be endogenised by linking it to observable firm characteristics, financing decisions, or the strategic interaction of shareholders, creditors, and regulators; the model could accommodate capital injections, transaction costs, debt financing, or alternative payout mechanisms; and the distress-resolution mechanism could be enriched with reorganisation features and stochastic grace periods, or alternative notions of delayed liquidation such as Parisian ruin. We hope the framework developed here offers a useful foundation for further study of how financial stability, regulatory design, and shareholder value interact.

\section*{Code}
The \textsf{R} code reproducing all numerical results and figures in
Section~\ref{sec:numerical} is publicly available at
\url{https://github.com/agi-lab/balancing-value-stability}.

\section*{Use of AI tools}
The authors used large language models (ChatGPT, OpenAI; Claude, Anthropic) to assist with the drafting and refining of portions of the exposition and the numerical illustrations (and their associated simulation appendix). All results, interpretations, and final wording are the authors’ own, and the authors take full responsibility for the content of the paper.

\section*{Acknowledgements}

This research was supported under Australian Research Council's Linkage (LP130100723, with funding partners Allianz Australia Insurance Ltd, Insurance Australia Group Ltd, and Suncorp Metway Ltd) and Discovery Project (DP200101859) funding schemes. Furthermore, Avanzi acknowledges support from a grant of the Natural Science and Engineering Research Council of Canada (project number RGPIN-2015-04975). The views expressed herein are those of the authors and are not necessarily those of the supporting organisations. 

The authors declare that they have no conflicts of interest.

\section*{References}

\begingroup
\small
\bibliographystyle{elsarticle-harv}
\bibliography{BibPool,libraries}
\endgroup

\newpage

\appendix

\section{Proofs}
\numberwithin{equation}{section}
\counterwithin{figure}{section}
\counterwithin{table}{section}

\noindent \textbf{Proof of Proposition \ref{alternative-value}}
Applying the tower property and conditioning on the  trajectory of the controlled surplus process $\{R_u^D\}_{u\ge 0}$,  we obtain,
\begin{align}
\mathcal{P}_D(x)
=&\mathbb{E}_x\bigg[\int_0^{T^D}e^{-\delta t}dD_t\bigg]\nonumber\\
=&\mathbb{E}_x\bigg[\Big(\int_0^{T^D}e^{-\delta t}dD_t\Big) \mathbbm{1}_ {\{T^D<\tau^D_{{a_0}\}}}\bigg]
+\mathbb{E}_x\bigg[\Big(\int_0^{T^D}e^{-\delta t}dD_t\Big)
\mathbbm{1}_{\{ T^D\ge \tau^D_{{a_0}}\}}\bigg] \nonumber\\
=&\mathbb{E}_x\bigg[\mathbb{E}_x\bigg[\left.\Big(\int_0^{T^D}e^{-\delta t}dD_t\Big)
   \mathbbm{1}_{\{ T^D<\tau_{{a_0}}^D\}}\right|\{R_u^D\}_{u\ge 0}\bigg]\bigg]
   +\mathbb{E}_x\bigg[\mathbb{E}_x\bigg[\left.\Big(\int_0^{T^D}e^{-\delta t}dD_t\Big)
   \mathbbm{1}_{\{ T^D\ge \tau_{{a_0}}^D\}}\right|\{R_u^D\}_{u\ge 0}\bigg]\bigg]\nonumber\\
=&\mathbb{E}_x\bigg[
\int_0^{\tau_{{a_0}}^D}\Big(\int_0^{s}e^{-\delta t}dD_t\Big)
P(\left.T^D\in ds\right|\{R_u^D\}_{u\ge 0})\bigg] 
+\mathbb{E}_x\bigg[\Big(\int_0^{\tau_{{a_0}}^D}e^{-\delta t}dD_t\Big)P(\left. T^D\ge \tau_{{a_0}}^D\right|\{R_u^D\}_{u\ge 0})\bigg],\label{alt-p1}
    \end{align}
where  the last equality follows by noticing $\tau_{{a_0}}^D\in \sigma(R_u^D: u\ge 0)$ and $T^D\le \tau^D_{{a_0}}$.
Since the  bankruptcy intensity at time $s$ is
$\omega(R_s^D)$ conditional on the surplus path,
\begin{align}
&P\left(
T^D\in ds
\,\middle|\,
\{R_u^D\}_{u\ge0}
\right)
=
\omega(R_s^D)
e^{-\int_0^s\omega(R_u^D)\,du}\,ds.
\\
&P\left(
T^D>t
\,\middle|\,
\{R_u^D\}_{u\ge0}
\right)
=
\exp\!\left(
-\int_0^t\omega(R_u^D)\,du
\right),
\qquad
t<\tau_{a_0}^D.
\end{align}
Substituting this expression  gives
\begin{align}
&\mathbb{E}_x\bigg[
\int_0^{\tau_{{a_0}}^D}\Big(\int_0^{s}e^{-\delta t}dD_t\Big)
P(\left.T^D\in ds\right|\{R_u^D\}_{u\ge 0})\bigg]\nonumber\\
=&\mathbb{E}_x\bigg[
\int_0^{\tau_{{a_0}}^D}\Big(\int_0^{s}e^{-\delta t}dD_t\Big)
\omega(R_s^D)
    e^{-\int_0^s\omega(R_u^D)du}ds\bigg]=\mathbb{E}_x\bigg[
\int_0^{\tau_{{a_0}}^D}
e^{-\delta t}
\Big(\int_{t}^{\tau_{{a_0}}^D}\omega(R_s^D)
    e^{-\int_0^s\omega(R_u^D)du}ds\Big)dD_t\bigg]\nonumber\\
   =&\mathbb{E}_x\bigg[
\int_0^{\tau_{{a_0}}^D}
e^{-\delta t} \bigg(e^{-\int_0^t\omega(R_u^D)du}
    -e^{-\int_0^{\tau^D_{{a_0}}}\omega(R_u^D)
    du}\bigg)dD_t\bigg],
   \\    &\mathbb{E}_x\bigg[\Big(\int_0^{\tau_{{a_0}}^D}e^{-\delta t}dD_t\Big)P(\left. T^D\ge \tau_{{a_0}}^D\right|\{R_u^D\}_{u\ge 0})\bigg]=\mathbb{E}_x\bigg[\Big(\int_0^{\tau_{{a_0}}^D}e^{-\delta t}dD_t\Big)e^{-\int_0^{\tau^D_{{a_0}}}\omega(R_u^D)
    du}\bigg].\label{alt-p3}
    \end{align}
    Combining \eqref{alt-p1}-\eqref{alt-p3} yields the following  equivalent representation of the  performance function:
    $
    \mathcal{P}_D(x)
   =\mathbb{E}_x\Big[
\int_0^{\tau_{{a_0}}^D}
e^{-(\delta t+\int_0^t\omega(R_u^D)du)}dD_t\Big]$ for $x> {a_0}$.
     Accordingly, the value function admits the representation
    $
    V(x)
   =\sup_{D\in\Pi}\mathbb{E}_x\Big[
\int_0^{\tau_{{a_0}}^D}
e^{-(\delta t+\int_0^t\omega(R_u^D)du)}dD_t\Big]$ for $x> {a_0}$. \hfill $\square$.\\

\noindent \textbf{Proof of Lemma \ref{lemma6.1}}
Note
\begin{align*}    \mathbb{E}_{x}\Big[e^{-\delta T^D}{f}(R_{T^D}^D); T^D<\tau_{{a_0}}^D\Big]
   &=\mathbb{E}_x\Big[\mathbb{E}_x\Big[e^{-\delta T^D}{f}(R_{T^D}^D)
   \mathbbm{1}_{\{ T^D<\tau_{{a_0}}^D\}}|\{R_u^D\}_{u\ge 0}\Big]\Big]\\
    &=\mathbb{E}_x\bigg[\int_0^{\tau_{{a_0}}^D}
    e^{-\delta s}{f}(R_s^D)\mathbb{P}_x(T^D\in ds|\{R_u^D\}_{u\ge 0})\bigg]\\
    &=\mathbb{E}_x\bigg[\int_0^{\tau_{{a_0}}^D}
    e^{-\delta s}{f}(R_s^D)\omega(R_s^D)
    e^{-\int_0^s\omega(R_t^D)dt}ds\bigg],
\end{align*}
where the second to the last equality follows by noticing $\tau_{{a_0}}^D\in \sigma(R_u^D: u\ge 0)$, and the last equality follows by noticing that the hazard rate of the bankruptcy at time $s$ is $\omega(y)$ given $R_s^D=y$. \hfill $\square$\\

\noindent \textbf{Proof of  Theorem \ref{verification-ext}} 
Define the operator $\mathcal{A}$ by
\begin{equation}\label{operator-ext}
\mathcal{A}(f)(x)=
\begin{cases}
\frac{\sigma_d^2(x)}2f^{\prime\prime}(x)+\mu_d(x) f^{\prime}(x)-(\delta+\omega(x))f(x), &{a_0}<x<{a_d},\\
\frac{\sigma_s^2(x)}2f^{\prime\prime}(x)+\mu_s(x) f^{\prime}(x)-\delta f(x), &{a_d}<x<{a_s},\\
\frac{\sigma^2(x)}2f^{\prime\prime}(x)+\mu(x) f^{\prime}(x)-\delta f(x), &x\ge{a_s}.
\end{cases}
\end{equation}
It follows by \eqref{eq1-ext}-\eqref{eq4-ext} that
\begin{align}
\mathcal{A}(f)(x)\le 0\ \mbox{ for $x\in ({a_0},{a_d})\cup({a_d},{a_s})\cup({a_s}, +\infty)$.}\label{Af-ext}
\end{align}
Recall that $D$ is a right continuous stochastic process with left limits and so we can decompose $D$ as: $dD_t=dD^c_t+\Delta D_t$ where $\Delta D_t:=D_t-D_{t-}$ and $\{D^c_t:t\ge 0\}$ represents the continuous part of $D$. By applying the generalized It\^o formula we can obtain
\begin{align}
&e^{-(\int_0^{t\wedge \tau^D_{{a_0}}}(\delta+\omega(R_u^D))du)}{f}(R_{t\wedge \tau^D_{{a_0}}}^D) \nonumber\\
= &{f}(R_{0-})+\int_0^{t\wedge \tau^D_{{a_0}}}e^{-(\int_0^{s}(\delta+\omega(R_u^D))du)}\mathcal{A}({f})(R_{s-}^D)ds +\int_0^{t\wedge \tau^D_{{a_0}}}e^{-(\int_0^{s}(\delta+\omega(R_u^D))du)}\sigma(R_{s-}^D) {f}^{\prime}(R_{s-}^D)dW_s\nonumber\\
&-\int_0^{t\wedge \tau^D_{{a_0}}}e^{-(\int_0^{s}(\delta+\omega(R_u^D))du)}{f}^{\prime}(R_{s-}^D)dD^c_s +\sum_{0<s\le t\wedge \tau^D_{{a_0}}}e^{-(\int_0^{s}(\delta+\omega(R_u^D))du)}\left({f}(R_s^D)-{f}(R_{s-}^D)\right) \nonumber\\
\leq&{f}(R_{0-})+\int_0^{t\wedge \tau^D_{{a_0}}}e^{-(\int_0^{s}(\delta+\omega(R_u^D))du)}\sigma(R_{s-}^D) {f}^{\prime}(R_{s-}^D)dW_s-\int_0^{t\wedge \tau^D_{{a_0}}}e^{-(\int_0^{s}(\delta+\omega(R_u^D))du)}{f}^{\prime}(R_{s-}^D)dD^c_s \nonumber\\
&+\sum_{0<s\le t\wedge \tau^D_{{a_0}}}e^{-(\int_0^{s}(\delta+\omega(R_u^D))du)}\left({f}(R_s^D)-{f}(R_{s-}^D)\right).  \label{int form 2-ext}
\end{align}
where the  last inequality follows by \eqref{Af-ext}.

Taking expectations on both sides of \eqref{int form 2-ext} yields
\begin{align}
&\mathbb{E}_{x}\left[e^{-(\int_0^{t\wedge \tau^D_{{a_0}}}(\delta+\omega(R_u^D))du)}{f}(R_{t\wedge \tau^D_{{a_0}}}^D)\right] \nonumber\\
\le&{f}(x)+\mathbb{E}_{x}\left[\int_0^{t\wedge \tau^D_{{a_0}}}e^{-(\int_0^{s}(\delta+\omega(R_u^D))du)}\sigma {f}^{\prime}(R_{s-}^D)dW_s\right] 
-\mathbb{E}_{x}\left[\int_0^{t\wedge \tau^D_{{a_0}}}e^{-(\int_0^{s}(\delta+\omega(R_u^D))du)}{f}^{\prime}(R_{s-}^D)dD^c_s\right]\nonumber\\
&+\mathbb{E}_{x}\bigg[\sum_{0<s\le t\wedge \tau^D_{{a_0}}} e^{-(\int_0^{s}(\delta+\omega(R_u^D))du)}\left({f}(R_s^D)-{f}(R_{s-}^D)\right)\bigg]. \label{expectation-ext}
\end{align}
Note that $\int_0^t e^{-(\int_0^{s}(\delta+\omega(R_u^D))du)} \sigma {f}^{\prime}(R_{s-}^D)dW_s$ is a local martingale. So there exists a positive sequence $\{t_n\}$ with $\lim\limits_{n\rightarrow +\infty} t_n=+\infty$ such that
\begin{gather} \label{martingale-ext}
    \mathbb{E}_{x}\bigg[\int_0^{t_n\wedge t\wedge \tau^D_{{a_0}}}e^{-(\int_0^{s}(\delta+\omega(R_u^D))du)}\sigma{f}^{\prime}(R_{s-}^D)dW_s\bigg]=0.
\end{gather}
Notice that under any admissible strategy $D$, any jump of the controlled process $R^D$ is caused by a singular dividend payment only and the jump  (if any) is downward only and thus  $R_s^D\le R^D_{s-}$, $R^D_s\mathbbm{1}_{\{R^D_{s-}\in({a_0},{a_d})\}}
=R^D_{s-}\mathbbm{1}_{\{R^D_{s-}\in({a_0},{a_d})\}}$. Furthermore, there are no dividend payments whenever the controlled process is below $a_s$ and so $  \mathbbm{1}_{\{R^D_{s-}\in({a_0},{a_d})\}}dD^c_s=0$. Hence,
\begin{align}
&\sum_{0<s\le t_{n}\wedge t\wedge \tau^D_{{a_0}}}e^{-(\int_0^{s}(\delta+\omega(R_u^D))du)}\left({f}(R_s^{D})-{f}(R_{s-}^{D})\right) \nonumber\\
=&\sum_{0<s\le t_{n}\wedge t\wedge \tau^D_{{a_0}}} e^{-(\int_0^{s}(\delta+\omega(R_u^D))du)}\left({f}(R_s^{D})-{f}(R_{s-}^{D})\right)\mathbbm{1}_{\{R_{s-}^D\ge {a_d}\}} \nonumber\\
\leq &\sum_{0<s\le t_{n}\wedge t\wedge \tau^D_{{a_0}}} e^{-(\int_0^{s}(\delta+\omega(R_u^D))du)}\left(R_s^{D}-R_{s-}^{D}\right)\mathbbm{1}_{\{R_{s-}^D\ge {a_d}\}}\le 0, \label{inequality 1-ext}
\end{align}
where the last inequality follows by noting ${f}^{\prime}(x)\ge 1$ for any $x\ge {a_d}$ (by \eqref{eq3-ext}). Furthermore,
\begin{equation} \label{inequality 2-ext}
\int_0^{t_n\wedge t\wedge \tau^D_{{a_0}}} e^{-(\int_0^{s}(\delta+\omega(R_u^D))du)}{f}^{\prime}(R_{s-}^D)\,dD_s^c\ge \int_0^{t_n\wedge t\wedge \tau^D_{{a_0}}} e^{-(\int_0^{s}(\delta+\omega(R_u^D))du)}\,dD_s^c.
\end{equation}
Combining \eqref{expectation-ext}, \eqref{martingale-ext}, \eqref{inequality 1-ext}, and \eqref{inequality 2-ext}, we obtain
\begin{equation}
\begin{aligned}
{f}(x)\ge& \mathbb{E}_{x}\left[e^{-(\int_0^{t_n\wedge t\wedge \tau^D_{{a_0}}}(\delta+\omega(R_u^D))du)}{f}(R^D_{t_n\wedge t\wedge \tau^D_{{a_0}}})\right] +\mathbb{E}_{x}\bigg[\int_0^{t_n\wedge t\wedge \tau^D_{{a_0}}}e^{-(\int_0^{s}(\delta+\omega(R_u^D))du)}dD_s\bigg].
\end{aligned}\label{2110-0-ext}
\end{equation}
Note $f^\prime(x)\ge 1$ for all $x>{a_0}$ and so $f(x)>f({a_0})=0$. Thus, by  Fatou's lemma we can obtain
\begin{align}
&\liminf_{n\rightarrow +\infty} \liminf_{t\rightarrow +\infty} \mathbb{E}_{x}\left[e^{-(\int_0^{t_n\wedge t\wedge \tau^D_{{a_0}}}(\delta+\omega(R_u^D))du)}{f}(R^D_{t_n\wedge t\wedge \tau^D_{{a_0}}})\right] 
\ge  \mathbb{E}_{x}\left[e^{-(\int_0^{ \tau^D_{{a_0}}}(\delta+\omega(R_u^D))du)}{f}(R^D_{ \tau^D_{{a_0}}})\right]=0,\label{2110-1-ext}
\end{align}
where the last equality follows by noticing $R^D_{ \tau^D_{{a_0}}}={a_0}$ and $f({a_0})=0$.
Note that $D$ is non-decreasing and so by  using the monotone convergence we can obtain,
\begin{equation} \label{2110-2-ext}
   \lim \limits_{\substack{ n\rightarrow +\infty \\ t\rightarrow +\infty}}\mathbb{E}_x\Big[\int_0^{t_n\wedge t\wedge \tau^D_{{a_0}}}e^{-(\int_0^{s}(\delta+\omega(R_u^D))du)}dD_s\Big]
   =\mathbb{E}_x\Big[\int_0^{\tau^D_{{a_0}}}e^{-(\int_0^{s}(\delta+\omega(R_u^D))du)}dD_s\Big].
\end{equation}
Combining \eqref{2110-0-ext}, \eqref{2110-1-ext} and \eqref{2110-2-ext} yields
$f(x)\ge \mathbb{E}_x\Big[\int_0^{\tau^D_{{a_0}}}e^{-(\int_0^{s}(\delta+\omega(R_u^D))du)}dD_s\Big],\quad x\ge {a_0}.$
It then follows by the arbitrariness of $D$ in $\Pi$ that
$f(x)\ge\sup_{D\in\Pi} \mathbb{E}_x\Big[\int_0^{\tau^D_{{a_0}}}e^{-(\int_0^{s}(\delta+\omega(R_u^D))du)}dD_s\Big]=V(x)$ for $x\ge {a_0},$
where the last equality follows by the definition of the value function $V$ (\eqref{alt-valueF}).
This completes the proof. \hfill $\square$\\

\noindent \textbf{Proof of  Lemma \ref{fb-sol-ext}}
From Definition \ref{fundamentalsol-ext} we know that $g_1$ is a solution to \eqref{diff1-ext}. Then for any constant $c_1$, $c_1g_1$ is also a solution to the equation \eqref{diff1-ext}.
From Definition \ref{fundamentalsol-ext} we know that $g_2$ and $g_3$ are solutions to \eqref{diff2-ext}, and $g_4$ and $g_5$ are solutions to \eqref{diff3-ext}. Computing the Wronskian of $g_2$ and $g_3$ we have
$W(g_2,g_3)({a_d})=g_2({a_d})\ g_3^{\prime}({a_d})-\ g_2^{\prime}({a_d})\ g_3({a_d})=1\neq 0,$ and thus
we conclude that $g_2(x)$ and $g_3(x)$ are two fundamental solutions to the equation \eqref{diff2-ext}. Hence, any solution to \eqref{diff2-ext} has the following general form:
$f(x)=c_2 g_2(x)+c_3 g_3(x), \quad x\in[{a_d},{a_s}),$
where $c_2$ and $c_3$ are constants. Similarly, $g_4(x)$ and $g_5(x)$ form a fundamental pair of solutions on $[a_s,b)$  to \eqref{diff3-ext} and any solution to this equation  is of the form $c_4 g_4(x)+c_5 g_5(x)$ for some constant $c_4$ and $c_5$.

Now define $c_1(b)$, $c_2(b)$, $c_3(b)$, $c_4(b)$ and $c_5(b)$  to be the quantities such that
\begin{align}
c_2(b) g_2({a_d})+c_3(b) g_3({a_d})&=c_1(b) g_1({a_d}),\label{cond1-ext} \quad 
c_2(b) g_2^{\prime}({a_d})+c_3(b) g_3^{\prime}({a_d})=c_1(b) g_1^{\prime}({a_d}),
\\
c_4(b)g_4(a_s)+c_5(b)g_5(a_s)&=c_2(b)g_2(a_s)+c_3(b)g_3(a_s),\label{cond3-ext}\\
c_4(b)g_4^{\prime}(a_s)+c_5(b)g_5^{\prime}(a_s)&=c_2(b)g_2^{\prime}(a_s)+c_3(b)g_3^{\prime}(a_s),\label{cond4-ext}\\
c_4(b) g_4^{\prime}(b)+c_5(b) g_5^{\prime}(b)&=1.\label{cond5-ext}
\end{align}
Define a function
\begin{align}
f_b(x)=\begin{cases}
c_1(b)g_1(x), & {a_0}\le x<{a_d},\\
c_2(b)g_2(x)+c_3(b)g_3(x), &{a_d}\le x<{a_s},\\
c_4(b)g_4(x)+c_5(b)g_5(x), &{a_s}\le x<{b}
\end{cases}
\end{align}
The conditions \eqref{cond1-ext}-\eqref{cond4-ext} ensure the continuity and differentiability  of $f_b$ at ${a_d}$ and ${a_s}$, respectively, and the condition \eqref{cond5-ext} guarantees the terminal condition  \eqref{diff4-ext}.\\
By substituting the initial values of $g_1,\ldots, g_5$ listed in Definition \ref{fundamentalsol-ext} in \eqref{cond1-ext}-\eqref{cond5-ext} we obtain
\begin{align} 
&c_2(b)=c_1(b) g_1({a_d}),
\qquad 
c_3(b)=c_1(b) g_1^{\prime}({a_d}),
\label{simult3-ext}
\qquad
c_4(b)=c_2(b)g_2(a_s)+c_3(b)g_3(a_s), \\
&c_5(b)=c_2(b)g_2^{\prime}(a_s)+c_3(b)g_3^{\prime}(a_s),
\label{simult5-ext} 
\qquad c_4(b) g_4^{\prime}(b)+c_5(b) g_5^{\prime}(b)=1,
\end{align}
which yields
\begin{align*}
    c_1(b)&=\frac{1}{(g_1(a_d)g_2(a_s)+g_1^{\prime}(a_d)g_3(a_s))g_4^{\prime}(b)
    +(g_1(a_d)g_2^{\prime}(a_s)+g_1^{\prime}(a_d)g_3^{\prime}(a_s))g_5^{\prime}(b)},\\
 c_2(b)&=g_1({a_d})c_1(b),\qquad
    c_3(b)=g_1^{\prime}({a_d})c_1(b),\\
    c_4(b)&=g_2(a_s)c_2(b)+g_3(a_s)c_3(b)=\left(g_2(a_s)g_1({a_d})+g_3(a_s)g_1^{\prime}({a_d})\right)c_1(b),\\
    c_5(b)&=g_2^{\prime}(a_s)c_2(b)+g_3^{\prime}(a_s)c_3(b)=\left(g_2^{\prime}(a_s)g_1({a_d})+g_3^{\prime}(a_s)g_1^{\prime}({a_d})\right)c_1(b).
    \end{align*}
We can verify that the  function $f_b$ defined above is a continuously differentiable solution to \eqref{diff1-ext}-\eqref{diff4-ext}.
 Note that there is a unique set of values for $c_1(b)$- $c_5(b)$ so that all \eqref{simult3-ext}-\eqref{simult5-ext} are satisfied. Thus, there exists a unique continuously differentiable solution to \eqref{diff1-ext}-\eqref{diff4-ext}.

It is clear that $f_b$ is twice continuously differentiable on $({a_0},{a_d})$, $({a_d},{a_s})$ and $({a_s},b)$.
From the way that $c_1(b)- c_5(b)$ were determined we can observe $f_{b}({a_d}-)=f_{b}({a_d}+)$, $f_{b}^\prime({a_d}-)=f_b^\prime({a_d}+)$, $f_{b}({a_s}-)=f_{b}({a_s}+)$ and $f_{b}^\prime({a_s}-)=f_b^\prime({a_s}+)$. This implies that $f_b$ is continuously differentiable at ${a_d}$ and ${a_s}$.
\hfill $\square$\\

\noindent \textbf{Proof of  Lemma \ref{fxbar-sol-ext}}
The derivations are similar to Lemma \ref{fb-sol-ext}.  We know that a solution to \eqref{diff2-0-ext} has the following general form: 
$f(x)=c_6g_1(x)$ for $x\in(a_0,a_d)$, and  any solution to \eqref{diff2-0-ext-1} has the following general form:
$f(x)=c_7 g_2(x)+c_8 g_3(x)$ for $x\in[{a_d},{a_s})$, 
where $c_6 - c_8$ are constants. 
Now define $c_6(b)$, $c_7(b)$, and $c_8(b)$  to be the quantities such that
\begin{align}
c_7({a_s}) g_2({a_d})+c_8({a_s}) g_3({a_d})&=c_6(b) g_1({a_d}),\label{cond1-ext-03}
\qquad
c_7({a_s}) g_2^{\prime}({a_d})+c_8({a_s}) g_3^{\prime}({a_d})=c_6({a_s}) g_1^{\prime}({a_d}),
\\
c_7({a_s}) g_2^{\prime}(a_s)+c_8({a_s}) g_3^{\prime}(a_s)&=1.\label{cond5-ext-03}
\end{align}
Define 
\begin{align}
f_{a_s}(x):=\begin{cases}
c_6({a_s})g_1(x), & {a_0}\le x<{a_d},\\
c_7({a_s})g_2(x)+c_8({a_s})g_3(x), &{a_d}\le x<{a_s},\\
\end{cases}
\end{align}
 By solving \eqref{cond1-ext-03}-\eqref{cond5-ext-03} we obtain \eqref{c7-ext}-\eqref{c8-ext}. 
We can see that $f_{a_s}$ indeed is a solution to \eqref{diff2-0-ext}-\eqref{diff3-0-ext}. 
The conditions \eqref{cond1-ext-03}-\eqref{cond5-ext-03} ensure the continuity and differentiability  of $f_{a_s}$ at ${a_d}$, and the condition \eqref{cond5-ext-03} guarantees the terminal condition  \eqref{diff3-0-ext}.
We can easily verify that $f_{{a_d}}$ defined above is the desired unique solution. \hfill $\square$\\

\noindent \textbf{Proof of  Theorem \ref{Vb-ext}}
From the definition of $V_b$ we can see immediately $V_b(x)=0$ for $x\le a_0$. 
Under the barrier strategy $D^b$, when the initial surplus is greater than $b$, a lump sum with an amount of $R_{0-}-b$ is paid out as dividends immediately, which reduces the surplus to $b$. This implies, $V_b(x)=V_b(b)+x-b$ for $x>b$. Hence it is sufficient to show that $V_b(x)=f_b(x)$ for ${a_0}\le x\le b$. From Lemma \ref{fb-sol-ext} we know $f_b({a_0})=0$. Noting that $V_b({a_0})=0$, we conclude $V_b({a_0})=f_b({a_0})$.

Now we proceed to show that $V_b(x)=f_b(x)$ for ${a_0}<x\le b$.
From Lemma \ref{fb-sol-ext} and  Lemma \ref{fxbar-sol-ext} we know that $f_b$ is a continuously  differentiable  and piecewise twice differentiable (more precisely, twice  differentiable in $({a_0},{a_d})$ and $({a_d}, b)$). From the special structure of the barrier strategy $D^{b}$ we can observe that given the initial value $R_{0-}\in ({a_0},b]$, the surplus process will always remain in $[{a_0},b]$ and so we can apply the generalized Ito's lemma to $e^{-(\int_0^{t\wedge \tau^{D^b}_{{a_0}}} (\delta+\omega(R_u^D))\,du)}f_b(R_{t\wedge \tau^{D^b}_{{a_0}}}^{D^{b}})$ given $R_{0-}=x\in ({a_0},b]$.

Noting that $dD^{b}_s=R_{s-}^{D^{b}}-R_s^{D^{b}}+dD^{b,c}_s$. Define the operator $\mathcal{A}$ in the same way as in \eqref{operator-ext}:
\begin{equation*}
\mathcal{A}(f)(x)=
\begin{cases}
\frac{\sigma_d^2(x)}2f^{\prime\prime}(x)+\mu_d(x) f^{\prime}(x)-(\delta+\omega(x))f(x), &{a_0}<x<{a_d},\\
\frac{\sigma_s^2(x)}2f^{\prime\prime}(x)+\mu_s(x) f^{\prime}(x)-\delta f(x), &{a_d}<x<{a_s},\\
\frac{\sigma_u^2(x)}2f^{\prime\prime}(x)+\mu_u(x) f^{\prime}(x)-\delta f(x), &x\ge{a_s}.
\end{cases}
\end{equation*}
From Lemma \ref{fb-sol-ext} we know that $\mathcal{A}(f_b)(x)=0$ for any $x \in({a_0},b)$. By applying the generalized Ito's lemma, we obtain
\begin{align}     \label{derive perf-ext}
  & \mathbb{E}_{x}\bigg[e^{-(\int_0^{t\wedge \tau^{D^b}_{{a_0}}}(\delta+\omega(R_u^{D^b}))du)}f_b(R_{t\wedge \tau^{D^b}_{{a_0}}}^{D^{b}})\bigg] \nonumber\\
   =&f_b(x)+\mathbb{E}_{x}\bigg[\int_0^{t\wedge \tau^{D^b}_{{a_0}}}e^{-(\delta s+\int_0^{s}\omega(R_u^{D^b})du)}\mathcal{A}(f_b)(R_{s-}^{D^{b}})ds\bigg]\nonumber\\
+& \mathbb{E}_{x}\bigg[\int_0^{t\wedge \tau^{D^b}_{{a_0}}}e^{-(\delta s+\int_0^{s}\omega(R_u^{D^b})du)}\sigma(R_{s-}^{D^b}){f_b}^{\prime}(R_{s-}^{D^b})dW_s\bigg]
-\mathbb{E}_{x}\bigg[\int_0^{t\wedge \tau^{D^b}_{{a_0}}}e^{-(\delta s+\int_0^{s}\omega(R_u^{D^b})du)}f_b^{\prime}(R_{s-}^{D^{b}})dD^{b}_s\bigg]\nonumber\\
   ={}&f_b(x)+\mathbb{E}_{x}\bigg[\int_0^{t\wedge \tau^{D^b}_{{a_0}}}e^{-(\delta s+\int_0^{s}\omega(R_u^{D^b})du)}\sigma(R_{s-}^{D^b}) {f_b}^{\prime}(R_{s-}^{D^b})dW_s\bigg]
    -\mathbb{E}_{x}\bigg[\int_0^{t\wedge \tau^{D^b}_{{a_0}}}e^{-(\delta s+\int_0^{s}\omega(R_u^{D^b})du)}f_b^{\prime}(R_{s-}^{D^{b}})dD^{b}_s\bigg].
\end{align}
Note that $\int_0^{t\wedge \tau^{D^b}_{{a_0}}}e^{-(\delta s+\int_0^{s}\omega(R_u^{D^b})du)}\sigma(R_{s-}^{D^b}) {f_b}^{\prime}(R_{s-}^{D^b})dW_s$ is a local martingale, so there exists a positive sequence $\{t_n\}$ with $\lim\limits_{n\rightarrow +\infty} t_n=+\infty$ such that
\begin{equation}\label{m2-ext}
\mathbb{E}_{x}\bigg[\int_0^{t_n\wedge t\wedge \tau^{D^b}_{{a_0}}}e^{-(\delta s+\int_0^{s}\omega(R_u^{D^b})du)}\sigma(R_{s-}^{D^b}) {f_b}^{\prime}(R_{s-}^{D^b})dW_s\bigg]=0.
\end{equation}
Under the barrier strategy $D^b$, the dividends are paid when the surplus reaches the predefined level $b$, such that $dD^{b}_s=\mathbbm{1}_{\{R_{s-}^{D^{b}}=b\}}dD^{b}_s$. By \eqref{diff3-ext}, we know that $f^{\prime}(R_{s-}^{D^{b}})\mathbbm{1}_{\{R_{s-}^{D^{b}}=b\}}=f^{\prime}(b)\mathbbm{1}_{\{R_{s-}^{D^{b}}=b\}}=\mathbbm{1}_{\{R_{s-}^{D^{b}}=b\}}$. Hence,
\begin{align}\label{e1-ext}
&\mathbb{E}_{x}\bigg[\int_0^{t_n\wedge t\wedge \tau^{D^b}_{{a_0}}}e^{-(\delta s
    +\int_0^{s}\omega(R_u^{D^b})du)}f^{\prime}(R_{s-}^{D^{b}})dD^{b}_s\bigg]
=\mathbb{E}_{x}\bigg[\int_0^{t_n\wedge t\wedge \tau^{D^b}_{{a_0}}}e^{-(\delta s+\int_0^{s}\omega(R_u^{D^b})du)}dD^{b}_s\bigg].
\end{align}
Combining \eqref{derive perf-ext}, \eqref{m2-ext} and \eqref{e1-ext}, we have
\begin{align}
f_b(x)&=\mathbb{E}_{x}\bigg[e^{-(\int_0^{t_n\wedge t\wedge \tau^{D^b}_{{a_0}}}(\delta+\omega(R_u^{D^b}))du)}f_b(R_{t_n\wedge t\wedge \tau^{D^b}_{{a_0}}}^{D^{b}})\bigg] +\mathbb{E}_{x}\bigg[\int_0^{t_n\wedge t\wedge \tau^{D^b}_{{a_0}}}e^{-(\delta s+\int_0^{s}\omega(R_u^{D^b})du)}dD^{b}_s\bigg].\label{23922-1-ext}
\end{align}
Note that $D^b$ is non-decreasing. So, by using the monotone convergence we can obtain
$$ \lim \limits_{\substack{ n\rightarrow +\infty \\ t\rightarrow +\infty}}\mathbb{E}_{x}\bigg[\int_0^{t_n\wedge t\wedge \tau^{D^b}_{{a_0}}}e^{-(\delta s+\int_0^{s}\omega(R_u^{D^b})du)}dD^{b}_s\bigg]
   =\mathbb{E}_{x}\bigg[\int_0^{\tau^{D^b}_{{a_0}}}e^{-(\delta s+\int_0^{s}\omega(R_u^{D^b})du)}dD^{b}_s\bigg].$$
Note that $R_{t_n\wedge t\wedge\tau^{D^b}_{{a_0}}}^{D^b}$ is always bounded above by $b$ and below by ${a_0}$ given that the initial reserve is greater than ${a_0}$ and that $\omega(\cdot)$ is non-negative, and so by applying the dominated convergence twice we can obtain
\begin{align*}
&\lim \limits_{\substack{ n\rightarrow +\infty \\ t\rightarrow +\infty}} \mathbb{E}_{x}\bigg[e^{-(\int_0^{t_n\wedge t}(\delta+\omega(R_u^{D^b}))du)}f_b(R_{t_n\wedge t}^{D^{b}})I\{t_n\wedge t< \tau^{D^b}_{{a_0}}\}\bigg]
=0, \ \ x\ge {a_0}.
\end{align*}
Note $R_{ \tau^{D^b}_{{a_0}}}^{D^{b}}={a_0}$ and hence, $f_b(R_{ \tau^{D^b}_{{a_0}}}^{D^{b}})=f_b({a_0})=0$. Consequently,
\begin{align}
&\lim \limits_{\substack{ n\rightarrow +\infty \\ t\rightarrow +\infty}} \mathbb{E}_{x}\bigg[e^{-(\int_0^{t_n\wedge t\wedge \tau^{D^b}_{{a_0}}}(\delta+\omega(R_u^{D^b}))du)}f_b(R_{t_n\wedge t\wedge \tau^{D^b}_{{a_0}}}^{D^{b}})\bigg]\nonumber\\
=&\lim \limits_{\substack{ n\rightarrow +\infty \\ t\rightarrow +\infty}} \bigg(\mathbb{E}_{x}\bigg[e^{-(\int_0^{t_n\wedge t}(\delta+\omega(R_u^{D^b}))du)}f_b(R_{t_n\wedge t}^{D^{b}})I\{t_n\wedge t< \tau^{D^b}_{{a_0}}\}\bigg]
\nonumber\\
&+\mathbb{E}_{x}\bigg[e^{-(\int_0^{\tau^{D^b}_{{a_0}}}(\delta+\omega(R_u^{D^b}))du)}
f_b(R_{\tau^{D^b}_{{a_0}}}^{D^{b}})I\{t_n\wedge t\ge  \tau^{D^b}_{{a_0}}\}\bigg]\bigg)=0.\label{23922-2-ext}
\end{align}
By combining \eqref{23922-1-ext} and \eqref{23922-2-ext} we arrive at
$
    f_b(x)=\mathbb{E}_{x}\bigg[\int_0^{\tau^{D^b}_{{a_0}}}e^{-(\delta s+\int_0^{s}\omega(R_u^{D^b})du)}dD^{b}_s\bigg]=\mathcal{P}_x(D^{b})=V_b(x),\ \ {a_0}<x\le b.
$

It follows by  Lemma \ref{fb-sol-ext} and Lemma \ref{fxbar-sol-ext} that
\begin{align}
f_b(b)=\begin{cases}
c_7(b)g_2(b)+ c_8(b)g_3(b)& b={a_s},\\
c_4(b)g_4(b)+ c_5(b)g_5(b)& b>{a_s},
\end{cases}\label{fb(b)-ext}
\end{align}
From the expressions for $c_1(b)-c_8(b)$ (see \eqref{c1b-ext}, \eqref{c5b-ext}, \eqref{c7-ext} and \eqref{c8-ext}),  we can observe that  all $C_i(b)$ $i=1,2,\ldots,8$ are continuous on $b\in [a_s,\infty)$ and 
\begin{align*}
\lim\limits_{b\downarrow {a_s}}c_1(b)&=\frac1{\left(g_1({a_d})g_2^{\prime}({a_s})+g_1^{\prime}({a_d})g_3^{\prime}({a_s})\right)}, \quad 
\lim\limits_{b\downarrow {a_s}}c_4(b)=\frac{g_2(a_s)g_1({a_d})+g_3(a_s)g_1^{\prime}({a_d})}{\left(g_1({a_d})g_2^{\prime}({a_s})+g_1^{\prime}({a_d})g_3^{\prime}({a_s})\right)}, \\
\lim\limits_{b\downarrow {a_s}}c_5(b)&=\frac{g_2^{\prime}(a_s)g_1({a_d})+g_3^{\prime}(a_s)g_1^{\prime}({a_d})}{\left(g_1({a_d})g_2^{\prime}({a_s})+g_1^{\prime}({a_d})g_3^{\prime}({a_s})\right)}.
\end{align*} 
Taking limits on  the second expression of \eqref{fb(b)-ext} we can obtain
\begin{align}
\lim\limits_{b\downarrow {a_s}}f_b(b)=\lim\limits_{b\downarrow {a_s}}c_4(b)g_4({a_s})+\lim\limits_{b\downarrow {a_s}}c_5(b)g_5({a_s})=\lim\limits_{b\downarrow {a_s}}c_4(b)=\frac{g_2(a_s)g_1({a_d})+g_3(a_s)g_1^{\prime}({a_d})}{\left(g_1({a_d})g_2^{\prime}({a_s})+g_1^{\prime}({a_d})g_3^{\prime}({a_s})\right)}=f_{a_s}(a_s),\label{21824-1-ext}
\end{align}
where the last equality follows by using the first expression 
of  \eqref{fb(b)-ext}.
Thus, $f_b(b)$ is  continuous in $b$ for $b> {a_s}$ and right continuous in $b$ for $b={a_s}$.
Notice by Theorem \ref{Vb-ext}, Lemma \ref{fb-sol-ext} and Lemma \ref{fxbar-sol-ext} we have
\begin{align}
V_b^{\prime\prime}(b-)&=f_b^{\prime\prime}(b)= \frac{2}{\sigma^2(b)}(\delta f_{b}(b)-\mu(b)f_b^{\prime}(b))=\frac{2}{\sigma^2(b)}(\delta f_{b}(b)-\mu(b)),\ \ b> {a_d},\label{Vbbdd-ext}
\end{align}
where the last equality follows by noting 
$f_b^\prime(b)=1$.
Hence, $V_b^{\prime\prime}(b-)$ is also continuous in $b$ for $b> {a_s}$.
This completes the proof.\hfill $\square$\\

\noindent \textbf{Proof of  Lemma \ref{finiteness}}
It is obvious from the definition of $b^{\ast}$  that $b^{\ast}\ge a_s$.
We use proof by contradiction to show $b^{\ast}<+\infty$. Suppose $b^{\ast}=+\infty$. Then by \eqref{b*-ext} it follows
\begin{align}
V_b^{\prime\prime}(b-)< 0 \ \mbox{for all $b>a_s$.}\label{Vbb--ext}
\end{align}
 Thus,
$
0> V_b^{\prime\prime}(b-)=\frac{2}{\sigma^2(b)}(\delta V_b(b)-\mu(b) )$ for $b> a_s$, 
where the last equality follows by  \eqref{Vbbdd-ext} and Theorem \ref{Vb-ext}. As a result,
$
\frac{V_b(b)}{b}< \frac{\mu(b)}{\delta b}<\frac{\mu(a_s)+ \sup_{y>a_s}\mu^\prime(y) b}{\delta b},\ \ b> a_s.
$
 Hence, by noting $\delta>\sup_{y>a_s}\mu^\prime(y) $, we can obtain
 \begin{align}
\limsup_{b\rightarrow +\infty}\frac{V_b(b)}{b}\le \frac{\sup_{y>a_s}\mu^\prime(y) }{\delta}<1 .\label{Vbb-bound-ext}
\end{align}
 On the other hand, by \eqref{Vbb--ext} it follows
$
 V_b^{\prime}(x)> V_b^{\prime}(b)=1,\ \ a_s<x<b, 
$ 
and hence,
$V_b(b)>V(a_s)+b-a_s\ge b-a_s$ for all $b\ge a_s$, which leads to
 $
\liminf_{b\rightarrow +\infty}\frac{V_b(b)}{b}\ge 1.
$
This is a contradiction to \eqref{Vbb-bound-ext}.\hfill $\square$\\

\noindent \textbf{Proof of  Lemma \ref{V-concavity-ext}}
From a previous result, we know $a_s<b^{*}<+\infty$. Thus, $V_{b^\ast}^{\prime\prime}(b^\ast-)=0$ (by Remark \ref{11123-1-ext}). 
Note from \eqref{vb-x+-ext} we can obtain $V_{b^\ast}(x)=1$ for $x\ge b^\ast$. It is, therefore, sufficient to show $V_{b^{*}}^{\prime\prime}(x)\le 0$ for $a_s\le x<b^{*}$. We use proof by contradiction.
Suppose there exists some $x_0$ on $(a_s,b^{*})$ such that
\begin{align}
V_{b^{*}}^{\prime\prime}(x_0)>0.\label{Vbx1-ext}
\end{align}
Define
\begin{align}
x_2&:=\sup\{x\in (x_0,b^{\ast}): V_{b^{*}}^{\prime\prime}(x)> 0\}.
\end{align}
Then,  obviously, 
$a_d\le x_0<x_2\le b^\ast$, and  by noting $V_{b^{*}}^{\prime\prime}(b^{*}-)= 0$, $ V_{b^{*}}^{\prime\prime}(x_0)>0$ and the continuity of $V_{b^\ast}^{\prime\prime}$ on $(a_s,b^{\ast})$
we can see
\begin{align}
&V_{b^{*}}^{\prime\prime}(x_{2}-)=0,
\quad 
V_{b^{*}}^{\prime\prime}(x)\le 0, \ x\in(x_{2}, b^{\ast}),\label{Vbx-3-ext}
\end{align}
Furthermore, since $V_{b^{*}}^{\prime\prime}(x_0)>0$ (\eqref{Vbx1-ext}), so the continuity  of $V_{b^{*}}^{\prime\prime}$ implies 
$V_{b^{*}}^{\prime\prime}(x)$ is positive in the left neighbourhood of $x_2$, that is, we can find an $x_1\in [x_0,x_2)$ such that
\begin{align}
&V_{b^{*}}^{\prime\prime}(x)>0, \ x\in(x_{1}, x_{2}).\label{Vbx-ext}\end{align}
Recall $V_{b^{\ast}}(x)=f_{b^{\ast}}(x)$ for $a_s< x< b^{\ast}$ (by Theorem \ref{Vb-ext}) and thus by Lemma \ref{fb-sol-ext} we know $V_{b^{\ast}}(x)$ satisfies
$
 \frac{1}{2}\sigma_u^2(x) {V}_{b^{\ast}}^{\prime\prime}(x)+\mu_u(x){V}_{b^{\ast}}^{\prime}(x)-\delta {V}_{b^{\ast}}(x)=0$ for $ x\in(a_s,b^{\ast}).
$ 
Thus,
\begin{align}
V_{b^{\ast}}^{\prime\prime}(x)&=\frac{2\left(\delta V_{b^{\ast}}(x)-\mu_u (x) V_{b^{\ast}}^{\prime}(x)\right)}{\sigma_u^2(x)}, \quad x\in(a_s,b^{\ast}).\label{diffV-2-ext}
\end{align}
By \eqref{Vbx3-ext} and \eqref{Vbx-ext} we know that for any $x\in (x_1,x_2)$,
\begin{align}
0<V_{b^{\ast}}^{\prime\prime}(x)-V_{b^{\ast}}^{\prime\prime}(x_2) =\frac{2}{\sigma_u^2(x)}\Big(\delta(V_{b^{\ast}}(x)-V_{b^{\ast}}(x_2))
-\mu_u(x)V_{b^{\ast}}^{\prime}(x)+\mu_u(x_2)V_{b^{\ast}}^{\prime}(x_2))\Big), \label{Vx-Vxk-ext}
\end{align}
where the last equality follows by \eqref{diffV-2-ext}.
By dividing both sides of \eqref{Vx-Vxk-ext} by $x-x_2$ and noticing $x-x_2<0$ for $x\in (x_1,x_2)$, we arrive at
\begin{align*}
&\delta\frac{V_{b^{\ast}}(x)-V_{b^{\ast}}(x_2)}{x-x_2}
-\frac{\mu_u(x)V_{b^{\ast}}^{\prime}(x)-\mu_u(x_2)V_{b^{\ast}}^{\prime}(x_2)}{x-x_2}<0.
\end{align*}
By letting  $x\uparrow x_2$  in the above inequality we can obtain
$
 (\delta-\mu_u^\prime(x_2)) V_{b^{\ast}}^{\prime}(x_2)-\mu_u(x_2)V_{b^{\ast}}^{\prime\prime}(x_2-)\le 0.
$ 
which along with $V_{b^{\ast}}^{\prime\prime}(x_2-)=0$ (see \eqref{Vbx3-ext}) and $\sup_x\mu_u^\prime(\cdot)<\delta$  implies
\begin{align}
V_{b^{\ast}}^{\prime}(x_2)\le 0.\label{negative-ext}
\end{align}
 However, by \eqref{Vbx-3-ext} we know $V_{b^{\ast}}^{\prime}(x_2)\ge V_{b^{\ast}}^{\prime}(b^{\ast})=1$, which contradicts \eqref{negative-ext}. This completes the proof. \hfill $\square$\\

\noindent \textbf{Proof of  Theorem \ref{maintheorem-ext}}
Note from Remark \ref{11123-1-ext} we know that
\begin{equation}\label{Verif-1-ext}
V_{b^{\ast}}^{\prime\prime}(b^{\ast}-)\ge 0.
\end{equation}
Recall
$\frac{1}{2}\sigma_u^2(x) {V}_{b^{\ast}}^{\prime\prime}(x)+\mu_u(x){V}_{b^{\ast}}^{\prime}(x)-\delta {V}_{b^{\ast}}(x)=0$ for $x\in({a_s},b^{\ast}).$ 
By letting $x\uparrow b^{\ast}$ and noting ${V}_{b^{\ast}}^{\prime}(b^{\ast})=1$,
\begin{equation}\label{Verif-2-ext}
\frac{1}{2}\sigma_u^2(x) {V}_{b^{\ast}}^{\prime\prime}(b^{\ast}-)+\mu_u(b^{\ast})-\delta {V}_{b^{\ast}}(b^{\ast}-)=0.
\end{equation}
Combining \eqref{Verif-1-ext} and \eqref{Verif-2-ext} yields
$\mu_u(b^{\ast})-\delta {V}_{b^{\ast}}(b^{\ast}-)\le 0,$
which implies
 $V_{b^{\ast}}(b^{\ast}-)\ge \frac{\mu_u(b^{\ast})}{\delta}.$
Recall $V_{b^{\ast}}(x)=V_{b^{\ast}}(b^{\ast})+x-b^{\ast}$, and so ${V}_{b^{\ast}}^{\prime}(x)=1$ and ${V}_{b^{\ast}}^{\prime\prime}(x)=0$ for $x>b^{\ast}$. Thus,
\begin{align}
&\frac{1}{2}\sigma_u^2(x) {V}_{b^{\ast}}^{\prime\prime}(x)+\mu_u(x){V}_{b^{\ast}}^{\prime}(x)-\delta {V}_{b^{\ast}}(x)=\mu_u(x)-\delta\left(V_{b^{\ast}}(b^{\ast})+x-b^{\ast}\right) \le \mu_u(x)-\delta\left(\frac{\mu_u(b^{\ast})}{\delta}+x-b^{\ast}\right) \nonumber\\
\le& (\sup_{y>a_s}\mu_u^\prime(y)-\delta)(x-b^{\ast})<0, \  x>b^{\ast},\label{Verif-3-ext}
\end{align}
where the last inequality is due to $\sup_{y> a_s}\mu_u^\prime(y)<\delta$. 
Combining \eqref{Verif-3-ext} with $V_{b^{\ast}}^\prime(x)=1$ for $x\ge b^{\ast}$ (by noting $V_{b^{\ast}}(x)=V_{b^{\ast}}+x-b^{\ast}$ for $x>b^{\ast}$) yields
\begin{align}
\max\left\{\frac{1}{2}\sigma_u^2(x) {V}_{b^{\ast}}^{\prime\prime}(x)+\mu_u(x){V}_{b^{\ast}}^{\prime}(x)-\delta {V}_{b^{\ast}}(x),1-V_{b^{\ast}}^\prime(x)\right\}=0. \label{23226-1}
    \end{align}
By Theorem \ref{Vb-ext} we have $V_{b^\ast}(x)=f_{b^\ast}(x)$ for ${a_0}\le x\le b^{\ast}$, and thus by Lemma \ref{fb-sol-ext} and Lemma \ref{fxbar-sol-ext} we can obtain
\begin{align}
&\frac{1}{2}\sigma_d^2(x) {V}_{b^{\ast}}^{\prime\prime}(x)+\mu_d(x){V}_{b^{\ast}}^{\prime}(x)-(\delta+\omega(x)) {V}_{b^{\ast}}(x)=0,\quad {a_0}<x<{a_d}, \label{a1-2-ext00000}\\
&\frac{1}{2}\sigma_s^2(x) {V}_{b^{\ast}}^{\prime\prime}(x)+\mu_s(x){V}_{b^{\ast}}^{\prime}(x)-\delta {V}_{b^{\ast}}(x)=0,\quad {a_d}<x<a_s, \label{a1-2-ext}\\
&\frac{1}{2}\sigma_u^2(x) {V}_{b^{\ast}}^{\prime\prime}(x)+\mu_u(x){V}_{b^{\ast}}^{\prime}(x)-\delta {V}_{b^{\ast}}(x)=0,\quad {a_s}<x<b^{\ast},\label{a1-3-ext}\\
&{V}_{b^{\ast}}({a_0})=0.\label{a1-4-ext}
\end{align}

Let us distinguish two cases (a) $b^\ast>a_s$
and (b) $b^\ast=a_s$. 

\noindent (a) We consider the case $b^\ast>a_s
$ here.  From Lemma \ref{V-concavity-ext}, we have
\begin{align}
{V}_{b^{\ast}}^{\prime}(x)\ge 1 \mbox{ for $ {a_s}\le x<b^{\ast}$, and } 
{V}_{b^{\ast}}^{\prime}(x)=1 \mbox{ for } x\ge b^{\ast}.\label{Verif-4-ext}
\end{align}
By combining 
\eqref{a1-3-ext}
and \eqref{Verif-4-ext} we can obtain
\begin{align}
    &\max\left\{\frac{1}{2}\sigma_u^2(x) {V}_{b^{\ast}}^{\prime\prime}(x)+\mu_u(x){V}_{b^{\ast}}^{\prime}(x)-\delta{V}_{b^{\ast}}(x),1-{V}_{b^{\ast}}^{\prime}(x)\right\}
        =0,\quad {a_s}<x<b^{\ast}.\label{a1-7-ext}
\end{align}
By combining \eqref{a1-2-ext00000}-\eqref{a1-4-ext}, \eqref{a1-7-ext} and \eqref{23226-1} we can verify that  $V_{b^{\ast}}(x)$ is a function that satisfies all the conditions in Theorem \ref{verification-ext} and as a result,
$
V_{b^{\ast}}(x)\ge V(x).
$

\noindent
(b) We now consider the case $b^\ast=a_s$.
Note $
{V}_{b^{\ast}}^{\prime}(x)=1$ for $x\ge b^{\ast}$ (Lemma \ref{V-concavity-ext}). So,
\begin{align}
{V}_{b^{\ast}}^{\prime}(x)=1,\quad x\ge a_s,\label{a1-0-ext}
\end{align}
which together with \eqref{a1-3-ext} implies 
\begin{align}
&\max\left\{\frac{1}{2}\sigma_u^2(x) {V}_{b^{\ast}}^{\prime\prime}(x)+\mu_u(x){V}_{b^{\ast}}^{\prime}(x)-\delta{V}_{b^{\ast}}(x),1-{V}_{b^{\ast}}^{\prime}(x)\right\}=0,\quad x>{a_s},\label{23226-2}
\end{align}
By
 combining \eqref{a1-2-ext}-\eqref{a1-3-ext}, \eqref{a1-0-ext} and \eqref{23226-2}  we can conclude that $V_{b^{\ast}}(x)$ is a function that satisfies all the conditions in Theorem \ref{verification-ext} and as a result,
$
V_{b^{\ast}}(x)\ge V(x).
$

In conclusion, whether $b^\ast>a_s$ or $b^\ast=a_s$, we always observe $
V_{b^{\ast}}(x)\ge V(x),\quad x\ge {a_0}.
$ 
On the other hand,
$
V_{b^{\ast}}(x)\coloneqq \mathcal{P}_x(D^{b^{\ast}})\le \sup_{D\in\Pi} \mathcal{P}_x(D)=V(x),\quad x\ge {a_0}.
$ 
Therefore, 
$V_{b^{\ast}}(x)= \mathcal{P}_x(D^{b^{\ast}})=V(x)$ and so, the strategy $D^{b^{\ast}}$ is optimal.\hfill $\square$\\

\noindent \textbf{Proof of  Theorem \ref{verification-ext-new}}
The existence and uniqueness of the stated solution can be shown by finding a solution and show that it is unique. This will be done later when we proceed to solve the solution.

Let $f$ represents the solution. 
Now we proceed to show that $m(x)=f(x)$ for ${a_0}<x\le b^\ast$.
Note that the stochastic process
$R_t^{D^{b^\ast}}$ will always remain in $[{a_0},b^\ast]$ for $t\in[0,T^{D^{b^\ast}}]$, if it starts with a  initial value  $R_{0-}\in [{a_0},b^\ast]$. 
Noting that $dD^{b^\ast}_s=R_{s-}^{D^{b^\ast}}-R_s^{D^{b^\ast}}+dD^{{b^\ast},c}_s$ where $D^{{b^\ast},c}$ is the continuous part of the stochastic process $D^{{b^\ast}}$. Define an operator $\mathcal{A}_0$:
\begin{equation*}
\mathcal{A}_0(f)(x)=\frac{\sigma^2(x)}2f^{\prime\prime}(x)+\mu(x) f^{\prime}(x)-\omega(x)f(x)=
\begin{cases}
\frac{\sigma_d^2(x)}2f^{\prime\prime}(x)+\mu_d(x) f^{\prime}(x)-\omega(x)f(x), &{a_0}<x<{a_d},\\
\frac{\sigma_s^2(x)}2f^{\prime\prime}(x)+\mu_s(x) f^{\prime}(x), &{a_d}<x<{a_s},\\
\frac{\sigma_u^2(x)}2f^{\prime\prime}(x)+\mu_u(x) f^{\prime}(x), &x\ge{a_s}.
\end{cases}
\end{equation*}
Since $f$ satisfies \eqref{eq1-00-ext} and \eqref{eq2-00-ext}, we have $\mathcal{A}_0(f)(x)=-1$ for  $x \in({a_0},b^\ast]$. By applying the generalized It\^o's lemma, we derive that for $x\in[a_0,b^\ast]$,

\begin{align}     \label{derive perf-00-ext}
  & \mathbb{E}_{x}\bigg[e^{-(\int_0^{t\wedge \tau^{D^{b^\ast}}_{{a_0}}}\omega(R_u^{D^{b^\ast}})du)}f(R_{t\wedge \tau^{D^{b^\ast}}_{{a_0}}}^{D^{b}})\bigg] \nonumber\\
=&f(x)+\mathbb{E}_{x}\bigg[\int_0^{t\wedge \tau^{D^{b^\ast}}_{a_0}}e^{-\int_0^{s}\omega(R_u^{D^{b^\ast}})du}\mathcal{A}_0(f)(R_s^{D^{{b^\ast}}})ds\bigg]
+ \mathbb{E}_{x}\bigg[\int_0^{t\wedge \tau^{D^{b^\ast}}_{{a_0}}}e^{-\int_0^{s}\omega(R_u^{D^{b^\ast}})du}
\sigma(R_{s-}^{D^{b^\ast}}){f}^{\prime}(R_{s-}^{D^{b^\ast}}dW_s\bigg]\nonumber\\
&-\mathbb{E}_{x}\bigg[\int_0^{t\wedge \tau^{D^{b^\ast}}_{{a_0}}}e^{-\int_0^{s}\omega(R_u^{D^{b^\ast}})du}f^{\prime}(R_{s-}^{D^{{b^\ast}}})dD^{{b^\ast}}_s\bigg]\nonumber\\
=&f(x)-\mathbb{E}_{x}\bigg[\int_0^{t\wedge \tau^{D^{b^\ast}}_{{a_0}}}e^{-\int_0^{s}\omega(R_u^{D^b})du)}ds\bigg]+\mathbb{E}_{x}\bigg[\int_0^{t\wedge \tau^{D^{b^\ast}}_{{a_0}}}e^{-\int_0^{s}\omega(R_u^{D^{b^\ast}})du}
\sigma(R_{s-}^{D^{b^\ast}}){f}^{\prime}(R_{s-}^{D^{b^\ast}}dW_s\bigg]\nonumber\\
   &-\mathbb{E}_{x}\bigg[\int_0^{t\wedge \tau^{D^{b^\ast}}_{{a_0}}}e^{-\int_0^{s}\omega(R_u^{D^{b^\ast}})du}f^{\prime}(R_{s-}^{D^{{b^\ast}}})dD^{{b^\ast}}_s\bigg].
\end{align}

Note that $\int_0^{t\wedge \tau^{D^{b^\ast}}_{{a_0}}}e^{-\int_0^{s}\omega(R_u^{D^{b^\ast}})du}
\sigma(R_{s-}^{D^{b^\ast}}){f}^{\prime}(R_{s-}^{D^{b^\ast}})dW_s$ is a local martingale, so there exists a positive sequence $\{t_n\}$ with $\lim\limits_{n\rightarrow +\infty} t_n=+\infty$ such that
\begin{equation}\label{m2-00-ext}
\mathbb{E}_{x}\bigg[\int_0^{t_n\wedge t\wedge \tau^{D^{b^\ast}}_{{a_0}}}e^{-\int_0^{s}\omega(R_u^{D^{b^\ast}})du}\sigma(R_{s-}^{D^{b^\ast}}){f}^{\prime}(R_{s-}^{D^{b^\ast}})dW_s\bigg]=0.
\end{equation}
Under the barrier strategy $D^{b^{\ast}}$, the dividends are paid when the surplus reaches the predefined level ${b^\ast}$, such that $dD^{{b^\ast}}_s=\mathbbm{1}_{\{R_{s-}^{D^{{b^\ast}}}={b^\ast}\}}dD^{{b^\ast}}_s$. Since $f^\prime(b^\ast)=0$ and so $f^{\prime}(R_{s-}^{D^{{b^\ast}}})\mathbbm{1}_{\{R_{s-}^{D^{{b^\ast}}}={b^\ast}\}}=0$, we can show
\begin{align}\label{e1-00-ext}
&\mathbb{E}_{x}\bigg[\int_0^{t_n\wedge t\wedge \tau^{D^{b^\ast}}_{{a_0}}}e^{- \int_0^{s}\omega(R_u^{D^{b^\ast}})du}f^{\prime}(R_{s-}^{D^{{b^\ast}}})dD^{{b^\ast}}_s\bigg]
=0.
\end{align}
Combining \eqref{derive perf-00-ext}, \eqref{m2-00-ext} and \eqref{e1-00-ext}, we have
\begin{align}
f(x)&=\mathbb{E}_{x}\bigg[e^{-(\int_0^{t_n\wedge t\wedge \tau^{D^{b^\ast}}_{{a_0}}}\omega(R_u^{D^{b^\ast}}))du)}f(R_{t_n\wedge t\wedge \tau^{D^{b^\ast}}_{{a_0}}}^{D^{{b^\ast}}})\bigg]+\mathbb{E}_{x}\bigg[\int_0^{t_n\wedge t\wedge \tau^{D^{b^\ast}}_{{a_0}}}e^{-\int_0^{s}\omega(R_u^{D^{b^\ast}})du)}ds\bigg], \quad x \in({a_0},b^\ast].\label{23922-1-00-ext}
\end{align}
Note that $R_{t_n\wedge t\wedge\tau^{D^{b^\ast}}_{{a_0}}}^{D^{b^\ast}}$ is always bounded in the interval $[a_0,b^\ast]$ given that the initial reserve is in the same interval and that $\omega(\cdot)$ is non-negative, and so by applying the dominated convergence twice we can obtain
\begin{align}
&\lim \limits_{\substack{n\rightarrow +\infty \\ t\rightarrow +\infty}} \mathbb{E}_{x}\bigg[e^{-(\int_0^{t_n\wedge t}\omega(R_u^{D^{b^\ast}}))du}f(R_{t_n\wedge t}^{D^{{b^\ast}}})I\{t_n\wedge t< \tau^{D^{b^\ast}}_{{a_0}}\}\bigg]
=0, \quad x\in [ {a_0},b^\ast], \label{25226-1}
\end{align}
where the last equality follows by noting $\lim \limits_{\substack{ n\rightarrow +\infty \\ t\rightarrow +\infty}}e^{-\int_0^{t_n\wedge t}\omega(R_u^{D^{b^\ast}})du} =0$ and  $f(R_{t_n\wedge t}^{D^{{b^\ast}}})I\{t_n\wedge t< \tau^{D^{b^\ast}}_{{a_0}}\}$ is bounded. 
Note $R_{ \tau^{D^{b^\ast}}_{{a_0}}}^{D^{{b^\ast}}}={a_0}$ and hence, 
$
    f(R_{ \tau^{D^{b^\ast}}_{{a_0}}}^{D^{{b^\ast}}})=f({a_0})=0.
$
Consequently,
\begin{align}
&\lim \limits_{\substack{ n\rightarrow +\infty \\ t\rightarrow +\infty}} \mathbb{E}_{x}\bigg[e^{-(\int_0^{t_n\wedge t\wedge \tau^{D^{b^\ast}}_{{a_0}}}\omega(R_u^{D^{b^\ast}}))du}f(R_{t_n\wedge t\wedge \tau^{D^{b^\ast}}_{{a_0}}}^{D^{{b^\ast}}})\bigg]\nonumber\\
=&\lim \limits_{\substack{ n\rightarrow +\infty \\ t\rightarrow +\infty}} \bigg(\mathbb{E}_{x}\bigg[e^{-\int_0^{t_n\wedge t}\omega(R_u^{D^{b^\ast}})du}f(R_{t_n\wedge t}^{D^{{b^\ast}}})I\{t_n\wedge t< \tau^{D^{b^\ast}}_{{a_0}}\}\bigg]\nonumber\\
&+\mathbb{E}_{x}\bigg[e^{-\int_0^{\tau^{D^{b^\ast}}_{a_0}}\omega(R_u^{D^{b^\ast}})du}
f(R_{\tau^{D^{b^\ast}}_{{a_0}}}^{D^{{b^\ast}}})I\{t_n\wedge t\ge  \tau^{D^{b^\ast}}_{{a_0}}\}\bigg]\bigg)\nonumber\\
=&0, \quad x \in({a_0},b^\ast].\label{23922-2-00-ext}
\end{align}
By applying the monotone convergence twice we can obtain
\begin{align}
\lim \limits_{\substack{ n\rightarrow +\infty \\ t\rightarrow +\infty}} 
\mathbb{E}_{x}\bigg[\int_0^{t_n\wedge t\wedge \tau^{D^{b^\ast}}_{{a_0}}}e^{-\int_0^{s}\omega(R_u^{D^{b^\ast}})du}ds\bigg]=\mathbb{E}_{x}\bigg[\int_0^{\tau^{D^{b^\ast}}_{{a_0}}}e^{-\int_0^{s}\omega(R_u^{D^{b^\ast}})du}ds\bigg].\label{161123-1-00-ext}
\end{align}
By combining \eqref{23922-1-00-ext} and \eqref{23922-2-00-ext} and \eqref{161123-1-00-ext} we have 
$    f(x)=\mathbb{E}_{x}\bigg[\int_0^{\tau^{D^{b^\ast}}_{{a_0}}}e^{-\int_0^{s}\omega(R_u^{D^{b^\ast}})du}ds\bigg]=m(x)$ for $ {a_0}<x\le {b^\ast}.
$

For the stochastic process starting with an initial value greater than $b^\ast$, under the barrier strategy $D^{\ast}$, the value of the controlled surplus $R^{D^{b^\ast}}$ at time $0$ becomes $b^\ast$  immediately due to a dividend payment with an amount, $R^{D^{b^\ast}}_0-b^\ast$, and hence,
\begin{align*}
    m(x)=&\mathbb{E}_{x}\bigg[\int_0^{\tau^{D^{b^\ast}}_{{a_0}}}e^{-\int_0^{s}\omega(R_u^{D^{b^\ast}})du}ds\bigg]=\mathbb{E}_{b^\ast}\bigg[\int_0^{\tau^{D^{b^\ast}}_{{a_0}}}e^{-\int_0^{s}\omega(R_u^{D^{b^\ast}})du}ds\bigg]=m(b^\ast), \quad x>b. \qquad  \square
\end{align*}


\noindent \textbf{Proof of Proposition \ref{survival-representation-prop}}
Conditioning on the controlled surplus path
$\{R_u^{D^{b^\ast}}\}_{u\ge0}$,
and decomposing according to whether liquidation occurs before the process reaches $a_0$, we obtain
\begin{align}
m(x)
=&\mathbb{E}_x\bigg[
T^{D^{b^{\ast}}} \mathbbm{1}_{\{T^{D^{b^{\ast}}}<\tau^{D^{b^{\ast}}}_{a_0}\}}\bigg]
+\mathbb{E}_x\bigg[T^{D^{b^{\ast}}}
\mathbbm{1}_{ \{T^{D^{b^{\ast}}}\ge \tau^{D^{b^{\ast}}}_{a_0}\}}\bigg]\nonumber\\
=&\mathbb{E}_x\bigg[\mathbb{E}_x\bigg[\left.T^{D^{b^{\ast}}}
   \mathbbm{1}_{\{ T^{D^{b^{\ast}}}<\tau_{{a_0}}^{D^{b^{\ast}}}\}}\right|\{R_u^{D^{b^{\ast}}}\}_{u\ge 0}\bigg]\bigg]
   +\mathbb{E}_x\bigg[\mathbb{E}_x\bigg[\left. T^{D^{b^{\ast}}}
   \mathbbm{1}_{\{ T^{D^{b^{\ast}}}\ge\tau_{{a_0}}^{D^{b^{\ast}}}\}}\right|{R_u^{D^{b^{\ast}}}}_{u\ge 0}\bigg]\bigg]\nonumber\\
=&\mathbb{E}_x\bigg[\int_0^{\tau_{a_0}^{D^{b^{\ast}}}} s\ 
P(\left.T^{D^{b^{\ast}}}\in ds\right|\{R_u^{D^{b^{\ast}}}\}_{u\ge 0})\bigg]+\mathbb{E}_x\bigg[\tau_{{a_0}}^{D^{b^{\ast}}}P( \left.T^{D^{b^{\ast}}}\ge \tau_{{a_0}}^{D^{b^{\ast}}}\right|\{R_u^{D^{b^{\ast}}}\}_{u\ge 0})\bigg],\quad x\ge a_0,\label{alt-p1-ext}
    \end{align}
where  the last equality follows by noticing $\tau_{{a_0}}^{D^{b^{\ast}}}\in \mathbf{\sigma}(R_u^{D^{b^{\ast}}}: u\ge 0)$ on  $\{T^{D^{b^{\ast}}}<\tau^{D^{b^{\ast}}}_{{a_0}}\}$ and $T^{D^{b^{\ast}}}=\tau^{D^{b^{\ast}}}_{{a_0}} $ on  $\{T^{D^{b^{\ast}}}\ge \tau^{D^{b^{\ast}}}_{{a_0}}\}$.

Since, given the path of the controlled surplus $R^{D^{b^{\ast}}}$,  the conditional liquidation intensity  at time $s$ is $\omega(R_s^{D^{b^{\ast}}})$,
\begin{align}
&\mathbb{E}_x\bigg[\int_0^{\tau_{a_0}^{D^{b^{\ast}}}}s\ 
P(\left. T^{D^{b^{\ast}}}\in ds\right|\{R_u^{D^{b^{\ast}}}\}_{u\ge 0})\bigg]\nonumber\\
=&\mathbb{E}_x\bigg[\int_0^{\tau_{a_0}^{D^{b^{\ast}}}}
s\ 
\omega(R_s^{D^{b^{\ast}}})
    e^{-\int_0^s\omega(R_u^{D^{b^{\ast}}})du}ds\bigg]=\mathbb{E}_x\bigg[\int_0^{\tau_{a_0}^{D^{b^{\ast}}}}
\left(\int_0^s dy\right)\ 
\omega(R_s^{D^{b^{\ast}}})
    e^{-\int_0^s\omega(R_u^{D^{b^{\ast}}})du}ds\bigg]\nonumber\\
    =& \mathbb{E}_x\bigg[\int_0^{\tau_{a_0}^{D^{b^{\ast}}}} \left(
\int_y^{\tau_{a_0}^{D^{b^{\ast}}}} \ 
\omega(R_s^{D^{b^{\ast}}})
    e^{-\int_0^s\omega(R_u^{D^{b^{\ast}}})du}ds\right)dy\bigg]\nonumber\\
    =& \mathbb{E}_x\bigg[\int_0^{\tau_{a_0}^{D^{b^{\ast}}}} 
    e^{-\int_0^y\omega(R_u^{D^{b^{\ast}}})du}ds-
\tau_{a_0}^{D^{b^{\ast}}}e^{-\int_0^{\tau_{a_0}^{D^{b^{\ast}}}}\omega(R_u^{D^{b^{\ast}}})du}\bigg], \quad x\ge a_0,\\
      &\mathbb{E}_x\bigg[\tau^{D^{b^{\ast}}}_{{a_0}}\ P( T^{D^{b^{\ast}}}\ge \tau_{{a_0}}^{D^{b^{\ast}}}|\{R_u^{D^{b^{\ast}}}\}_{u\ge 0})\bigg]=\mathbb{E}_x\bigg[\tau_{{a_0}}^{D^{b^{\ast}}}\ e^{-\int_0^{\tau^D_{{a_0}}}\omega(R_u^{D^{b^{\ast}}})
    du}\bigg], \quad x\ge a_0.\label{alt-p3-ext}
    \end{align}
Substituting these expressions into 
\eqref{alt-p1-ext} yields
$
m(x)
=\mathbb{E}_x\bigg[\int_0^{\tau_{a_0}^{D^{b^{\ast}}}}
    e^{-\int_0^s\omega(R_u^{D^{b^{\ast}}})du}ds\bigg]$ for $x\ge a_0$. \hfill $\square$\\


\noindent
\textbf{Proof of  Theorem \ref{survival-representation}}
The representation follows from the linearity of the differential equations \eqref{eq1-00-ext}--\eqref{eq4-00-ext}. On each continuation region, the expected survival time is given by the sum of a general solution to the associated homogeneous equation and a particular solution.

Suppose first that $b^{\ast}>a_s$. Then
\begin{align*}
m(x)=
\begin{cases}
\tilde C_4 \tilde g_4(x)+ \tilde C_5 \tilde g_5(x)+B_1(x),&a_0\le x\le a_d,\\
\tilde C_6 \tilde g_6(x)+ \tilde C_7 \tilde g_7(x)+B_2(x),&a_d< x\le a_s,\\
\tilde C_8 \tilde g_8(x)+ \tilde C_9 \tilde g_9(x)+B_3(x),& a_s< x\le b^\ast.
\end{cases}
\end{align*}
The constants are determined from the boundary and smooth-fit conditions:
\begin{align}
  &  \tilde C_4 \tilde g_4(a_0)+ \tilde C_5 \tilde g_5(a_0)+B_1(a_0)=0,\label{171123-1-ext-new}\\
   & \tilde C_4 \tilde g_4(a_d)+ \tilde C_5 \tilde g_5(a_d)+B_1(a_d)=
    \tilde C_6 \tilde g_6(a_d)+ \tilde C_7 \tilde g_7(a_d)+B_2(a_d),\label{171123-2-ext-new}\\
 &    \tilde C_4 \tilde g_4^\prime(a_d)+ \tilde C_5 \tilde g_5^\prime(a_d)+B_1^\prime(a_d)=
    \tilde C_6 \tilde g_6^\prime(a_d)+ \tilde C_7 \tilde g_7^\prime(a_d)+B_2^\prime(a_d),\label{171123-3-ext-new}\\
    & \tilde C_6 \tilde g_6(a_s)+ \tilde C_7 \tilde g_7(a_s)+B_2(a_s)=
    \tilde C_8 \tilde g_8(a_s)+ \tilde C_9 \tilde g_9(a_s)+B_3(a_s),\label{171123-4-ext-new}\\
 &    \tilde C_6 \tilde g_6^\prime(a_s)+ \tilde C_7 \tilde g_7^\prime(a_s)+B_2^\prime(a_s)=
    \tilde C_8 \tilde g_8^\prime(a_s)+ \tilde C_9 \tilde g_9^\prime(a_s)+B_3^\prime(a_s),\label{171123-5-ext-new}\\
&\tilde C_8 \tilde g_8^\prime(b^\ast)+ \tilde C_9 \tilde g_9^\prime(b^\ast)+B_3^\prime(b^\ast)=0.\label{171123-6-ext-new}
\end{align}

Using
$
B_1(a_0)=0$, $B_2(a_d)=B_2'(a_d)=0$, $
B_3(a_s)=B_3'(a_s)=0,
$ 
and the initial conditions defining $\tilde g_4-\tilde g_9$, the above system reduces to
\begin{align}
&\tilde C_5=0, 
\quad \tilde C_7=\tilde C_4 \tilde g_4(a_d)+B_1(a_d),
\quad \tilde C_6=\tilde C_4 \tilde g_4'(a_d)+B_1'(a_d), 
\label{171123-2-0-ext-new} \\
&\tilde C_9=\tilde C_6 \tilde g_6(a_s)+ \tilde C_7 \tilde g_7(a_s)+B_2(a_s),
\quad \tilde C_8=\tilde C_6 \tilde g_6'(a_s)+ \tilde C_7 \tilde g_7'(a_s)+B_2'(a_s),\label{171123-5-0-ext-new}
\end{align}
where $\tilde C_4$ is uniquely determined from \eqref{171123-6-ext-new}, as specified in \eqref{171123-6-ext} and 
 $\tilde C_6-\tilde C_9$ subsequently determined by \eqref{171123-2-0-ext-new}--\eqref{171123-5-0-ext-new}, as stated in \eqref{171123-9-ext}-\eqref{171123-10-ext-00}.

If $b^{\ast}=a_s$, then
\begin{align*}
m(x)=
\begin{cases}
\tilde C_{10} \tilde g_4(x)+ \tilde C_{11} \tilde g_5(x)+B_1(x), & a_0\le x\le a_d,\\
\tilde C_{12}\tilde g_6(x)+\tilde C_{13}\tilde g_7(x)+B_2(x), & a_d< x\le a_s.
\end{cases}
\end{align*}
The constants satisfy
\begin{align}
  &  \tilde C_{10} \tilde g_4(a_0)+ \tilde C_{11} \tilde g_5(a_0)+B_1(a_0)=0,\label{171123-100-ext-new}\\
  &\tilde C_{10}\tilde g_4(a_d)+\tilde C_{11}\tilde g_5(a_d)+B_1(a_d)=\tilde C_{12}\tilde g_6(a_d)+\tilde C_{13}\tilde g_7(a_d)+B_2(a_d),\label{171123-200-ext-new}\\
  &\tilde C_{10}\tilde g_4^\prime(a_d)+\tilde C_{11}\tilde g_5^\prime(a_d)+B_1^\prime(a_d)=\tilde C_{12}\tilde g_6^\prime(a_d)+\tilde C_{13}\tilde g_7^\prime(a_d)+B_2^\prime(a_d),\label{171123-300-ext-new}\\
   &\tilde C_{12} \tilde g_6^\prime(a_s)+ \tilde C_{13} \tilde g_7^\prime(a_s)+B_2^\prime(a_s)=0.\label{171123-400-ext-new}
\end{align}
Using again the initial conditions and the properties of $B_1$ and $B_2$, the system reduces to
 $\tilde C_{11} =0,$  
 $ \tilde C_{10}\tilde g_4(a_d)+\tilde C_{11}\tilde g_5(a_d)+B_1(a_d)=\tilde C_{13}$,  $\tilde C_{10}\tilde g_4^\prime(a_d)+\tilde C_{11}\tilde g_5^\prime(a_d)+B_1^\prime(a_d)=\tilde C_{12}$ and 
 $\tilde C_{12} \tilde g_6^\prime(a_s)+ \tilde C_{13} \tilde g_7^\prime(a_s)+B_2^\prime(a_s)= 0,$ 
which uniquely determine $\tilde C_{10}-\tilde C_{13}$, as specified in  
\eqref{171123-600-ext3}. This completes the proof. \hfill $\square$\\

%
%
%

\section{Brownian-bridge absorption correction} \label{app:bridge}

The simulation study of Section~\ref{num:simulation} discretises the controlled
surplus on a uniform time grid $0=t_0<t_1<\cdots$ with $t_{k+1}-t_k=\Delta t$,
using the Euler--Maruyama recursion
\begin{equation}\label{eq:euler}
  \widetilde R_{t_{k+1}}
  = \widetilde R_{t_k}
    + \mu\big(\widetilde R_{t_k}\big)\,\Delta t
    + \sigma\big(\widetilde R_{t_k}\big)\,\big(W_{t_{k+1}}-W_{t_k}\big),
\end{equation}
and declaring liquidation at the hard threshold $a_0$ the first time a grid value
falls to or below $a_0$. Monitoring the boundary only at grid points overlooks
excursions that cross $a_0$ and return within a single step, and therefore
systematically \textit{overstates} the time to liquidation. This bias is not a
second-order nuisance: for a diffusion killed at a boundary, the discretely
monitored Euler scheme converges only at rate $\Delta t^{1/2}$, against the
$\Delta t$ rate enjoyed by the Euler scheme in the absence of killing, and this
slower rate is exact and intrinsic to discrete killing \citep{Gobet2000}.

The standard remedy is to account for within-step crossings analytically rather
than by brute-force grid refinement \citep{Glasserman2004}. Conditional on its
two endpoints $\widetilde R_{t_k}=u$ and $\widetilde R_{t_{k+1}}=v$, and over a
step short enough that the coefficients are effectively constant, the simulated
path is a Brownian bridge with diffusion coefficient $\sigma(u)$---which is
precisely the locally Gaussian, frozen-coefficient approximation that
\eqref{eq:euler} already makes. By the reflection principle, the probability that
such a bridge attains the level $a_0$ at some time in $[t_k,t_{k+1}]$, given
$u>a_0$ and $v>a_0$, is available in closed form,
\begin{equation}\label{eq:bridgeprob}
  p_k \;=\; \exp\!\left(-\,\frac{2\,(u-a_0)\,(v-a_0)}{\sigma(u)^2\,\Delta t}\right).
\end{equation}
The drift does not enter: once the endpoints are fixed, the law of the bridge is
independent of $\mu$. The simulation accordingly augments the absorption test as
follows. At each step the path is liquidated at $a_0$ if $\widetilde
R_{t_{k+1}}\le a_0$ (a directly observed crossing) or, when both endpoints lie
strictly above $a_0$, with probability $p_k$ in \eqref{eq:bridgeprob}---
implemented by drawing an independent uniform $U_k\sim\mathrm{U}(0,1)$ and
absorbing the path when $U_k<p_k$. The same device is used to price discretely
monitored barrier options, where it is known as a continuity correction
\citep{BroadieGlassermanKou1997}.

Two features make the correction attractive here. First, \eqref{eq:bridgeprob} is
\textit{exact} for a Brownian bridge: it introduces no modelling or numerical
approximation beyond the locally Gaussian step that defines the Euler scheme
itself, requires neither calibration nor any auxiliary simulation (one uniform
draw and one exponential evaluation per step), and removes the leading
$\Delta t^{1/2}$ boundary term, restoring the first-order weak-convergence rate.
Second, the correction applies \textit{only} to the hard absorbing threshold
$a_0$. Liquidation within the distress zone is governed by the separate
state-dependent intensity $\omega$, which is simulated exactly through its
integrated-hazard (exponential) clock and is unaffected by the correction; and a
residual $O(\Delta t)$ error remains from the time discretisation of the hazard
integral and of the dividend reflection at the barrier $b^{\ast}$. This is why
refining $\Delta t$ continues to improve accuracy even once the boundary term has
been removed, and why the mean simulated survival in Table~\ref{tab:ex5} still
sits slightly above the analytic value at the step size used.

\end{document}